\documentclass{article}

\usepackage{natbib}
\setcitestyle{authoryear,open={((},close={))}}

\usepackage[english]{babel}
\usepackage[letterpaper,margin=1.00in]{geometry}
\usepackage{amsmath, amssymb, amsthm, thmtools, amsfonts, dsfont}

\usepackage{graphicx}
\usepackage[colorlinks=true, allcolors=blue]{hyperref}
\usepackage[utf8]{inputenc} 
\usepackage[T1]{fontenc}    
\usepackage[colorinlistoftodos]{todonotes}
\renewcommand{\epsilon}{\varepsilon}
\usepackage{amsthm}
\usepackage{amssymb}
\usepackage{amsfonts}
\usepackage{amsmath}
\usepackage{dsfont}
\usepackage{thm-restate}
\usepackage{nicefrac}
\usepackage{tikzsymbols}
\usepackage{tikz}
\usepackage{tikzscale}
\usepackage{bbm}
\usepackage{breqn}
\usetikzlibrary{patterns}
\usepackage{algorithm}
\usepackage{algorithmic}
\usepackage{physics}
\usepackage{booktabs}
\usepackage{multirow}
\usepackage{multicol}

\usepackage{graphicx}
\usepackage{subfigure}

\usepackage{thmtools} 
\usepackage{thm-restate}

\usepackage{enumitem}
\usepackage[capitalize,nameinlink]{cleveref}

\newcommand{\e}{\epsilon}
\newcommand{\expect}{{\bf E}}

\newcommand{\prob}{{\bf Pr}}

\theoremstyle{plain}
\newtheorem{thm}{Theorem} 
\newtheorem*{thm*}{Theorem}
\newtheorem{Def}[thm]{Definition}
\newtheorem{cor}[thm]{Corollary}
\newtheorem{prop}[thm]{Proposition}

\newtheorem{lem}[thm]{Lemma}
\newtheorem{obs}[thm]{Observation}

\newtheorem{claim}[thm]{Claim}

\newtheorem*{cla*}{Claim}

\newtheorem{definition}[thm]{Definition}

\newtheorem{ineq}[thm]{Inequality}

\renewcommand{\geq}{\geqslant}
\renewcommand{\leq}{\leqslant}

\newcommand{\eps}{\varepsilon}

\DeclareMathOperator{\polylog}{polylog}

\newcommand{\R}{\mathcal{R}}

\newcommand{\F}{\mathcal{F}}

\newcommand{\C}{\mathcal{C}}
\newcommand{\I}{\mathcal{I}}

\newcommand{\tc}{t_{c}} 
\newcommand{\N}[2]{N_{G_{#1}}(#2) } 

\newcommand\floor[1]{\lfloor#1\rfloor}

\newcommand{\dataemail}{\textsf{email-Enron}}
\newcommand{\dataastroph}{\textsf{ca-AstroPh}}
\newcommand{\datafacebook}{\textsf{musae-facebook}}
\newcommand{\datahepth}{\textsf{cit-HepTh}}

\newcommand{\algopivot}{\textsc{Pivot-Dynamic}}
\newcommand{\algosingletons}{\textsc{Singletons}}
\newcommand{\algoagreement}{\textsc{Agree-Static}}

\newcommand{\dcc}{\textsc{Dynamic-Agreement}}

\newcommand{\datadrift}{\textsf{Drift}}

\newcommand{\vanillaagreementalgo}[1]{\textsc{AgreementAlgorithm(\ensuremath{#1})}}
\newcommand{\probabilisticagreement}[3]{\textsc{ProbabilisticAgreement(\ensuremath{#1, #2 ,#3})}}
\newcommand{\heavyprocedure}[2]{\textsc{Heavy(\ensuremath{#1, #2})}}

\newcommand{\connectprocedure}[2]{\textsc{Connect(\ensuremath{#1, #2})}}

\author{
    Vincent Cohen-Addad \\
    \small{Google Research}\\
    \small{\texttt{cohenaddad@google.com}}
    \and    Silvio Lattanzi \\
    \small{Google Research}\\
    \small{\texttt{silviol@google.com}}    \and
    Andreas Maggiori \\
    \small{Columbia University}\\
    \small{\texttt{cam6292@columbia.edu}}
    \and   Nikos Parotsidis \\
    \small{Google Research}\\
    \small{\texttt{nikosp@google.com}}
}

\title{Dynamic Correlation Clustering in Sublinear  Update Time}

\begin{document}
\maketitle








\begin{abstract}

We study the classic problem of correlation clustering in dynamic node streams. In this setting, nodes are either added or randomly deleted over time, and each node pair is connected by a positive or negative edge. The objective is to continuously find a partition which minimizes the sum of positive edges crossing clusters and negative edges within clusters. We present an algorithm that maintains an $O(1)$-approximation with $O(\polylog n)$ amortized update time.
Prior to our work~\citet{CharicarSinglePassStreaming} achieved a $5$-approximation with $O(1)$ expected update time in edge streams which translates in node streams to an $O(D)$-update time where $D$ is the maximum possible degree.
Finally we complement our theoretical analysis with experiments on real world data.
\end{abstract}

\section{Introduction} \label{sec:IntroductionDCC}

Clustering is a cornerstone of contemporary machine learning and data analysis.  A successful clustering algorithm partitions data elements so that similar items reside within the same group, while dissimilar items are separated. Introduced in 2004 by Bansal, Blum and Chawla~\cite{bansal2004correlation}, the correlation clustering objective offers a natural approach to model this problem. Due to its concise and elegant formulation, this problem has drawn significant interest from researchers and practitioners, leading to applications across diverse domains. These include ensemble clustering identification \citep{bonchi2013overlapping}, duplicate detection \citep{arasu2009large}, community mining  \citep{chen2012clustering}, disambiguation tasks \citep{kalashnikov2008web}, automated labeling \citep{agrawal2009generating, chakrabarti2008graph}, and many more.


In the correlation clustering problem we are given a graph where each edge has either a positive or negative label, and where a positive edge $(u,v)$ indicates that $u,v$ are similar elements (and a negative edge $(u,v)$ indicates that $u,v$ are dissimilar), the objective is to compute a partition of the graph that 
minimizes the number of negative edges within clusters plus positive 
edges between clusters. Since the problem is NP-hard, researchers
have focused on designing approximation algorithms.


The algorithm proposed by~\citet{clusterlp} achieves an approximation ratio of $1.43 + \epsilon$, improving upon the previous $1.73 + \epsilon$ and $1.994 + \epsilon$ achieved by \citet{Cohen-AddadL0N23,VincentCCBestApprox}. Prior to these developments, the best approximation guarantee of $2.06$ was attained by the algorithm of~\citet{chawla2015near}.\footnote{Note also that there is also a version of the problem~\cite{bansal2004correlation} where the objective is to maximize the number of positive edges whose both endpoints are in the same cluster plus the number of negative edges across clusters. Similarly, if the input is weighted an $O(\log n)$ approximation has
been shown by Demaine et al.~\cite{demaine2006correlation}.}

These above approaches are linear-programming-based: they require to 
solve a linear program and then provide a rounding algorithm. The best 
known ``combinatorial'' algorithm is due to 
a recent local search algorithm of~\citet{abs-2404-05433} 
achieving a $1.845+\eps$-approximation. Prior to this, 
the best known ``combinatorial'' algorithm was the celebrated pivot algorithm 
of~\citet{PivotAlgorithm} which consists in repeatedly
creating a cluster by picking a random unclustered node, and 
clustering it with all its positive unclustered neighbors. The 
versatility of the scheme has led the pivot algorithm to be used
in a variety of contexts, and in particular for designing dynamic 
algorithms~\cite{BehnezhadDHSS19, BehnezhadCMT22, BehnezhadCMT23}.


Dynamic algorithms hold a key position in algorithm design due to their relevance in handling real-world, evolving datasets.
Consequently, substantial research has focused on crafting clustering algorithms expressly designed for dynamic environments (including streaming, 
online, and distributed settings)~\cite{lattanzi2017consistent, fichtenberger2021consistent, jaghargh2019consistent, cohen2019fully, guo2021consistent, OCCC, AssasdiCC, onlineneurips, BehnezhadCMT22, BehnezhadCMT23, bateni2023optimal}.

While the classical approach is to design variations of the Pivot algorithm of~\citet{PivotAlgorithm}, ~\citet{DBLP:conf/icml/Cohen-AddadLMNP21} provide an alternative approach based on a notion called \emph{agreement}, which entails the calculation of the positive neighborhood similarity of pairs of nodes. While that approach is initially used in the context of distributed correlation clustering, it has been also used in~\cite{AssasdiCC} to design a static algorithm with $n \polylog n$ time complexity. In the latter setting the algorithm does not read the entire input, otherwise a time complexity of $n \polylog n$ would be impossible for dense graphs. However the algorithm has access to the graph through queries. We formalize that model in~\cref{def: database model} and ask the following natural question: in which settings can 
correlation clustering be solved in $o(n^2)$ time?

\paragraph{Our Contribution} In this paper, our focus lies on the dynamic case, where nodes are inserted adversarially and / or deleted randomly over time. 
This setting has already been studied for other clustering problems in~\cite{epasto2015efficient} and serves as a bridge between the fully adversarial and random input models.
Our objective is to maintain an $O(1)$-approximate solution at any point in time, while paying as little computation time
as possible upon modification of the input (node insertion
or deletion).
Unfortunately, the approximation algorithms mentioned above are computationally expensive and cannot be re-executed 
each time the input changes.

The best known bound for the fully dynamic
setting is due to~\citet{BehnezhadCMT23} who provided a $5$-approximation in $O(m)$ total update time for adversarial edge insertions and deletions, where $m$ is the total number of positive edges of the graph.
Of course, for node updates, the above approach
can be used to achieve a $5$-approximation with $O(D)$ update 
time, where $D$ is the maximum positive degree of a node throughout the vertex sequence. Note that for dense graphs this is equivalent to a $\Theta (n^2)$ algorithm.
We ask whether it is possible to go beyond this bound if we are given indirect access to the graph through queries and we do not need to read the entire input.

We answer the above question positively. More precisely, we provide the first algorithm which achieves a constant factor approximation model with poly-logarithmic update time per node insertion/deletion .
We also complement our theoretical result with experiments showing the effectiveness of our algorithm in practice.

\section{Problem Definition and the Database Model of Computation}\label{sec: ProblemDefinitionDCC}

The \textit{disagreements minimization} version of the correlation clustering problem receives as input a complete signed undirected graph $G = (V, E, s)$ where each edge $e = \{u,v\}$ is assigned a sign $s(e) \in \{ \text{`}+\text{'}, \text{`}-\text{'}\}$ and the goal is to find a partition of the nodes such that the number of $\text{`}-\text{'}$ edges inside the same cluster and $\text{`}+\text{'}$ edges in between clusters is minimized. For simplicity we denote the set of $\text{`}+\text{'}$ and $\text{`}-\text{'}$ edges by $E^+$ and $E^-$ respectively. A \emph{clustering} is a partition of the nodes $\mathcal{C} = \{C_1, C_2, \dots, C_k \}$ and the cost of that clustering is ${|\{u,v\} \in E^+ : u \in C_i, v \in C_j, i\neq j|} + {|\{u,v\} \in E^- : \exists i:  u, v \in C_i| }  $

Note that a complete undirected signed graph $G= (V, E, s)$ can be converted into a non-signed undirected graph $G=(V, E)$ where for each pair of nodes $\{u,v\}$ there is an edge between them in $G$ if and only if $s \left( \{u,v\} \right) = \text{`}+\text{'}$. Thus the absence of an edge between two nodes corresponds to a negative edge in the original signed graph and the presence of an edge to a positive edge in the original graph. For simplicity, throughout the paper we work with the non-signed equivalent definition of the correlation clustering problem.

The cost of a clustering $\mathcal{C}$ becomes:
\begin{align*}
    |\{u,v\} \in E : u \in C_i, v \in C_j, i\neq j| + \\
|\{u,v\} \not \in E : \exists i:  u, v \in C_i| 
\end{align*}

In our setting, nodes arrival are adversarial and node deletions are random. More precisely, at each time $t$ an adversary can decide either to add to the graph an adversarially chosen node  or to delete a random node from the graph. Upon arrival a node reveals all the edges to previously arrived nodes and upon deletion all edges of the node are deleted. We denote by $u_t$ the node that arrived or left at time $t$ and by $G_t$ the graph structure after the first $t$ node arrivals/deletions. We also denote by $V_t, E_t$ the set of nodes and edges of graph $G_t$ and by $n, m$ the total number of nodes and edges respectively that appeared throughout the dynamic stream. Note that $\forall t$ we have $|V_t| \leq n$ and $|E_t| \leq m$.

We denote by $OPT_t$ the cost of an optimal correlation clustering solution for graph $G_t$ and by $ALG_t$ the cost of a dynamic algorithm solution at time $t$. We say that an algorithm maintains a $c-$approximate solution if $\forall t$, input graphs and node streams we have $ALG_t \leq c \cdot OPT_t$. 

We now formally define the computation model that we study in the paper. For an unsigned graph $G = (V, E)$ we denote by $V$, $E$ the set of nodes and edges respectively, and for a node $u \in V$ we use $N_G (u)$ to denote the neighborhood of $u$ in $G$. This model was considered
by Assadi and Wang~\cite{AssasdiCC} for designing sublinear algorithms
for correlation clustering.

\begin{Def} [Database model~\cite{AssasdiCC}]\label{def: database model}
Given a graph $G = (V, E)$ we have access to the graph structure through the following queries which have a cost of $O (\log |V|)$ :
\begin{enumerate}
    \item Degree queries: $\forall u \in V$ we get its degree $|N_G (u)|$
    \item Edge queries: $\forall u, v \in V$ we get whether $(u, v) \in E$
    \item Neighborhood sample queries: $\forall u$, we get a node  $v \in N_G (u)$ uniformly at random from set $N_G (u)$
\end{enumerate}
\end{Def}
Note that all these queries are easily implementable in the classical computational model where the graph is stored in the same processing unit as the one we use to compute our clustering solution. Thus, other than permitting us to avoid reading and storing the graph locally the Database model is strictly harder than the classical RAM computational model.
Our goal is to maintain a constant approximation with respect to $OPT_t$ using only $\polylog n$ amortized update time (queries and computational operations). Note that for a dense graph this is sublinear in the time of reading the entire input sequence.

\section{Algorithm and Techniques}\label{sec: Algorithm and techniques}

Our approach draws inspiration from the Agreement algorithm, initially presented in~\cite{DBLP:conf/icml/Cohen-AddadLMNP21}.
In particular, we leverage their key insight that to obtain a constant factor approximation it is enough to cluster together nodes with similar neighborhoods. Essentially, it is enough to focus on identifying near-clique structures.
A second key idea that we use comes from~\citet{AssasdiCC} where it is noted that to discover these dense substructures one does not need to examine the entire neighborhood of each node but it is possible to carefully sub-sample the edges of the graph to obtain a sparser structure.
We build upon these two ideas along with developing  several new techniques to obtain a constant factor approximation algorithm for dynamic graph with sublinear complexity.

The section is structured as follows: (1) we introduce the Agreement algorithm of~\cite{DBLP:conf/icml/Cohen-AddadLMNP21} along with its useful properties; (2) we describe the challenges in applying that algorithm on a dynamic graph; (3) we describe a \emph{notification} procedure which is the base of our algorithm; and (4) provide the pseudocode of our algorihtm.

\subsection{The Agreement Algorithm}\label{sec:static-algorithm}

\newcommand{\T}[1]{%
  \ifx\relax#1%
    \mathrm{Type}
  \else
    \mathrm{Type}_{#1}%
  \fi}
  
Before describing the Agreement algorithm of~\cite{DBLP:conf/icml/Cohen-AddadLMNP21} we need to introduce two central notions to quantify the similarity between the neighborhood of two nodes.

\begin{Def} [Agreement]\label{def: epsilon-agreement between nodes}
Two nodes $u,v$ are in $\epsilon$-agreement in $G$ if 
$$|N_{G}(u) \triangle N_{G}(v)|  < \epsilon \max \{ |N_{G}(u)|, |N_{G}(v)| \}$$
where $\triangle$ denotes the symmetric difference of two sets.
\end{Def}

\begin{definition}[Heaviness]\label{def:heaviness}
A node is called \emph{$\e$-heavy} if it is in $\e$-agreement with more than a $(1-\e)$-fraction of its neighbors. Otherwise it is called \emph{$\e$-light}.
\end{definition}

When $\e$ is clear from the context we will simply say that two nodes are or are not in agreement and that a node is heavy or light.

The Agreement algorithm uses the agreement and heaviness definitions to compute a solution to the correlation clustering problem, as described in~\cref{alg:vanilla-agreement}. We call the output of~\cref{alg:vanilla-agreement} the agreement decomposition of the graph $G$.

\begin{algorithm}[h]
\caption{\vanillaagreementalgo{G} \label{alg:vanilla-agreement}}
\begin{algorithmic}
	\STATE {Create a graph $\tilde{G}$ from $G$ by discarding all edges whose endpoints are not in $\epsilon$-agreement.}
	
	\STATE {Discard all edges of $\tilde{G}$ between light nodes of $G$.}
	
	\STATE Compute the connected components of $\tilde{G}$, and output them as the solution.
\end{algorithmic}
\end{algorithm}

At a high level, the first two steps of~\cref{alg:vanilla-agreement} can be characterized as a filtering procedure which ensures that two nodes with similar neighborhoods end up in the same connected component of  $\Tilde{G}$ and consequently in the same cluster of the final partitioning.
The main lemma which helps bounding the approximation ratio of the Agreement algorithm and which will be also used to analyze the performance of our algorithm is the following:
\footnote{We note that~\cref{lem: original paper and how to bound the approximation ratio} is not explicitly stated in~\cite{DBLP:conf/icml/Cohen-AddadLMNP21} but it is a combination of lemmas 3.5, 3.6, 3.7 and 3.8 in the latter paper}

\begin{lem}[rephrased from~\cite{DBLP:conf/icml/Cohen-AddadLMNP21}]\label{lem: original paper and how to bound the approximation ratio}
Let $\mathcal{C} = \{C_1, C_2, \dots, C_k \}$ be a clustering solution for graph $G = (V, E)$ and $\epsilon$ a small enough constant. If the following properties hold:
\begin{enumerate}
    \item $\forall i \in \{ 1, 2, \dots, k\}$ and $u \in C_i$ such that $|C_i|>1$ we have $|N_G (u) \cap C_i| \geq \nicefrac{3}{4} |C_i|$
    \item $\forall e = (u, v) \in E$ such that $u \in C_i$, $v \in C_j$ and $i \neq j$ then either $u$ and $v$ are not in $\epsilon$-agreement or both nodes are not $\epsilon$-heavy.
\end{enumerate}
Then the cost of $\mathcal{C}$ is a constant factor approximation to that of the optimal correlation clustering solution for graph $G$
\end{lem}

For $\epsilon$ small enough~\citet{DBLP:conf/icml/Cohen-AddadLMNP21} prove that~\cref{alg:vanilla-agreement} satisfies both properties of~\cref{lem: original paper and how to bound the approximation ratio} and therefore produces a constant factor approximation.

\subsection{Challenges of Dynamic Agreement}

Our goal is to design a dynamic version of the Agreement algorithm which consistently maintains a sparse graph $\Tilde{G}$ whose induced clustering satisfy both conditions of~\cref{lem: original paper and how to bound the approximation ratio} while only spending $O(\polylog n)$ update time upon node insertions and random deletions.

We briefly describe what are the main challenges that we face in such endeavour:
\begin{enumerate}
    \item computing the agreement between two nodes or the heaviness of a node may take time $\Theta(n)$;
    \item the number of agreement calculations performed by~\cref{alg:vanilla-agreement} is equal to the number of edges in our graph; and
    \item since the total complexity that we aim is $O(n \polylog n)$, the graph $\tilde{G}$ that we maintain should be both sparse at any point and stable (do not change significantly between consecutive times).
\end{enumerate}

From a high level perspective, we solve those issues as follows. First, instead of computing exactly whether two nodes are in $\epsilon$-agreement and whether a node is $\epsilon$-heavy, we design two stochastic procedures, \mbox{namely~\probabilisticagreement{u}{v}{\epsilon} and} \heavyprocedure{u}{\epsilon}, which only need a sample of logarithmic size of the two neighborhoods to answer correctly, with high probability, those questions. We defer the description of those procedures in~\cref{sec: Agreement and heavyness calculation}. Second, for each dense substructure, we maintain dynamically a random set of $O(\polylog n)$ heavy nodes. We call that sample the \emph{anchor} set and show that the connections of those nodes are enough to recover a good clustering. Finally, to efficiently maintain our sparse graph $\Tilde{G}$ we design a message-passing procedure, which we call $\mathrm{Notify}$, to communicate events across neighboring nodes. Roughly, this procedure propagates information about the arrival or deletion of a node $u$ to a $O(\polylog n)$-size randomly chosen subset of nodes within small hop distance from $u$. Whenever a node receives a ``notification'', we either add, with some probability, this node to the \emph{anchor} set and we revisit the agreement between that node and nodes already in the anchor set.

While similar ideas to resolve the first and second challenges have been already explored in the sublinear static algorithm of~\citet{AssasdiCC} it is important to observe that applying the same principles in the dynamic sublinear setting is highly non-trivial. Indeed, in this setting it is not even clear if one can even maintain a good approximation of the degrees of nodes in sublinear time\footnote{While we are able to circumvent the degree computation in our algorithm, we note that this is an interesting open problem.}.


\subsection{Notify Procedure}\label{subsec: notify procedure}

A central sub-procedure in our algorithm, which allow us to keep track of evolving clusters is the Notify procedure.
We believe that it is of independent interest and we devote~\cref{subsec: notify procedure} entirely to its description.
As mentioned before the Notify procedure is responsible to propagate the information of node arrivals and deletions to $O (\polylog n)$ nodes.
We distinguish between different types of notifications depending on how many nodes did the notification propagate through from the ``source'' node that initiated the notify procedure.

Notifications are subdivided in categories depending on their type, which could be $\T{i}, i = 0,1,2$.
A central definition in our algorithm and in its analysis is the ``interesting event'' definition.

\begin{definition}
   We say that $u$ participates in an ``interesting event'' either when $u$ arrives or $u$ receives a $\T{0}$ or $\T{1}$ notification.
\end{definition}

To simplify the description of the algorithm. We denote by $d(u)$ the current degree of a node $u$ and we define the function $l(x) = \floor{\log (x)}, \forall x > 0$. With a slight abuse of notation for a node $u$ we denote by $l_u$ the quantity $l(d(u))$. Note that for all nodes $u$ we have that $2^{l_u} \leq d(u) < 2^{l_u+1} $.

Each node $u$ stores sets $I^0_u, I^1_u,  I^2_u, \dots, I^{\log n}_u$ and $B^0_u, B^1_u,  B^2_u, \dots, B^{\log n}_u$. Set $I^i_u$ contains $u$'s last neighborhood sample when its degree was in $[2^i, 2^{i+1})$ and set $B^i_u$ contains all nodes $v$ such that $u \in I^i_v$. 

Each time a node $u$ participates in an ``interesting event'' or receives a $\T{2}$ notification it stores a random sample of size $O (\log n)$ of its neighborhood in $I^{l_u}_u$. Further, $u$ enters the sets $B^{l_u}_v$ for all nodes $v \in I^{l_u}_u$ of its sample. Those nodes are responsible to notify $u$ when they get deleted. Thus, upon $u$'s deletion all nodes in $\bigcup_{i} B^i_u$ are notified.

This notification strategy has two key properties, that are: (1) $u$ gets notified and participates in an ``interesting event'' when a constant fraction of its neighborhood gets deleted; and (2) w.h.p. $u$ gets notified when its $2$-hop neighborhood changes substantially.
Interestingly, the notification strategy does this  while maintaining the expected complexity bounded by $O(\polylog n)$ each time the \emph{Notify} procedure is called. \cref{algo:notify} contains the pseudocode of the notification procedure

\begin{algorithm}
\label{algo:notify}
\caption{Notify($ u, \epsilon$)}
\label{algo:notify}
\begin{algorithmic}
\IF{$u$ arrived}
    \STATE $I^{l_u}_u \xleftarrow{} $  sample $10^{10} \log n / \epsilon$ random neighbors
    \STATE $\forall v \in I^{l_u}_u$: $B^{l_u}_v \xleftarrow{} B^{l_u}_v \cup \{ u\}$
    \STATE $\forall v \in I^{l_u}_u$: $u$ sends $v$ a $\T{0}$ notification
\ELSIF{$u$ was deleted}
    \STATE $\forall v \in \bigcup_i B^i_u$: $u$ sends $v$ a $\T{0}$ notification
\ENDIF
\FOR{$i = 0, 1, 2$}
    \FORALL{$w$ that received a $\T{i}$ notification}
        \STATE $\forall v \in I^{l_w}_w$: $B^{l_w}_v \xleftarrow{} B^{l_w}_v \setminus \{ w\}$
        \STATE $I^{l_w}_w \xleftarrow{} $  sample $10^{10} \log n / \epsilon$ random neighbors
        \STATE $\forall v \in I^{l_w}_w$: $B^{l_w}_v \xleftarrow{} B^{l_w}_v \cup \{ w\}$
        \IF{$i \in \{0, 1\}$}
            \STATE $\forall v \in I^{l_w}_w$: $w$ sends $v$ a $\T{i+1}$ notification
        \ENDIF
    \ENDFOR
\ENDFOR
\end{algorithmic}
\end{algorithm}

\newcommand{\np}[1]{%
  \if\relax#1%
    \mathrm{Notify}(\cdot)(\epsilon)
  \else
    \mathrm{Notify}({#1})(\epsilon)%
  \fi}

\subsection{Our Dynamic Algorithm Pseudocode}
Our~\cref{alg:dynamicagreement} contains 4 procedures: the $\text{Notify}(v_t, \epsilon)$ procedure which we described in~\cref{subsec: notify procedure}, the Clean($ u, \epsilon, t$) procedure, the Anchor($ u, \epsilon, t$) procedure and the Connect($ u, \epsilon, t$) procedure. Before describing the last three we introduce some auxiliary notation. 

We denote by $\I_t$ the set of nodes that participated in an ``interesting event'' at round $t$.
We also denote by $\Phi$ a dynamically changing (across the execution of our algorithm) subset of the nodes which we call anchor set. We avoid the subscript $t$ in the set $\Phi$ notation as it is always clear for the context at what time we are referring to. When a node $u$ is deleted, then, we also apply $\Phi \gets \Phi \setminus \{u\}$. We maintain a sparse graph $\Tilde{G_t}$ with the same node set of $G_t$ and with edge set $\Tilde{E_t} \subseteq E_t$. We start our algorithm with $\Tilde{G_0}$ being an empty graph. The backbone of our sparse solution $\Tilde{G_t}$ is the anchor set nodes $\Phi$. Indeed, we have that $\forall e = (u, v) \in \Tilde{E_t}$ either $u$ or $v$ are in $\Phi$. Moreover, for any node $u$ we denote by $\Phi_u$ the subset of the nodes in $\Phi$ connected to $u$ in our sparse solution $\Tilde{G_t}$.

\renewcommand{\algorithmicwhile}{\textbf{on arrival/deletion of}}
\renewcommand{\algorithmicendwhile}{\algorithmicend\ \textbf{on}}

\newcommand{\LFORALL}[1]{\State\algorithmicforall\ #1\ \algorithmicdo}
\newcommand{\ELFORALL}{\unskip\ \algorithmicend\ \algorithmicforall}

\begin{algorithm}
\caption{Dynamic Agreement (DA)}\label{alg:dynamicagreement}
\begin{algorithmic}
\WHILE{$v_t$}
    \STATE $\text{Notify}(v_t, \epsilon)$
    \FORALL{$u \in \I_t$}
        \STATE Clean($ u, \epsilon, t$)
        \STATE Anchor($ u, \epsilon, t$)
        \STATE Connect($ u, \epsilon, t$)
    \ENDFOR
\ENDWHILE
\end{algorithmic}
\end{algorithm}

After the arrival or deletion of node $v_t$ and the propagation of notifications through the $\text{Notify}(v_t, \epsilon)$ procedure, for all nodes $u$ in an important event we do the following.

First we call the Clean($ u, \epsilon, t$) procedure which is responsible to delete edges between nodes which are not in $\epsilon$-agreement anymore and delete from the anchor set nodes that lost too many edges in our sparse solution. These operations enable our algorithm to refine clusters which became too sparse or refine cluster assignment for nodes that are not anymore in $\epsilon$-agreement.

Then we call the Anchor($ u, \epsilon, t$) where if $u$ is heavy then with probability $\min \{ 1, \frac{10^{7}\log n}{\epsilon |N_{G_t} (u)|} \}$ we add this node to $\Phi$. If $u$ is added in $\Phi$ then we calculate  agreements with all of its neighbors and whenever $u$ is in agreement with a neighbor $v$ we add edge $(u, v)$ to our sparse solution $\Tilde{G_t}$. As mentioned previously the anchor nodes are representative nodes of clusters. Thus, the Anchor($ u, \epsilon, t$) procedure allows us to update the set of anchor nodes so that they behave approximately like a uniform sample of each cluster.

Finally, we initiate the \textit{Connect} procedure. The purpose of this step is to add some redundant information so that the clustering is stable and also ensure that nodes that are inserted lately are guaranteed to be connected to some anchor node of their cluster.

\begin{algorithm}
\caption{Connect($ u, \epsilon, t$)}\label{alg:connect}
\begin{algorithmic}
\STATE Let $J_u$ be a random sample of size $\nicefrac{10^{5} \log n}{\epsilon}$ from $N_{G_t}(u)$.
\FOR{$w \in J_u$}
    \FOR{$r \in \Phi_w$}
        \STATE If $r$ is heavy and in agreement with $u$ then add the edge $(r, w)$ to $\Tilde{G_t}$
    \ENDFOR
\ENDFOR
\end{algorithmic}
\end{algorithm}

\begin{algorithm}
\caption{Anchor($ u, \epsilon, t$)}\label{alg:anchor}
\begin{algorithmic}
\STATE $X_u \sim \text{Bernoulli}(\min \{ \nicefrac{10^{7} \log n}{\epsilon|N_{G_t}(u)|}, 1\})$
 \IF{$u \in \Phi$}
        \STATE Delete all edges $(u, v)$ where $v \not \in \Phi$ from $\Tilde{G_t}$
    \ENDIF
    \IF{$u$ is heavy in $G_t$ and $X_u = 1$}
        \STATE For every neighbor $v$, if $u$ and $v$ are in agreement add edge $(u, v)$ in our sparse solution $\Tilde{G_t}$.
    \ENDIF
    \IF{$X_u = 1$}
        \STATE Add $u$ to $\Phi$ at the end of iteration $t$
    \ELSIF{$ u \in \Phi$ and $X_u = 0$}
        \STATE Delete $u$ from $\Phi$ at the end of iteration $t$
    \ENDIF
\end{algorithmic}
\end{algorithm}

\begin{algorithm}
\caption{Clean($ u, \epsilon, t$)}\label{alg:clean}
\begin{algorithmic}
\FOR{$w \in \Phi_u$}
    \IF{$w$ not in agreement with $u$ or $w$ is not heavy}
        \STATE Delete edge $(w, u)$ from $\Tilde{G_t}$
    \ENDIF
    \IF{$w$ lost more than an $\epsilon$ fraction of its edges in $\Tilde{G_t}$ from when it entered $\Phi$}
        \STATE Delete $w$ from $\Phi$ and all edges between $w$ and its neighbors in $\Tilde{G_t}$
    \ENDIF
\ENDFOR
\end{algorithmic}
\end{algorithm}



\section{Overview of our Analysis}\label{sec: proof sketch}

In this section we give a high level description of our proof strategy, the full proof is available in the Appendix. 
We first present the high-level ideas on how to prove~\cref{thm: all together theorem} and then we bound the running time.

\begin{restatable}{thm}{maintheorem}\label{thm: all together theorem}
    For each time $t$ the Dynamic Agreement algorithm outputs an $O(1)-$approximate clustering with probability at least $1 - 5/n$.
\end{restatable}

\subsection{Overview of the Correctness Proof}
To prove the correctness of our algorithm we show that at any time $t$: 1)~for every cluster $C$ that is identified by the offline Agreement algorithm (which is known to be constant factor approximate) our algorithm w.h.p. forms a cluster that is a superset of $C$, and 2) every cluster $C'$ detected by our algorithm is a dense cluster w.h.p., meaning that $\forall u \in C'$, $|\N{t}{u} \cap C'| \geq (1 - c' \cdot \epsilon) |C'| $ for small enough $\epsilon$, and a constant $c' \ll 1/\epsilon$.  
Combining the above two facts and union bounding over all time $t$, we get that all clusters detected by the offline agreement algorithm are found, and all detected clusters are very dense. So using the fact that the offline agreement algorithm is a constant approximation algorithm we can show that also our algorithm is.
The formal proof of the two facts requires the introduction of several concepts and probabilistic events, and is deferred to  \cref{sec: Finding dense clusters} and \cref{sec: All found clusters are dense}. Here we give a high-level overview of our proof strategy.
For the remainder of this section, we call a cluster computed by the offline agreement algorithm a good cluster.

\paragraph{All good clusters are detected.}

As discussed in~\cref{sec: Algorithm and techniques} the key idea of the algorithm is to design a sampling strategy and a notify procedure to keep track of the good clusters efficiently. The key idea is to not identify a good cluster only at the time that is formed, but to design a strategy to track the most important events that affect any node in the graph during the execution of the algorithm and to maintain the clustering structure through connections to the anchor set nodes.

Let a cluster $C$ be a good cluster at time $t$. In our analysis, we analyze the last $\frac{\epsilon}{c}|C|$ interesting events involving nodes of cluster $C$. Let $L\subset C$ be the set of nodes involved in these last interesting events, we further subdivide $L$ into $L_1$ containing the half of $L$ that participated earlier in interesting events, and $L_2$ containing those that participated later. We also denote $R=C\setminus L$. Intuitively we show that every node connects to another node in the anchor set either in the first or second part of our analysis.

Let $t_u$ be the last time in which a node $u \in L$  participates in an ``interesting event''. Our notification procedure ensures that w.h.p. $u$'s neighborhood does not change significantly after its last participation in an ``interesting event'', i.e., $N_{G_{t_u}} (u) \simeq N_{G_{t'}} (u) $, $\forall t' \in  [t_u, t]$ (see~\cref{lem: different neighborhood has very low probability} and the preceding discussion). At the same time, we know that at time $t$, $u$ belongs to the good cluster $C$ and by the properties of the agreement decomposition (see~\cref{sec: structural properties of the decomposition}) we have that $N_{G_{t}} (u) \simeq C$. Combining the last two observations we can conclude that $N_{G_{t_u}} (u) \simeq C$. This line of arguments can be extended to all nodes  $v \in N_{G_{t_u}} (u) \cap R$ which by the definition of $R$ do not participate in an ``interesting event'' after time $t_u$. Thus, $\forall v \in N_{G_{t_u}} (u) \cap R$, it holds that: $ N_{G_{t_u}} (v) \simeq C$. In addition, $R$ contains almost all nodes of $C$, thus $ N_{G_{t_u}} (u) \cap R \simeq C $. Combining all these observations we can conclude that $u$ is in agreement with almost all of its neighbors at time $t_u$, and therefore it is heavy.
In addition, if $u$ enters in the anchor set, it remains there until time $t$. This is proved in~\cref{thm: ui are heavy} of~\cref{sec: Finding dense clusters}.

We are ready to prove that our algorithm finds a cluster $C' \supseteq C$ at time $t$ for every good cluster $C$ detected by the offline agreement algorithm.
As described in~\cref{sec: Algorithm and techniques}, the clusters that our algorithms form are determined by the connected components in our sparse solution graph $\tilde{G}_t$. To this end we argue that each node of $C$ is connected to a node in $C$ that is in the anchor set in $\tilde{G}_t$. We show how this is true for each of the three sets $R, L_1, L_2$.
For $R$ we recall that each node in $L_1\subset L$ is heavy and enters the anchor set with probability $O(\frac{\log n}{\epsilon |N_{G_{t_u}}(u)|} )$ at time $t_u$.

Given that $|L_1|=\frac{c}{2\epsilon}|C|$ for $c \ll \epsilon$, we can show that each $v\in R$ has a neighbor $v' \in L_1$ that enters the anchor set w.h.p., and $v'$ connects to $v$ in $\tilde{G}_t$ during the $Anchor(v', \epsilon, t_{v'})$ procedure. This is formally proved in \cref{lem: all nodes in R get selected by some node in A1} of~\cref{sec: Finding dense clusters}.
Similarly, each node in $L_1$ has a neighbor in $L_2$ that enters the anchor set w.h.p..
Finally, we prove that each node $v$ in $L_2$ is connected to a node $v'$ in the anchor set in $L_1$ w.h.p.. In fact, since most pairs of nodes in $L_1\cup L_2$ are in agreement w.h.p. there are many common neighbors $w \in R$ such that $w$ is in agreement with both $v'$ and $v$, which implies that the $Connect(v, \epsilon, t_v)$ will connect $v$ to $v'$. See~\cref{lem: all nodes in L2 get connected trough the Connect procedure} of~\cref{sec: Finding dense clusters} for the formal proof. Hence, all nodes in $C$ get connected to a node in the anchor set (which as we claimed above remains in the anchor set until time $t$). To conclude the argument we also note that the nodes in the anchor set are connected to each other because they share most of their neighbors. This is proved in~\cref{thm: all nodes of C are clustered together} of~\cref{sec: Finding dense clusters}.

\paragraph{All found clusters are good.} To prove that all clusters identified by our algorithm are good clusters, we follow a proof strategy similarly to~\cite{DBLP:conf/icml/Cohen-AddadLMNP21}. Roughly speaking, we show that each connected component $C'$ of $\tilde{G}_t$ has diameter $4$, which follows by observing that all nodes in the anchor set are within distance $2$ from each other, and that each other node is connected to a node in the anchor set. 
Then, due to the transitivity of the agreement property, it follows that all nodes in the connected component $C'$ are in agreement with each other.
This last claim, then further implies that for each $v\in C'$ it holds that $|N_{G_{t}} \cap C'| \geq (1-c \epsilon) |C'|$
which makes $C'$ a good cluster. See~\cref{sec: All found clusters are dense} for the formal proof.







\subsection{Overview of Running Time Bound}
We observe that our algorithm is correct even when deletions occur adversarially, but this unfortunately does not hold for the analysis of its running time. Notice that for insertions the running time is bounded by the $O(\polylog n)$ forward notifications sent as a result of the insertion of a new node. On the other hand, the deletion of a node $u$ may cause all nodes in $\bigcup_{i} B^i_u$ to receive a notification, and this in turn causes those nodes to send forward notifications. The issue arises in that there is no bound in the size of $|\bigcup_{i} B^i_u|$ when deletions occur adversarially. In fact, we can construct an instance where $|\bigcup_{i} B^i_u| \in \widetilde{\Theta}(n)$ for $\widetilde{O}(n)$ many deletions. Thankfully, in the case where the deletions appear in a random order, this cannot happen as we can nicely bound the expectation of $|\bigcup_{i} B^i_u|$ to be $O(\polylog n)$ for a node $u$ chosen uniformly at random. We devote \cref{sec: Runtime analysis} in the formal proof of the running time analysis.
\section{Experimental Evaluation}\label{sec: experimental evaluation}

We conduct two sets of experiments. We first evaluate the performance of our algorithm to the same set of real-world graphs that were used in~\cite{onlineneurips}. Then, we investigate how the running time of our algorithm scales with the size of the input.


\subsection{Baselines and Datasets}

We compare our algorithm to \algosingletons{} where its output always consists of only singleton clusters and~\algopivot{}~\cite{BehnezhadCMT23} which is a dynamic variation of the Pivot algorithm~\cite{PivotAlgorithm} for edge streams with $O(1)$ update time. While~\algopivot{} guarantees a constant factor approximation, \algosingletons{} does not have any theoretical guarantees. Nevertheless, it has been observed in~\cite{onlineneurips}  that sparse graphs tend to not have a good correlation clustering structure, and often the clustering that consists of only singleton clusters is a competitive solution.

We consider five real-world datasets. 
For the first set of experiments, we use four graphs from SNAP~\cite{snapnets} that include a Social network (\datafacebook{}), an email network (\dataemail{}), a collaboration network (\dataastroph{}), and a paper citation network (\datahepth). 

In the second set of experiments where we investigate the runtime of our algorithm with respect to the size of the input we use \datadrift{}~\cite{vergara2012chemical, rodriguez2014calibration} from the UCI Machine Learning Repository~\cite{Dua:2019}. The dataset contains $13,910$ points embedded in a space of $129$ dimensions. Each point  corresponds to a node in our graph and we add a positive connection between two nodes if the euclidean distance of the corresponding points is below a certain threshold. The lower we set that threshold the denser the graph becomes.

All graphs are formatted so as to be undirected and without parallel edges.
In addition, we create the node streams of node additions and deletions as follows: we first create a random arrival sequence for all the nodes. Subsequently in between any two additions, with probability $p$, we select at random a node of the current graph and delete it. If all nodes have already arrived then at each time we randomly select one of those and delete it. 

\subsection{Setup and Experimental Details}
Our code is written in Python 3.11.5 and is available at \url{https://github.com/andreasr27/DCC}. We set the deletion probability in between any two node arrivals to be $0.2$. The agreement parameter is set to $\e =  0.2$, as this setting exhibited the best behavior in~\cite{DBLP:conf/icml/Cohen-AddadLMNP21} and~\cite{onlineneurips}. 

In addition, we set the number of samples in our procedures to a small constant. More precisely, all our subroutines use a random sample of size $2$ and the probability of a node joining the anchor set is set to $20/d_u$ where $d_u$ denotes its degree in the current graph. Here we deviate from the numbers we use in theory as we observe that, in practice, for sparse graphs only running time is affected. We note that in the runtime calculation we do not include reading the input and calculating the quality of our clustering. We do this in an effort to best approximate the Database model in~\cref{def: database model} while reading the input we implement suitable data structures\footnote{e.g.: \url{https://leetcode.com/problems/insert-delete-getrandom-o1/description/}} where the graph is stored in the form of adjacency lists which permit node additions, deletions and getting a random sample in $O(1)$ expected time.

\begin{figure}[ht]
  \centering
  \includegraphics[width=0.49\textwidth]{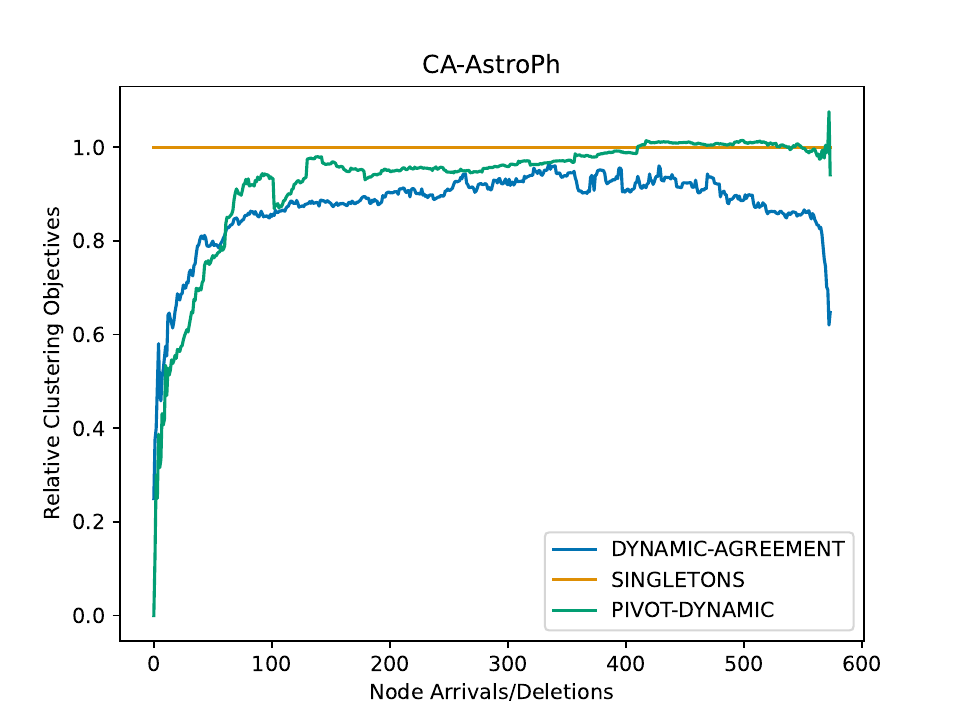} 
  \caption{Correlation clustering objective relative to singletons}
  \label{fig: Ca-AstroPh}
\end{figure}

\paragraph{Solution quality.}
In the first set of experiments we run all three algorithms and plot their performance relative to \algosingletons{}, that is, we plot the cost of solution produced by our algorithm or \algopivot{} divided by the cost of the solution of \algosingletons{}. In all datasets, \dcc{} consistently outperforms both \algopivot{} and \algosingletons{}. For example, in~\cref{fig: Ca-AstroPh} we plot the correlation clustering objective every $10$ nodes additions/deletions in the node stream. We observe that after a constant fraction of all nodes have arrived, the clustering objective of our algorithm relative to \algosingletons{} remains stable both for node additions and deletions. On the contrary the performance of \algopivot{} fluctuates and tends to increase especially in the last part of the sequence when all nodes have already arrived and the node stream contains predominantly deletions. A similar behaviour is observed in the other three datasets which are deferred to~\cref{sec: AdditionalExperiments}. In~\cref{sec: AdditionalExperiments} we also present a table which gives an estimate of the optimum offline solution based on the classical \textsc{Pivot} algorithm of~\citet{PivotAlgorithm}. In~\cref{tab:runtimes} we see that our algorithm is slower to the \algopivot{} implementation. This is something that we expect since \algopivot{} is extremely efficient for sparse graphs.

\begin{table}[t]
\label{tab:runtimes}
\begin{center}
\begin{small}
\begin{sc}
\begin{tabular}{lcccr}
\toprule
& {DA} & {PD}\\
\midrule
\datafacebook{} & 2.27 & 0.1 \\
\dataemail{} & 2.79 & 0.12 \\
\datahepth{} & 3.84 & 0.21 \\
\dataastroph{} & 2.96 & 0.11 \\
\bottomrule
\end{tabular}
\caption{Runtimes of the algorithms \dcc{} (DA) and \algopivot{} (PD).}
\end{sc}
\end{small}
\end{center}
\end{table}

\begin{table}[htb]
  \centering
  \label{tab:results}
  \begin{tabular}{lcccr}
    \toprule
    & \multicolumn{2}{c}{Relative Objective} & \multicolumn{2}{c}{Running Time} \\
    \cmidrule(lr){2-3} \cmidrule(lr){4-5}
    Density & DA & PD & DA & PD \\
    \midrule
    253.36 & 0.696 & 0.610 & 12.48 & 0.28 \\
    114.87 & 0.605 & 0.551 & 14.54 & 0.16 \\
    69.74 & 0.522 & 0.504 & 15.82 & 0.11 \\
    52.17 & 0.389 & 0.416 & 14.10 & 0.09 \\
    42.25 & 0.320 & 0.375 & 12.85 & 0.07 \\
    \bottomrule
  \end{tabular}
  \caption{Summary of our scalability experiment for the algorithms \dcc{} (DA) and \algopivot{} (PD), on the  \datadrift{} dataset at varying densities.}
\end{table}

\paragraph{Running time.}

As mentioned previously, we constructed a graph using the  \datadrift{} dataset by associating points with nodes and adding positive edges between nodes if the Euclidean distance of the corresponding points is less than a threshold. Different thresholds lead to the creation of distinct graphs. Now we relate the density of the graph (average node degree) with the runtimes and clustering quality of both \dcc{} and \algopivot{}. \cref{tab:results} shows the average relative clustering quality of each algorithm over the entire node stream and their running times.

We observe that both \dcc{} and \algopivot{} outperform \algosingletons{}, as expected, since \algosingletons{} excels in sparse graphs and offers poor quality solutions in denser graphs. Additionally, the running time of \dcc{} remains stable when the density of the graphs increases confirming that the algorithm's running time is not affected by the graph density. On the other hand, the running time of \algopivot{} increases linearly with the density. This suggests that for very large and dense graphs, where even reading the entire input is prohibitive \dcc{} scales smoothly, further validating the theory.

\paragraph{Experimental summary.} We observe that the newly proposed algorithm computes high quality solutions for both sparse and dense graphs. This is in contrast with comparison methods that fail to produce good solutions in at least one of the two settings. We also show that even if the runtime of our algorithm is higher compared to the competitor, its runtime does not increase with the density of the input graph. This is inline with our theoretical results and is achieved via our sampling and notify strategies.

\section*{Conclusion and Future Work}
We provide the first sublinear time algorithm for correlation clustering in dynamic streams where nodes are added adversarially and deleted randomly. Our algorithm is based on new and carefully defined sampling and notification strategies that can be of independent interest. We also show experimentally that our algorithm provides high quality solution both for dense and sparse graph outperforming previously known algorithms.

A very interesting open question is to extend our result in the general setting where both nodes' addition and deletion is adversarial. One possible way of achieving the later would be to use a different sparse-dense decomposition which is more stable and requires less updates so that it is maintained approximately, e.g., the one proposed by~\citet{AssasdiCC}.

Reducing the amortize update time in our setting is another very interesting and natural question.


\bibliography{main}

\newpage
\appendix
\onecolumn

\section{Finding Dense Clusters}\label{sec: Finding dense clusters}

In this section we prove that our algorithm correctly finds the non-singleton clusters of $G_t$ that are also found by the \vanillaagreementalgo{G_t} when the agreement and heaviness parameters are set to a small enough value. That is, let $C$ be a non-singleton cluster that is found by \vanillaagreementalgo{G_{t_{current}}} at time $t_{current}$ when the agreement and heaviness parameters are set to $\nicefrac{\epsilon}{10^{14}}$. It can be proven that cluster $C$ is extremely dense, i.e., it forms almost a clique, with very few outgoing edges.
In~\cref{thm: all nodes of C are clustered together} we prove that at time $t_{current}$ our algorithm outputs a cluster $C'$ that contains all nodes of $C$, i.e., $C'\supseteq C$.
While the latter may not seem surprising per se (note that a trivial algorithm which clusters all nodes in the same partition also achieves that property), our approach consists in a delicate argument where  we prove that for all non-trivial clusters $C$ found by \vanillaagreementalgo{G_{t_{current}}} our algorithm always constructs a cluster $C' \supseteq C$ and at the same time in~\cref{sec: All found clusters are dense} all clusters $C'$ constructed by our algorithm (which may not be constructed by \vanillaagreementalgo{G_{t_{current}}}) are very dense.
In order to simplify the notation we use $\tc$ instead of $t_{current}$ in our formulas.

The first challenge in the current section is to prove that the notify procedure correctly samples enough nodes of $C$, ``close enough'' to time $t_{current}$, so that many of these nodes enter the anchor set and help us reconstruct $C$.
Note that the latter ensures that enough nodes from a specific cluster enter the anchor set, but does not ensure that these clusters are correctly found. Indeed, the last arriving node may not enter the anchor set and at the same time the ``interesting events'' that its arrival generates may not induce any node to join the anchor set. However, to get a constant factor approximation, we still need to ensure that the last node of $C$ is clustered correctly with the rest of the cluster.
In order to ensure the latter, our algorithm uses the \connectprocedure{v}{\epsilon} procedure. \connectprocedure{v}{\epsilon} constructs a sample of $v$'s neighborhood and for every node $w$ in that sample, it checks whether $v$ is in $\epsilon$-agreement with any node $r \in \Phi_w$, i.e., any node that is in the anchor set of and if so $v$ is connected to $r$.

In the current section we make great use of the relation between $C$ and $\N{\tc}{u}$ for a node $u \in C$, indeed since $C$ is a dense cluster found by \vanillaagreementalgo{G_{t_{current}}} at time $t_{current}$, informally we have that $\N{\tc}{u} \approx C$ (see the properties of the Agreement decomposition in~\cref{sec: structural properties of the decomposition}).

Consider the last $\lfloor \nicefrac{\epsilon}{10^{4}} |C|\rfloor$ nodes of $C$ that participate in an ``interesting event'' and denote those nodes by $u_i$ for $i = 1, 2, \dots, \lfloor\nicefrac{\epsilon}{10^{4}} |C|\rfloor$. For each $u_i$ denote by $t_i$ the last time that this node participated in an ``interesting event'', so that $t_1 \leq t_2 \leq \dots \leq t_{\lfloor\nicefrac{\epsilon}{10^{4}} |C|\rfloor}$. If two nodes $u_i, u_j$ participated in an ``interesting event'' after the same arrival or deletion in the node streams then we order with respect to the type of notification received, that is, $t_i \leq t_j$ if $u_i$ just arrived or received a notification of a lower type than $u_j$. Ties are broken arbitrarily but consistently. Note that both nodes $u_i$ and times $t_i$ for $i = 1, 2, \dots, \lfloor\nicefrac{\epsilon}{10^{4}} |C|\rfloor$ are random variables which depend on the internal randomness of the Notify procedure.

We now introduce some auxiliary notation:
\begin{enumerate}
    \item $L = \{u_1, u_2, \dots, u_{\lfloor\nicefrac{\epsilon}{10^{4}} |C|\rfloor }\}$ and $R = C \setminus L$.
    \item $L_1 = \{u_1, u_2, \dots, u_{\lfloor\nicefrac{\epsilon}{2\cdot 10^{4}} |C|\rfloor}\}$.
    \item $L_2 = \{u_{\lfloor\nicefrac{\epsilon}{2\cdot 10^{4}} |C|\rfloor+1}, u_{\lfloor\nicefrac{\epsilon}{2\cdot 10^{4}} |C|\rfloor+2}, \dots, u_{\lfloor\nicefrac{\epsilon}{10^{4}} |C|\rfloor}\}$.
    \item $L^i = \{u_1, u_2, \dots, u_i\}$.
    \item We denote as $t'$ the earliest time that all nodes of $R$ have arrived.
\end{enumerate}
Again, note that since the $u_i$'s and $t_i$'s are random variables then also $R, L, L_1, L_2, L^i$ and $t'$ are random variables.

We start by some simple observations where we argue that after time $t'$ all nodes of $C$ that have already arrived form a dense subgraph. 
\begin{obs}\label{obs: bounds on degree of u at time tc wrt C}
$\forall u \in C$ we have:
$$
(1 - \epsilon/10^{13}) |C| \leq |\N{\tc}{u}| \leq (1 + \epsilon/10^{13}) |C|
$$
\end{obs}
\begin{proof}
   From~\ref{property: N_G(u)  >= (1 - 9 epsilon) C} and~\ref{property: N_G(u) setminus C < 3 epsilon N_G(u) < (3 epsilon)/(1 - 3 epsilon) C} in~\cref{sec: structural properties of the decomposition} the result is immediate.
\end{proof}

\begin{obs}\label{obs: at tc u is connected to almost everybody in C}
$\forall u \in C$ we have:
$$
|\N{\tc}{u} \cap C| \geq (1 - \epsilon/10^{13}) |C|
$$
\end{obs}
\begin{proof}
Follows immediately from ~\ref{property: N_G(u) cap C  >= (1 - 9 epsilon) C}.
\end{proof}

\begin{obs}\label{obs: at time t the v is connected with most nodes in C}
For each $u\in C$ that arrived at time $t \in [t',\tc]$ we have:
$$
|\N{t}{u} | \geq |\N{t}{u} \cap C | \geq (1 - 2 \epsilon/ 10^{4})|C|
$$
\end{obs}
\begin{proof}
   The left-hand side inequality is obvious. For the right inequality note that at time $t'$ all nodes in $R$ have arrived. Thus, 
   \begin{align*} 
   |\N{t}{u} \cap C| &\geq |\N{\tc}{u} \cap C| - |L| \geq (1 - \epsilon/10^{13}) |C| - \epsilon/10^{4} |C| \\
                     &\geq (1 - 2 \epsilon/10^{4}) |C|
    \end{align*}
    Where in the second inequality we use~\cref{obs: at tc u is connected to almost everybody in C} and $|L|  = \lfloor\epsilon/10^{4} |C|\rfloor$.
\end{proof}

\begin{obs}\label{obs: C minus Ni is upper bounded}
For each $u\in C$ that has already arrived before time $t \in [t',\tc]$ we have:
$$
|C \setminus \N{t}{u} | \leq 2 \epsilon/  10^{4}|C| \leq  \epsilon/  10^{3}|\N{t}{u}|
$$
\end{obs}
\begin{proof}
  The left inequality is immediate from~\cref{obs: at time t the v is connected with most nodes in C} and the second one holds for $\epsilon$ small enough again using \cref{obs: at time t the v is connected with most nodes in C}.
\end{proof}

\cref{obs: C minus Ni is upper bounded} informally states that the neighborhood of any node $v \in C $ between times $t'$ and $ \tc$ contains $C$ almost in its entirety. Let $t \in [t', \tc]$ be a time when node $v$ has already arrived, then we have that: $\N{t}{v}  \supseteq C'$ and $C' \simeq C$. Ideally, we would like all nodes $v \in C $ to have almost their final neighborhood between times $t'$ and $ \tc$ so that the Anchor procedure correctly reconstructs the cluster $C$. That is, we would like to also have $\N{t}{v}  \subseteq C''$ and $C'' \simeq C$. While the latter may be not true in general, we prove that with high probability at every time $t_i$ it is true for $u_i$ and a large part of its neighborhood.

We introduce auxiliary notation to formalize these claims.

\begin{definition}\label{def: Mv and Av}
For a node $v \in C$ and time $t$ we define the following events:
\begin{align*}
    M_v^{t} &= \{|\N{t}{v} \setminus \N{\tc}{v}| < \nicefrac{\epsilon}{8 \cdot 10^{4}}|\N{t}{v}| \}\\
    A_v^{t} &= \left\{ \left\{ v \in \I_{t}\right\} \lor \left\{ v \text{ received a } Type_2 \text{ notification at time } t\right\} \right\}
\end{align*}
\end{definition}
$M_v^{t}$ is true if at least a $(1 - \nicefrac{\epsilon}{8 \cdot 10^{4}})$ fraction of $v$'s neighborhood at time $t$ does not get deleted until time $\tc$. Also, note that if $A_v^{t}$ is true then $v$ samples with replacement $10^{10} \log n / \epsilon$ neighbors at time $t$. By $\overline{M_v^{t}}$ and $\overline{A_v^{t}}$ we denote the complementary events of ${M_v^{t}}$ and $ {A_v^{t}}$ respectively.

The rest of the section is devoted in arguing that with high probability at each time $t_i$, both $u_i$ and nearly all its neighbors possess almost their ``final'' neighborhood.
The crux of our analysis is based on the following two observations (stated informally for the moment). For all $v \in R\cup L_i$: (1) If $M_v^{t_i}$ is true, then $\N{t_i}{v} \simeq C$; and (2) The event $\overline{M_v^{t_i}} \land A_v^{t_i}$ happens with low probability.

In the following lemma we prove structural properties of the neighborhood $\N{t_i}{v}$ of a node when $v$ when $M_v^{t_i}$ is true.

\begin{lem}\label{lem: high degree implies very different neighborhood}
    $\forall v \in R \cup L^i$ we have that if $M_v^{t_i}$ is true then:
    \begin{align}
        & |\N{t_i}{v}| < (1 + \nicefrac{\epsilon}{2 \cdot 10^{4}}) |C|\\
        &|\N{t_i}{v} \setminus C| < \nicefrac{\epsilon}{10^{4}} |C| < \nicefrac{\epsilon}{ 10^{3}}|\N{t_i}{v}|
    \end{align}
\end{lem}
\begin{proof}
\begin{enumerate}
\item For the first statement we prove the contrapositive, i.e., we argue that $|\N{t_i}{v}| \geq (1 + \nicefrac{\epsilon}{2 \cdot 10^{4}}) |C|$ implies  $\overline{M_v^{t_i}}$. We distinguish between the cases:
    \begin{itemize}
        \item $|\N{t_i}{v}| \geq \frac{|\N{\tc}{v}|}{(1 - \epsilon/(8 \cdot 10^{4}))}$, then
        \begin{align*}
            |\N{t_i}{v} \setminus \N{\tc}{v}| &\geq |\N{t_i}{v}| - |\N{\tc}{v}| \\
            &\geq \epsilon/(8 \cdot 10^{4})|\N{t_i}{v}|
        \end{align*}
        \item $|\N{t_i}{v}| < \frac{|\N{\tc}{v}|}{(1 - \epsilon/(8 \cdot 10^{4}))}$, in that case from~\cref{obs: bounds on degree of u at time tc wrt C} we have the bound $|\N{\tc}{v}| \leq (1 + \epsilon/10^{13}) |C|$
        \begin{align*}
            &|\N{t_i}{v} \setminus \N{\tc}{v}|\\
            &\geq |\N{t_i}{v}| - |\N{\tc}{v}| \\
            &\geq ( 1 + \epsilon/(2 \cdot 10^{4}) ) |C| - |\N{\tc}{v}| \\
            &\geq \left( \frac{1 + \epsilon/(2 \cdot 10^{4})}{1 + \epsilon/10^{13}} -1 \right) |\N{\tc}{v}| \\
            &\geq \left( \frac{1 + \epsilon/(2 \cdot 10^{4})}{1 + \epsilon/10^{13}} -1 \right) \left( 1 - \epsilon/(8 \cdot 10^{4}) \right) |\N{t_i}{v}| \\
            &\geq \left( \epsilon/(2 \cdot 10^{4}) - \epsilon/10^{13} \right) \left( \frac{1 - \epsilon/(8 \cdot 10^{4})}{1 + \epsilon/10^{13}} \right)  |\N{t_i}{v}| \\
            &\geq  \epsilon/(8 \cdot 10^{4})  |\N{t_i}{v}| \\
        \end{align*}
        Where in the second inequality we used our assumption that $|\N{t_i}{v}| \geq (1 + \nicefrac{\epsilon}{2 \cdot 10^{4}}) |C|$, the third one from~\cref{obs: bounds on degree of u at time tc wrt C}, the fourth one from the fact that we are in the second case and the last one holds for $\epsilon$ small enough.
    \end{itemize}
\item From the first part of the current lemma we have that $|\N{t_i}{v}| < (1 + \nicefrac{\epsilon}{2 \cdot 10^{4}}) |C|$ and from~\cref{obs: at time t the v is connected with most nodes in C} we have that $|\N{t_i}{v} \cap C| \geq (1- 2 \epsilon/10^{4}) |C|$, consequently:
\begin{align*}
    |\N{t_i}{v} \setminus C| &= |\N{t_i}{v}|-|\N{t_i}{v} \cap C|\\
                             &< (1 + \epsilon/(2\cdot 10^{4})) |C| - (1 - 2 \epsilon/ 10^{4}) |C|\\
                             &< 3 \epsilon/10^{4} |C|\\
                             &< 3 \epsilon/10^{4} \frac{|\N{t_i}{v}|}{(1 - 2 \epsilon/ 10^{4})}\\
                             &< \epsilon/10^{3} |\N{t_i}{v}|
\end{align*}
Where the second inequality holds for $\epsilon$ small enough and proves the first inequality of (2) in the current lemma, the third inequality uses again~\cref{obs: at time t the v is connected with most nodes in C} and the last inequality holds for $\epsilon$ small enough.
\end{enumerate}

\end{proof}

By Combining~\cref{lem: high degree implies very different neighborhood} and~\cref{obs: C minus Ni is upper bounded} we can argue that $\forall v \in R \cup L^i$ if $M_v^{t_i}$ is true then $\N{t_i}{v} \simeq C$. Consequently for two neighboring nodes $ v, u \in R \cup L^i$: $M_v^{t_i} \land M_u^{t_i}$ implies that $\N{t_i}{v} \simeq \N{t_i}{u} \simeq C$ and consequently $u$ and $v$ are in agreement. We formalize the latter in the following lemma.

\begin{lem}\label{lem: Oui and Ow suffice for ui and w being in agreement}
For all neighboring nodes $ v, u \in R \cup L^i$: if $M_v^{t_i} \land M_u^{t_i}$ is true then $v$ and $u$ are in $\epsilon/10$-agreement.
\end{lem}
\begin{proof}
Since both $M_u^{t_i}$ and $M_{v}^{t_i}$ are true, from~\cref{lem: high degree implies very different neighborhood} we have that:
\begin{align*}
    & |\N{t_i}{u} \setminus C| < \epsilon/ 10^{3}|\N{t_i}{u}|
\end{align*}
and
\begin{align*}
    & |\N{t_i}{v} \setminus C|  < \epsilon / 10^{3} |\N{t_i}{v}|
\end{align*}
Also from~\cref{obs: C minus Ni is upper bounded} it holds that:
$$
|C \setminus \N{t_i}{u} |  \leq  \epsilon/  10^{3}|\N{t_i}{u}|
$$
and
$$
|C \setminus \N{t_i}{v} |  \leq  \epsilon/  10^{3}|\N{t_i}{v}|
$$
To argue that $| \N{t_i}{u} \triangle \N{t_i}{v} | \leq \epsilon/10 \max \{ |\N{t_i}{u}|, \N{t_i}{v}| \}$, we bound both $|\N{t_i}{u} \setminus \N{t_i}{v}|$ and $|\N{t_i}{v} \setminus\N{t_i}{u} |$ by $\epsilon/10^{2} \max \{ |\N{t_i}{u_i}|, \N{t_i}{w}| \}$. From the latter the lemma follows since $2\epsilon/10^{2} < \epsilon/10 $ for $\epsilon$ small enough. We upper bound the size of $\N{t_i}{u}\setminus\N{t_i}{v}$:
\begin{align*}
|\N{t_i}{u}\setminus\N{t_i}{v}| &\leq |\N{t_i}{u}\setminus C| + | C \setminus\N{t_i}{v}|\\
                                  &\leq \epsilon /10^{3} |\N{t_i}{u}| + \epsilon /10^{3} | \N{t_i}{v}|\\
                                  &\leq \epsilon/10^{2} \max \{ |\N{t_i}{u}|, \N{t_i}{v}| \}
\end{align*}
and the upper bound of $|\N{t_i}{v}\setminus\N{t_i}{u}|$ follows the same line of arguments.
\end{proof}

We know prove that the probability of a node $v \in R \cup L_i$ both sampling its neighborhood at time $t_i$ and having a neighborhood which is very different than $C$ is very low. The following lemma is crucial, as it subsequently help us argue that most of the neighboring nodes of $u_i$ at time $t_i$ have almost their final neighborhoods.

\begin{lem}\label{lem: different neighborhood has very low probability}
    $\forall v \in C$ and time $t$ we have that:
    $$
    \prob[\overline{M_v^{t}} \land A_v^{t} \land \{ v \in R \cup L^i\} \land \{t \geq t_i \}] < 1/n^{10^{4}}
    $$
\end{lem}
\begin{proof}
Let $d_{t'}$ denote the degree of node $v$ at time $t'$ and $l_{t'} = l(d_{t'})$. In this proof we write $I_{v, t'}$ to denote the neighborhood sample stored in $I^{l_{t}}_v$ at time $t'$. Note that $I_{v, t'}$ and $I_{v, t''}$ for $t' < t''$ may be different if for example $v$ received a $\T{2}$ notification between those times and updated its neighborhood sample.

Also, for all $w \in \N{t}{v} \setminus \N{\tc}{v} $ let $t_w$ be the time when $w$ gets deleted from the node stream and let $w_j$ be the $j$-th element to be deleted in $\N{t}{v} \setminus \N{\tc}{v} $.
For $ j = 1, 2, \dots, |\N{t}{v} \setminus \N{\tc}{v}|$, we define event $E_j = \left\{ w_j \not \in I_{v, t_w} \right\}$.

We argue that $\{ v \in R \cup L^i\} \land \{t \geq t_i \}$ implies 
$\bigcap\limits_{1 \leq j \leq |\N{t}{v} \setminus \N{\tc}{v}|} E_j$. Indeed, $\overline{E_j}$ implies that $v \in B^{l_{t}}_w$ at time $t_w$. Consequently, upon its deletion, $w$ would send a $\T{0}$ notification to $v$ and $v$ would participate in an ``interesting event'' between times $[t_i +1, \tc)$. The latter implies that $v \not \in R \cup L_i$.

Thus: 

\begin{align*}
\prob[\overline{M_v^{t}} \land A_v^{t} \land \{ v \in R \cup L^i\} \land \{t \geq t_i \}] &\leq \prob[\overline{M_v^{t}} \land A_v^{t} \land \bigcap\limits_{1 \leq j \leq |\N{t}{v} \setminus \N{\tc}{v}|} E_j]\\
&\leq \prob[ A_v^{t} \land \bigcap\limits_{1 \leq j \leq \nicefrac{\epsilon}{8 \cdot 10^{4}} d_t } E_j]\\
&= \prod_{1 \leq j \leq \nicefrac{\epsilon}{8 \cdot 10^{4}} d_t} \prob[  A_v^{t} \land E_j \mid \bigcap\limits_{j' < j } \{ A_v^{t} \land E_{j'} \}]
\end{align*}
where in the second inequality we used that $|\N{t}{v} \setminus \N{\tc}{v}| \geq \nicefrac{\epsilon}{8 \cdot 10^{4}}|\N{t}{v}|$. 
Note that since:
\begin{align*}
    \left(1 - \nicefrac{1}{2 \cdot d_t}\right)^{ \left(\nicefrac{10^{10} \log n}{\epsilon} \right) \cdot \left(\nicefrac{\epsilon}{8 \cdot 10^{4}} d_t \right)} < 1/n^{10^{4}}
\end{align*}

It is enough to argue that:
\begin{align*}
    \prob[  A_v^{t} \land E_j \mid \bigcap\limits_{j' < j } \{ A_v^{t} \land E_{j'} \}] < (1 - \nicefrac{1}{2 \cdot d_t})^{10^{10} \log n / \epsilon}
\end{align*}

Let $\Tilde{t}_{w_{j'}}, j'<j$, be the random variable denoting the last time before $t_{w_{j'}}$ that $I^{l_t}_{v}$ was updated.
Since $t < t_{w_j}$:

\begin{enumerate}
    \item $A_v^{t}$ implies that:
        \begin{enumerate}
            \item $t \leq \Tilde{t}_{w_j} < t_{w_j}$; and
            \item $w_j \in \N{\Tilde{t}_{w_j}}{v}$
        \end{enumerate}
    \item  By the definition of $\Tilde{t}_{w_{j}}$: $|\N{\Tilde{t}_{w_j}}{v}| < 2 \cdot d_t$
\end{enumerate}

Let $r_i$ be the $i$-th random sample used to construct $I^{l_t}_v$ at time $t_w$, i.e., $I_{w_j, t_{w_j}}$. By the principle of deferred decisions it is enough to decide whether $r_i$ is equal to $w_j$ at time $t_{w_j}$. Note that given $\bigcap\limits_{j' < j } \{ A_v^{t} \land E_{j'} \}$, $r_i$ is a uniform at random sample from a set that contains $w_j$ and has at most $2 \cdot d_t$ elements, i.e., set $\N{\Tilde{t}_{w_j}}{v} \setminus \bigcup_{j' < j} w_{j'}$. Thus:

\begin{align*}
    \prob[  \left\{ r_i = v \right\} \mid \bigcap\limits_{j' < j } \{ A_v^{t} \land E_{j'} \} \}] \leq (1 - \nicefrac{1}{2 \cdot d_t})
\end{align*}
and the proof is concluded by noting that $I_{w_j, t_{w_j}} = \bigcup\limits_{i \leq 10^{10} \log n / \epsilon}{r_i}$ and $r_i$'s are independent.
\end{proof}

As already mentioned, we prove that for time $t_i$ most of $u_i$'s neighboring nodes have almost their final neighborhood. Towards that goal, for every $t_i$ we define a random set which contains all neighboring nodes of $u_i$ whose neighborhood at time $t_i$ is ``very'' different from their final one.

\begin{definition}\label{def: random set D_{t_i}}
    For $t_i$ we define the random set: $ D_{t_i} = \{  v : v\in \N{t_i}{u_i} \cap (R \cup L^i)  \land \overline{M_v^{t_i}} \}$.
\end{definition}

We prove that the size of this random set is small with high probability.

\begin{lem}\label{lem: almost all nodes have almost their final neighborhood at time ti}
    $\forall t_i$ we have that:
    $$
    \prob[ |D_{t_i}| \geq \epsilon/10^{5}  |\N{t_i}{u_i}|] < 1/n^{10^{3}}
    $$
\end{lem}

\begin{proof}
In this proof when we write $I_{u_i}$ we refer to $u_i$'s neighborhood sample constructed via the connect procedure at that specific time $t_i$.
Note that:
\begin{align*}    
\prob[ |D_{t_i}| \geq \epsilon/10^{5}  |\N{t_i}{u_i}|]
&= \prob[\{ |D_{t_i}| \geq \epsilon/10^{5}  |\N{t_i}{u_i}|\} \land \{ D_{t_i} \cap I_{u_i} = \emptyset \}]\\ 
&+ \prob[\{ |D_{t_i}| \geq \epsilon/10^{5}  |\N{t_i}{u_i}|\} \land \{D_{t_i} \cap I_{u_i} \neq \emptyset \}]
\end{align*}

We bound each of the two terms in the right hand side independently.

For the first term we have:
\begin{align*}   
\prob[\{ |D_{t_i}| \geq \epsilon/10^{5}  |\N{t_i}{u_i}|\} \land \{ D_{t_i} \cap I_{u_i} = \emptyset \}] 
    & \leq (1 - \epsilon/10^{5})^{10^{10} \log n / \epsilon} \\
    & \leq 1 / n^{10^{5}}
\end{align*}
since $\N{t_i}{u_i} \setminus D_{t_i}$ is at least a $(1 - \epsilon/10^{5})$ fraction of $\N{t_i}{u_i}$.

For the second term we have:

\begin{align*}   
\prob[\{ |D_{t_i}| \geq \epsilon/10^{5}  |\N{t_i}{u_i}|\} \land \{D_{t_i} \cap I_{u_i} \neq \emptyset \}]
    &= \sum_{v \in C} \prob[ \{ v \in \N{t_i}{u_i} \} \land \{ v \in D_{t_i} \} \land \{ v \in I_{u_i} \} ] \\
    &\leq \sum_{v \in C} \prob[ \{ v \in \N{t_i}{u_i} \} \land \{ v \in D_{t_i} \} \land \{ v \in A_{v}^{t_i} \} ] \\
    &\leq \sum_{v \in C} \prob[ \{ v \in D_{t_i} \} \land  A_{v}^{t_i} ] \\
    &\leq \sum_{v \in C} \prob[ \{ v \in R \cup L^i \} \land \{ \overline{M_v^{t_i}} \} \land  A_{v}^{t_i} ] \\
    &\leq n \cdot  (1 /n^{10^{4}})\\
    &\leq 1 /n^{10^{4}-1}
\end{align*}
Where in the first inequality note that $v \in I_{u_i}$ implies $A_v^{t_i}$ (since $v$ receives a notification from $u_i$), in the second inequality we used that $D_{t_i} \subseteq \N{t_i}{u_i}$, in the third inequality we used the definition of set $D_{t_i}$ and in the fourth inequality we used~\cref{lem: different neighborhood has very low probability}.

Combining the two upper bounds concludes the proof.
\end{proof}

The next lemma argues that if $D_{t_i}$ is small and $\N{t_i}{u_i} \simeq C$ then $u_i$ is heavy at time $t_i$.

\begin{thm}\label{thm: ui are heavy}
    If $ \{ |D_{t_i}| < \epsilon/10^{5}  |\N{t_i}{u_i}| \} \land M_{u_i}^{t_i}$ is true then $u_i$ is $\epsilon/10$-heavy at time $t_i$. In addition, if $u_i$ also enters the anchor set at time $t_i$ then it remains in it at least until time $\tc$.
\end{thm}
\begin{proof}

By the definition of $D_{t_i}$ we have that $\forall w \in \N{t_i}{u_i} \cap (R \cup L^i) \setminus D_{t_i}$ the event $M_w^{t_i}$ is true. Thus, since $M_{u_i}^{t_i}$ is also true from~\cref{lem: Oui and Ow suffice for ui and w being in agreement} we can conclude that $u_i$ and $w$ are in $\epsilon/10$-agreement.

We now argue that $u_i$ is in $\epsilon/10$ agreement with at least a $1 - \epsilon/10$ fraction of its neighborhood at time $t_i$. For this, note that:
\begin{align*}
    |\N{t_i}{u_i} \setminus \left( \N{t_i}{u_i} \cap (R \cup L^i) \setminus D_{t_i} \right)| 
        &\leq |\N{t_i}{u_i} \setminus \left( \N{t_i}{u_i} \cap (R \cup L^i) \right)  | + | D_{t_i}| \\
        & = |\N{t_i}{u_i} \setminus (R \cup L^i) | + | D_{t_i}| \\
        & \leq |\N{t_i}{u_i} \setminus C | + | C \setminus (R \cup L^i) | + | D_{t_i}| \\
        & \leq |\N{t_i}{u_i} \setminus C | + | C \setminus R | + | D_{t_i}| \\
        & = |\N{t_i}{u_i} \setminus C | + | L | + | D_{t_i}| \\
        &\leq \epsilon /10^{3} | \N{t_i}{u_i}| + \epsilon /10^{4} | C| + |D_{t_i}|\\
        &\leq \epsilon /10^{3} | \N{t_i}{u_i}| + 5 \epsilon /10^{4} | \N{t_i}{u_i}| + |D_{t_i}|\\
        &\leq  \epsilon /10^{3} | \N{t_i}{u_i}| + 5 \epsilon /10^{4} | \N{t_i}{u_i}| +  \epsilon / 10^{5} |\N{t_i}{u_i}|\\
        &\leq \epsilon /10 | \N{t_i}{u_i}|
\end{align*}
Where in the sixth inequality we used (2) of~\cref{lem: high degree implies very different neighborhood} (since $M_{u_i}^{t_i}$ is true), in the seventh inequality we used that  $|D_{t_i}| < \epsilon / 10^{5} |\N{t_i}{u_i}|$ and the last inequality holds for $\epsilon$ small enough.

To prove the second claim of the theorem, we note that if $u_i$ enters the anchor set at time $t_i$ then to exit that set before time $\tc$ at least an $\epsilon$-fraction of its neighborhood at time $t_i$ needs to call the $Clean (\cdot, \epsilon)$ procedure and consequently participate in an ``interesting event''. That is, for $\epsilon$ small enough, at least $\epsilon (1 - \epsilon/10)|\N{t_i}{u_i}| \geq \epsilon/2 |\N{t_i}{u_i}|$ nodes in $\N{t_i}{u_i}$ need to participate in an ``interesting event''. However, such nodes can be at most:
\begin{align*}
    |\N{t_i}{u_i} \setminus C| + |L \setminus (R \cup L^i) | 
            &\leq |\N{t_i}{u_i} \setminus C| + |L|\\
            &\leq \epsilon/10^{3} | \N{t_i}{u_i}| + \epsilon/10^{4} | C| \\
            &\leq \epsilon/10^{3} | \N{t_i}{u_i}| + \epsilon/10^{3} | \N{t_i}{u_i}|\\
            &\leq 2 \cdot \epsilon/10^{3} | \N{t_i}{u_i}| \\
\end{align*}
which is smaller than $\epsilon/2 |\N{t_i}{u_i}|$. In the second inequality we used the cardinality of $L$ and in the third inequality we used (2) of~\cref{lem: high degree implies very different neighborhood}.
\end{proof}
Note that $A_{u_i}^{t_i}$ is always true from the definition of $u_i$ and $t_i$. Combining~\cref{lem: different neighborhood has very low probability} and~\cref{lem: almost all nodes have almost their final neighborhood at time ti} we deduce that event $ \{ |D_{t_i}| < \epsilon/10^{5}  |\N{t_i}{u_i}| \} \land M_{u_i}^{t_i}$ happens with high probability. Thus, from~\cref{thm: ui are heavy} we conclude that $u_i$ is $\epsilon/10$-heavy at time $t_i$ with high probability.

The proof proceeds arguing that the following events happen with high probability:
\begin{itemize}
    \item For every node $v \in R$ there exists a node $u_i \in L_1$ such that: (1) $u_i$ enters the anchor set at time $t_i$; (2) $u_i$ is in agreement with node $v$ at time $t_i$; and (3) $u_i$ does not exit the anchor set before time $\tc$. Consequently edge $(u_i, v)$ is added to our sparse solution graph $\Tilde{G_{t_i}}$ and remains in our sparse solution at least until time $\tc$.
    \item Similarly, we argue that for every node $v \in L_1$ there exists a node $u_i \in L_2$ such that (1), (2) and (3) hold. So that edge $(u_i, v)$ is added to our sparse solution graph $\Tilde{G_{t_i}}$ and remains in our sparse solution at least until time $\tc$.
    \item The last part of the proof argues how the \textit{Connect}($\cdot, \epsilon$) procedure clusters nodes in $L_2$ that do not enter the anchor set with the rest of $C$. At time $t_i$ when $u_i\in L_2 $ participates in an ``interesting event'' most nodes in $R$ have their neighborhood similar to their final, i.e., at time $\tc$, neighborhood, consequently they are in agreement with $u_i$. In addition to that, $u_i$ is in agreement with almost all nodes in $L_1$, therefore there are many triangles of the form $u_i, v, w$ where $v \in R$, $w \in \Phi_v \cap L_1$ and all three nodes are in agreement. The \textit{Connect}($u_i, \epsilon$)  connects $u_i$ to the rest of the cluster with high probability using the edge $(u_i, w)$ in one of those triangles.
\end{itemize}

To prove the aforementioned points, we define the following random set:
\begin{definition}\label{def: random set Tv}
    For $v \in C$ we define the random set:
    $$
    T_v = \left\{i \in \{1,2,\dots, \lfloor \epsilon/10^{4} |C| \rfloor \}: \overline{M_v^{t_i}} \land \{ (u_i, v) \in E_{t_i} \} \right\}
    $$
\end{definition}

For a node $v \in C$,~\cref{def: random set Tv} captures which times during the last $\lfloor \epsilon/10^{4} |C| \rfloor$ ``interesting events'', node $v$ was connected to $u_i$ and its neighborhood was different from its final one. We proceed by proving that for every node in $v \in R \cup L_1$, $|T_v|$ is small with high probability.

\begin{lem}\label{lem: v in R has almost the same neighborhood as its final neighborhood for the largest part of the last important events}
Let $v \in R \cup L_1$, then
$$
\prob [|T_v| \geq \epsilon/10^{5} |C| ] < 1 / n^{10^{2}}
$$
\end{lem}
\begin{proof}
In this proof, for a node $u_i \in L$ when we write $I_{u_i}$ we refer to $u_i$'s neighborhood sample constructed via the connect procedure at time $t_i$. For short we write $\left[ \epsilon/10^{4} |C|  \right]$ instead of $\{1,2,\dots, \lfloor \epsilon/10^{4} |C| \rfloor \}$. We first argue that event $\bigcap_{i \in \left[ \epsilon/10^{4} |C|  \right]} M_{u_i}^{t_i}$ happens with high probability.
\begin{align*}
    \prob [\bigcap_{i \in \left[ \epsilon/10^{4} |C|  \right]} M_{u_i}^{t_i}] &=
    1 - \prob [\exists i \in \left[ \epsilon/10^{4} |C|  \right]: \overline{M_{u_i}^{t_i}}]\\
    &\geq 1 - \sum_{i \in \left[ \epsilon/10^{4} |C|  \right]} \prob [\overline{M_{u_i}^{t_i}}]\\
    &\geq 1 - \sum_{i \in \left[ \epsilon/10^{4} |C|  \right]} \prob [\overline{M_{u_i}^{t_i}} \land A_{u_i}^{t_i}]\\
    &\geq 1 - n \cdot 1/n^{10^{4}}\\
    &\geq 1 - 1/n^{10^{3}}\\
\end{align*}

where in the first inequality we used the union bound, in the second inequality we used that $A_{u_i}^{t_i}$ is always true by definition of $u_i$ and $t_i$ and in the third inequality we used~\cref{lem: different neighborhood has very low probability}.
Thus using
\begin{align*}
    \prob [ |T_v| \geq \epsilon/10^{5} |C| ] &=
\prob [ \{|T_v| \geq \epsilon/10^{5} |C| \} \land \bigcap_{i \in \left[ \epsilon/10^{4} |C|  \right]} M_{u_i}^{t_i} ] +
\prob [ \{|T_v| \geq \epsilon/10^{5} |C| \} \land \overline{\bigcap_{i \in \left[ \epsilon/10^{4} |C|  \right]} M_{u_i}^{t_i}} ] \\
&\leq \prob [ \{|T_v| \geq \epsilon/10^{5} |C| \} \land \bigcap_{i \in \left[ \epsilon/10^{4} |C|  \right]} M_{u_i}^{t_i} ] + 1 / n^{10^{3}}
\end{align*}
We focus on upper bounding the first term of the right hand side. Again, using the law of total probability we have that:
\begin{align*}
\prob [ \{|T_v| \geq \epsilon/10^{5} |C| \} \land \bigcap_{i \in \left[ \epsilon/10^{4} |C|  \right]} M_{u_i}^{t_i} ] &= 
\prob [ \{|T_v| \geq \epsilon/10^{5} |C| \} \land \bigcap_{i \in \left[ \epsilon/10^{4} |C|  \right]} M_{u_i}^{t_i} \land \{ \forall i \in T_v \; v \not \in I_{u_i}\} ]\\
&+\prob [ \{|T_v| \geq \epsilon/10^{5} |C| \} \land \bigcap_{i \in \left[ \epsilon/10^{4} |C|  \right]} M_{u_i}^{t_i} \land \{ \exists i \in T_v: \; v \in I_{u_i}\} ]
\end{align*}

For the first term note that for $\epsilon$ small enough and using (1) of~\cref{lem: high degree implies very different neighborhood} we have that: if $\bigcap_{i \in \left[ \epsilon/10^{4} |C|  \right]} M_{u_i}^{t_i}$ is true then $\forall i \in \left[ \epsilon/10^{4} |C|  \right]$ we have that $ |\N{t_i}{u_i}| < 2 |C|$. Now consider the following random process: each time our algorithm samples $u_i$'s neighborhood to create $I_{u_i}$ we sample uniformly at random the set $\N{t_i}{u_i} \cup \{ r_{|\N{t_i}{u_i}|+1},r_{|\N{t_i}{u_i}|+2},\dots, r_{2|C|}  \}$ to create sets $\Tilde{I}_{u_i}$ where $r_j$ nodes only exist for analysis purposes. Now note that:
\begin{align*}
    &\prob [ \{|T_v| \geq \epsilon/10^{5} |C| \} \land \bigcap_{i \in \left[ \epsilon/10^{4} |C|  \right]} M_{u_i}^{t_i} \land \{ \forall i \in T_v :\; v \not \in I_{u_i}\} ] \leq \\
    &\prob [ \{|T_v| \geq \epsilon/10^{5} |C| \} \land \bigcap_{i \in \left[ \epsilon/10^{4} |C|  \right]} M_{u_i}^{t_i} \land \{ \forall i \in T_v :\; v \not \in \Tilde{I}_{u_i}\} ] \leq \\
    &\left( 1-  1 /(2 |C|) \right)^{(10^{10} \log n/ \epsilon) \cdot (\epsilon/10^{5} |C|)} \leq\\
    &\left( 1-  1 /(2 |C|) \right)^{10^{5} \log n |C|} \leq\\
    &(1/n)^{10^{5} /2 } \leq\\
    & 1/n^{10^{4}}
\end{align*}

We continue by upper bounding $\prob [ \{|T_v| \geq \epsilon/10^{5} |C| \} \land \bigcap_{i \in \left[ \epsilon/10^{4} |C|  \right]} M_{u_i}^{t_i} \land \{ \exists i \in T_v: \; v \in I_{u_i}\} ]$ as follows:

\begin{align*}   
\prob [ \{|T_v| \geq \epsilon/10^{5} |C| \} \land \bigcap_{i \in \left[ \epsilon/10^{4} |C|  \right]} M_{u_i}^{t_i} \land \{ \exists i \in T_v: \; v \in I_{u_i}\} ] 
    &\leq \sum_{i \in \{1,2,\dots, \epsilon/10^{4} |C| \}} \prob[ i \in T_v , v \in I_{u_i} ] \\
    &\leq \sum_{i \in \{1,2,\dots, \epsilon/10^{4} |C| \}} \prob[ \Tilde{M_v^{t_i}} \land  A_v^{t_i} ] \\
    &\leq n \cdot (1 /n^{10^{4}})\\
    &\leq 1 /n^{10^{4}-1}
\end{align*}
Where in the first inequality we used the union bound, in the second inequality we used the fact that event $\{v \in I_{u_i}\}$ implies $A_v^{t_i}$, and in the third one we used~\cref{lem: different neighborhood has very low probability}.

The proof follows by noting that $ 1 / n^{10^{3}} + 1 / n^{10^{4}} + 1 / n^{10^{4}-1} < 1 / n^{10^{2}} $
\end{proof}
We continue proving that, with high probability, every node in $R$ is selected by a node in $L_1$ which enters the anchor set and remains in the anchor set until at least time $\tc$.
\begin{lem}\label{lem: all nodes in R get selected by some node in A1}
For every node $v \in R$ let 
\begin{align*}
    X_v = \left\{ \exists i \in [\epsilon/(2\cdot 10^{4}) |C| ]: 
\{u_i \in L_1 \text{ enters the anchor set }\} \land \{ (v, u_i) \in  E_{t_{i'}}, \forall i' \in [i, \tc] \right\}
\end{align*}
then $\prob [X_v] > 1 - 1/n^{10} $.
\end{lem}
\begin{proof}
Let $i \in \{1,2,\dots, \lfloor \epsilon/(2\cdot 10^{4}) |C| \rfloor \}$ be such that $ \{ |D_{t_i}| < \epsilon/10^{5}  |\N{t_i}{u_i}| \}  \land {M_v^{t_i}} \land {M_{u_i}^{t_i}} \land \{ (u_i, v) \in E_{t_i} \}$ is true, we then know from~\cref{thm: ui are heavy} that if $u_i$ enters the anchor set at that time $t_i$, edge $(u_i, v)$ is added to our sparse solution and remains in it at least until time $\tc$. Thus, it is useful to define the following set:
\begin{align*}
    Y_v = \left\{  i \in \{1,2,\dots, \lfloor \epsilon/(2\cdot 10^{4}) |C| \rfloor: \; \{ |D_{t_i}| < \epsilon/10^{5}  |\N{t_i}{u_i}| \} \land {M_v^{t_i}} \land {M_{u_i}^{t_i}} \land \{ (u_i, v) \in E_{t_i} \} \right\}
\end{align*}
We proceed by proving that with high probability $|Y_v|$ is ``large''.

Note that from~\cref{obs: at tc u is connected to almost everybody in C} we have that: 
\begin{align*}
    |\left\{  i \in \{1,2,\dots, \lfloor \epsilon/(2\cdot 10^{4}) |C| \rfloor: \; (u_i, v) \in E_{t_i} \right\}|
    \geq \epsilon/(2 \cdot 10^{4}) |C| - \epsilon/ 10^{13} |C| > \epsilon/(3 \cdot 10^{4}) |C|
\end{align*}
where the second inequality holds for $\epsilon$ small enough.

For simplicity let:
\begin{align*}
    \mathcal{R} =  \bigcap_{ i \in \left[ \epsilon/(2\cdot 10^{4}) |C|  \right]} \left\{ \{ |D_{t_i}| < \epsilon/10^{5}  |\N{t_i}{u_i}| \}  \land M_{u_i}^{t_i} \right\} \land \{ |T_v| < \epsilon/10^{5} |C| \}
\end{align*}
Note that $\mathcal{R}$ implies that:
\begin{align*}
    |Y_v| \geq \epsilon/(3\cdot 10^{4}) |C| - \epsilon/10^{5} |C| > \epsilon/(4\cdot 10^{4}) |C|    
\end{align*}
where the second inequality holds for $\epsilon$ small enough.
We continue upper bounding the probability that $\mathcal{R}$ does not occur as follows:
\begin{align*}
    \prob [\overline{\mathcal{R}}] \leq &\sum_{i \in \left[ \epsilon/(2\cdot 10^{4}) |C|  \right]} \prob [  \{ |D_{t_i}| \geq \epsilon/10^{5}  |\N{t_i}{u_i}| \}] \\
&+\sum_{i \in \left[ \epsilon/(2\cdot 10^{4}) |C|  \right]} \prob [\overline{ M_{u_i}^{t_i} }]\\
&+ \prob [ \{ |T_v| \geq \epsilon/10^{5} |C| \}] \\
&\leq n \cdot (1/n^{10^{4}}) +  n \cdot (1/n^{10^{3}}) +  1/n^{10^{2}} \\
&\leq 3/n^{10^2}
\end{align*}
where in the first inequality we use the union bound and in the second we use~\cref{lem: different neighborhood has very low probability},~\cref{lem: almost all nodes have almost their final neighborhood at time ti} and~\cref{lem: v in R has almost the same neighborhood as its final neighborhood for the largest part of the last important events}.

Note that $\mathcal{R}$ also implies that $u_i$ enters the anchor set at time $t_i$ with probability at least  $\min \{ \nicefrac{10^{7} \log n}{2 \epsilon|C|}, 1\}$ where we used (1) in~\cref{lem: high degree implies very different neighborhood} for $\epsilon$ small enough. In addition, note that under any realization of the random variables $|\N{t_1}{u_1}|, |\N{t_2}{u_2}|, \dots$ the randomness of the \textit{Anchor} procedure is independent from the randomness of the \textit{Connect} procedure. Consequently:
\begin{align*}
    \prob [\overline{X_v}] &= \prob \overline{[X_v} \cap \overline{\mathcal{R}}] + \prob [\overline{X_v} \cap \mathcal{R}]\\
        &\leq \prob [\overline{\mathcal{R}}] + \prob [\overline{X_v} \cap \mathcal{R}]\\
        &\leq 3/n^{10^{2}} + \prob [\overline{X_v} \cap \mathcal{R}]\\
        &\leq 3/n^{10^{2}} + \prob [\overline{X_v} \cap \{ |Y_v| > \epsilon/(4\cdot 10^{4}) |C|\}]\\
        &\leq 3/n^{10^{2}} + (1 - \min \{ \nicefrac{10^{7} \log n}{2 \epsilon|C|}, 1\})^{\epsilon/(4\cdot 10^{4})|C|} \\
        &\leq 3/n^{10^{2}} + 1/n^{10^{2}} \\
        &< 1/n^{10} \\
\end{align*}
\end{proof}

Using the same line of arguments as~\cref{lem: all nodes in R get selected by some node in A1} in~\cref{lem: all nodes in L1 get selected by some node in L2} we argue that every node in $L_1$ is selected by a node in $L_2$ which enters the anchor set and remains in the anchor set until at least time $\tc$. We state the lemma and omit the proof.

\begin{lem}\label{lem: all nodes in L1 get selected by some node in L2}
For every node $v \in L_1$ let 
\begin{align*}
    X_v = \left\{ \exists i \in \{\epsilon/(2\cdot 10^{4}) |C| + 1,  \dots, \epsilon/10^{4} |C|\} ]: 
\{u_i \in L_2 \text{ enters the anchor set }\} \land \{ (v, u_i) \in  E_{t_{i'}}, \forall i' \in [i, \tc] \right\}
\end{align*}
then $\prob [X_v] > 1 - 1/n^{10} $.
\end{lem}
\begin{proof}
    Omitted as it follows the same line of reasoning with the proof of~\cref{lem: all nodes in R get selected by some node in A1}.
\end{proof}

We underline that while the essence of both~\cref{lem: all nodes in R get selected by some node in A1} and~\cref{lem: all nodes in L1 get selected by some node in L2} could be summarized in a single lemma, in the last part of this section we use the facts that~\cref{lem: all nodes in R get selected by some node in A1} refers to how nodes in $R$ get into the cluster through nodes in $L_1$ and~\cref{lem: all nodes in L1 get selected by some node in L2} refers to how nodes in $L_1$ get into the cluster through nodes in $L_2$.

The last part of the section is devoted in arguing that through the \textit{Connect}($\cdot, \epsilon$) procedure, with high probability, all nodes in $L_2$ get connected to some node in $L_1$ in our sparse solution.

\begin{lem}\label{lem: all nodes in L2 get connected trough the Connect procedure}
    For every node $u_j \in L_2$ let 
    \begin{align*}
    X_{u_j} = \left\{ \exists i \in [\epsilon/(2\cdot 10^{4}) |C| ]: 
\{u_i \in L_1 \text{ enters the anchor set }\} \land \{ (u_j, u_i) \in  E_{t_{i'}}, \forall i' \in [t_j, \tc] \right\}
\end{align*}
then 
\begin{align*}
    \prob [ X_{u_j} ] \geq 1 - 1/n^{10}
\end{align*}
\end{lem}
\begin{proof}
The proof proceeds by arguing that:
\begin{enumerate}
    \item $u_j$ is connected to almost all nodes $R$;
    \item let $w \in R$, then $\N{t_j}{u_j} \cap \N{t_j}{w} \cap L_1$ is a sufficiently large fraction of $L_1$;
    \item with high probability almost all nodes in $\N{t_j}{u_j} \cap \N{t_j}{w} \cap L_1$ have almost their final neighborhood; and
    \item with high probability one of those nodes is in the anchor set.
\end{enumerate}
It is useful to define the following event:
\begin{align*}
    \mathcal{R} =  \bigcap_{ i \in \left[ \epsilon/(2\cdot 10^{4}) |C|  \right]} \left\{ \{ |D_{t_i}| < \epsilon/10^{5}  |\N{t_i}{u_i}| \}  \land M_{u_i}^{t_i} \right\}
\end{align*}
We upper bound the probability that $\mathcal{R}$ does not occur as follows:
\begin{align*}
    \prob [\overline{\mathcal{R}}] \leq &\sum_{i \in \left[ \epsilon/(2\cdot 10^{4}) |C|  \right]} \prob [  \{ |D_{t_i}| \geq \epsilon/10^{5}  |\N{t_i}{u_i}| \}] \\
&+\sum_{i \in \left[ \epsilon/(2\cdot 10^{4}) |C|  \right]} \prob [\overline{ M_{u_i}^{t_i} }]\\
&\leq n \cdot (1/n^{10^{4}}) +  n \cdot (1/n^{10^{3}}) \\
&\leq 1/n^{10^{2}}
\end{align*}
where in the first inequality we use the union bound and in the second we use~\cref{lem: different neighborhood has very low probability} and \cref{lem: almost all nodes have almost their final neighborhood at time ti}.

Let $w \in \N{t_j}{u_j} \cap R$, from~\cref{obs: bounds on degree of u at time tc wrt C} we have that:
\begin{align*}
|\N{t_j}{u_j} \cap \N{t_j}{w} \cap L_1|
    &\geq |L_1| - \epsilon/ 10^{13} |C| - \epsilon/ 10^{13} |C| \\
    &\geq \epsilon/ (2\cdot 10^{4}) |C| - \epsilon/ 10^{13} |C| - \epsilon/ 10^{13} |C|\\
    &\geq \epsilon/ (4 \cdot 10^{4} )|C|
\end{align*}

Moreover, note that event $\mathcal{R}$ implies that at least $(1 - \epsilon/10^{5}) |\N{t_j}{u_j}|$ neighbors of $u_j$ at time $t_j$ have almost their final neighborhood. Consequently, $\mathcal{R}$ implies the following bound for all $w \in \N{t_j}{u_i} \cap R$:
\begin{align*}
     |v \in \N{t_j}{w} \cap \N{t_j}{u_j} \cap L_1: \; M_{v}^{t_j}| &> \epsilon/ (4 \cdot 10^{4} ) |C| -\epsilon/ 10^{5} |\N{t_j}{u_j}|\\
     &> \epsilon/(4 \cdot 10^{4} ) |C| - 2 \epsilon/ 10^{5} |C|\\
     &> \epsilon/ 10^{5} |C|
\end{align*}
where in the second inequality we used (1) in~\cref{lem: high degree implies very different neighborhood} for $\epsilon$ small enough.

As in the proof of~\cref{lem: all nodes in R get selected by some node in A1}, $\mathcal{R}$ implies that $u_i$ enters the anchor set at time $t_i$ with probability at least  $\min \{ \nicefrac{10^{7} \log n}{2 \epsilon|C|}, 1\}$ (using (1) in~\cref{lem: high degree implies very different neighborhood} for $\epsilon$ small enough). In addition, under any realization of the random variables $|\N{t_1}{u_1}|, |\N{t_2}{u_2}|, \dots$ the randomness of the \textit{Anchor} procedure is independent from the randomness of the other procedures of our algorithm. Consequently, $\forall w \in \N{t_j}{u_j} \cap R$ let $H_w = \left\{ \{v \in \N{t_j}{w} \cap \N{t_j}{u_j} \cap L_1: \; M_{v}^{t_j}\} \cap \Phi = \emptyset \right\}$, we have:

\begin{align*}
    \prob [\mathcal{R} \cap  H_w ] &\leq (1 - \min \{ \nicefrac{10^{7} \log n}{2 \epsilon|C|}, 1\})^{\epsilon/10^{5} |C|}\\
    &\leq e^{-  \nicefrac{10^{7} \log n}{2 \epsilon|C|}\cdot {\epsilon/10^{5} |C|}}\\
    &\leq 1/n^{50}\\
\end{align*}


From~\cref{obs: at time t the v is connected with most nodes in C} we have that $|\N{t_j}{u_j} \cap C| > (1- 2\epsilon/10^{4}) |C|$, consequently:
\begin{align*}
|\N{t_j}{u_j} \cap R|
&\geq |\N{t_j}{u_j} \cap C| - |C \setminus R|\\
&= |\N{t_j}{u_j} \cap C| - |L|\\
&\geq(1- 2 \epsilon/ 10^{4})|C| - \epsilon/ 10^{4}|C|\\
&\geq(1- 3\epsilon/10^{4})|C|
\end{align*}
Again, using (1) in~\cref{lem: high degree implies very different neighborhood} for $\epsilon$ small enough, $\mathcal{R}$ implies that
\begin{equation*}
    |\N{t_j}{u_j} \cap R| \geq(1- 3\epsilon/(2\cdot 10^{4}))|\N{t_j}{u_j}| >(1- 4\epsilon/10^{4})|\N{t_j}{u_j}|
\end{equation*}
where the second inequality holds for $\epsilon$ small enough.

We are now ready to upper bound $\prob [ \overline{X_{u_j}} \cap \mathcal{R} ]$, note that $J_{u_j}$ denotes the random sample constructed via the Connect procedure of $u_j$'s neighborhood at time $t_j$.
\begin{align*}
    \prob [ \overline{X_{u_j}} \cap \mathcal{R} ] &\leq  \prob [ \overline{X_{u_j}} \cap \mathcal{R} \cap \{ J_{u_j} \cap \N{t_j}{u_j} \cap R = \emptyset \}] \\
    &+ \prob [ \overline{X_{u_j}} \cap \mathcal{R} \cap \{ J_{u_j} \cap \N{t_j}{u_j} \cap R \neq \emptyset \}]\\
    &\leq (4 \epsilon / 10^{4})^{10^{5} \log n / \epsilon} + \prob [\mathcal{R} \bigcap \bigcup_{w \in J_{u_j} \cap \N{t_j}{u_j} \cap R} H_w]\\
     &\leq (4 \epsilon / 10^{4})^{10^{5} \log n / \epsilon} + n \cdot (1/n^{50})\\
     &\leq 1/n^{20}
\end{align*}
We conclude again using the law of total probability:
\begin{align*}
    \prob [ \overline{X_{u_j}} ] &= \prob [ \overline{X_{u_j}} \cap \mathcal{R} ] + \prob [ \overline{X_{u_j}} \cap \overline{\mathcal{R}} ] \\
    &\leq \prob [ \overline{X_{u_j}} \cap \mathcal{R} ] + \prob [\overline{\mathcal{R}} ] \\
    &\leq 1/n^{20} + 1/n^{10^{2}}\\
    &\leq 1/n^{10}\\
\end{align*}
\end{proof}

\begin{thm}\label{thm: all nodes of C are clustered together}
   With probability at least $(1 - 1/n^{8})$ all nodes of $C$ are clustered together at time $\tc$ by our algorithm.
\end{thm}
\begin{proof}

As in previous lemmas we define the following event:
\begin{align*}
    \mathcal{R} =  \bigcap_{ i \in \left[ \epsilon/ 10^{4} |C|  \right]} \left\{ \{ |D_{t_i}| < \epsilon/10^{5}  |\N{t_i}{u_i}| \}  \land M_{u_i}^{t_i} \right\}
\end{align*}
for which $\prob [\overline{\mathcal{R}}] \leq 1/n^{10^{2}}$.

From~\cref{thm: ui are heavy} we have that $\mathcal{R}$ implies that $\forall u_i \in L$ that joins the anchor set at time $t_i$:
\begin{enumerate}
    \item they remain there until at least time $\tc$; and
    \item all edges $(v, u_i)$ where $v\in R$ added in our sparse solution are not deleted until at least time $\tc$.
\end{enumerate}

We now argue that $\mathcal{R}$ also implies that any two nodes $u_i, u_j \in L$ which joined the anchor set at times $t_i$ and $t_j$ respectively belong to the same connected component of our sparse solution at time $\tc$. 
To that end, let $W_{t_i}$ and $W_{t_j}$ denote the set of nodes that are in agreement respectively with $u_i$ and $u_j$ at times $t_i$ and $t_j$. 
\begin{align*}
    |W_{t_i}| &\geq (1 - \epsilon/10) |\N{t_i}{u_i}| \\
    &> (1 - \epsilon/10) (1 - 2 \epsilon/ 10^{4})|C| > \\
    &>(1 - 2 \epsilon/10) |C|
\end{align*}
where in the second inequality we used~\cref{obs: at time t the v is connected with most nodes in C}.

In addition, from~\cref{obs: at time t the v is connected with most nodes in C} we have that $|\N{t_i}{u_i} \cap C| > (1- 2\epsilon/10^{4}) |C|$, consequently:
\begin{align*}
|\N{t_i}{u_i} \cap R|
&\geq |\N{t_i}{u_i} \cap C| - |C \setminus R|\\
&= |\N{t_i}{u_i} \cap C| - |L|\\
&\geq(1- 2 \epsilon/ 10^{4})|C| - \epsilon/ 10^{4}|C|\\
&\geq(1- 3\epsilon/10^{4})|C|
\end{align*}
Consequently, combining the latter two inequalities and the fact that $W_{t_i} \subseteq \N{t_i}{u_i}$:
\begin{align*}
    |W_{t_i} \cap R| &> (1 - 2\epsilon/10 - 3\epsilon/10^{4}) |C|\\
    &> (1 - 3 \epsilon/10) |C|
\end{align*}
where the second inequality holds for $\epsilon$ small enough.

Thus:
\begin{align*}
     |W_{t_i} \cap W_{t_j} \cap R| &> (1 - 6 \epsilon/10) |C|\\
     &\geq 1
\end{align*}
where the second inequality holds for $\epsilon$ small enough.

To conclude, let $C_{correct}$ be the event that all nodes in $C$ are clustered together by our algorithm at time $\tc$. Then:
\begin{align*}
    \prob [\overline{C_{correct}}] &= \prob \left[\overline{C_{correct}} \cap \left\{ \mathcal{R} \bigcap_{v \in C} X_v \right\}\right] +\prob \left[\overline{C_{correct}} \cap \overline{\left\{ \mathcal{R} \bigcap_{v \in C} X_v \right\}}\right]\\
    &\leq \prob \left[\overline{C_{correct}} \cap \left\{ \mathcal{R} \bigcap_{v \in C} X_v \right\}\right] +\prob \left[\overline{\left\{ \mathcal{R} \bigcap_{v \in C} X_v \right\}}\right]\\
    &\leq 0 +\prob \left[\overline{\left\{ \mathcal{R} \bigcap_{v \in C} X_v \right\}}\right]\\
     &\leq \prob \left[\left\{ \overline{\mathcal{R}} \cup \bigcup_{v \in C} \overline{X_v} \right\}\right]\\
    &\leq 0 + 1/n^{10^{2}}+ n \cdot  1/n^{10}\\
    &<1/n^{8}
\end{align*}
where in the first inequality we used the law of total probability, in the third we used that $\mathcal{R} \bigcap_{v \in C} X_v $ implies $C_{correct}$ and in the fifth one we used~\cref{lem: all nodes in R get selected by some node in A1}, ~\cref{lem: all nodes in L1 get selected by some node in L2} and~\cref{lem: all nodes in L2 get connected trough the Connect procedure}.
\end{proof}
\newcommand{\ta}[1]{t_{#1}^{\text{arrival}} }
\newcommand{\tl}[1]{t_{#1}^{\text{last}} } 
\newcommand{\Canchor}{C_{\Phi}}
\newcommand{\Cnotanchor}{C_{\overline{\Phi}}}
\newcommand{\Nt}[2]{N_{\widetilde{G}_{#1}}(#2) } 

\section{All Found Clusters are Dense.}\label{sec: All found clusters are dense}
The goal of this section is to prove that all clusters found by our algorithm are dense, i.e., any node $u$ that belongs to a cluster $C$ is connected to almost all nodes in $C$ in graph $G_t$.
We underline that a cluster $C$ found by our algorithm is always induced by a connected component of the sparse solution $\widetilde{G_t}$ and our goal is to prove the main~\cref{thm: C is dense} of this section which states that $\forall u \in C$, $|\N{t}{u} \cap C| \geq (1 - 541080 \epsilon) |C| $ for a small enough $\epsilon$.

Similarly to the notation of~\cref{sec: Finding dense clusters} we denote by $\tc$ the current time and by $C$ a cluster found by our algorithm at that time.
For all $u \in C$ we denote by $t_u$ the last time before $\tc + 1$ that $u$ participated in an ``interesting event'', note that $t_u$ (similarly to the definition of times $t_i$ in~\cref{sec: Finding dense clusters}) is a random variable. We denote by $C_{\Phi_t}$ the subset of nodes in $C$ that are in the anchor set at time $t$, i.e., $\Phi_t \cap C$. We avoid the subscript $t$ in the set $\Canchor$ notation as it will always be clear for the context at what time we are referring to. Equivalently, we denote by $\Cnotanchor$ the rest of the nodes in $C$, i.e., $C \setminus \Canchor$, at time $t$.

Initially we prove the two crucial lemmas, these are~\cref{lem: When my degree drops I get notified} and~\cref{lem: bound on Adu}. In both lemmas we use the properties of our notification procedure and argue that with high probability for a node $u$ after time $t_u$ (and at least until time $\tc$):
\begin{itemize}
    \item~\cref{lem: When my degree drops I get notified}: $u$ does not lose more than a very small fraction of its neighborhood after time $t_u$.
    \item~\cref{lem: bound on Adu} $u$'s neighborhood does not increase with many nodes of ``small'' degree.
\end{itemize}

We proceed to the formal statement and proof of these two lemmas and start with~\cref{lem: When my degree drops I get notified} where we actually prove something slightly stronger that what was mentioned in the previous paragraph. To facilitate the description of the next lemmas, similarly to~\cref{sec: Finding dense clusters}, we define the following:

\begin{definition}\label{def: Mv and Av}
For a node $v \in C$ and t  we define the following events:
\begin{align*}
    T_i^{v, t} &= \left\{  v \text{ received a } Type_i \text{ notification at time } t \right\}\\
    T_i^{v, t, t'} &= \left\{  v \text{ received a } Type_i \text{ notification during the interval } (t, t'] \right\}\\
    \mathcal{T}^{v, t, t'} &= \left\{  v \text{ participated in an ``interesting event'' during the interval } (t, t'] \right\}\\
    A_v^{t} &= \left\{ \left\{ v \in \I_{t}\right\} \lor \left\{ v \text{ received a } Type_2 \text{ notification at time } t\right\} \right\}\\
    M_v^{t, t'} &= \{|\N{t}{v} \setminus \N{t'}{v}| \leq \nicefrac{\epsilon}{10^{5}}|\N{t}{v}| \}\\
\end{align*}
\end{definition}

\begin{lem}\label{lem: When my degree drops I get notified}
Let $u \in V$ and times $t', t $ where $t < t'$. Then:
\begin{equation*}
    \prob \left[  A_v^t \land \overline{T_0^{v, t, t'}} \land \overline{M_v^{t, t'}} \right] < 1/n^{10^{4}}
\end{equation*}
\end{lem}

\begin{proof}
The proof of the current lemma follows the same line of arguments as the proof of~\cref{lem: different neighborhood has very low probability} and it is omitted.
\end{proof}

To facilitate the description of the~\cref{lem: bound on Adu} we define the following random variable.

\begin{definition}
    For every node $v \in V$ and times $t, t'$ where $t < t'$ and a positive integer $d$:
    \begin{align*}
       P_{d, t, t'}^v =  \left( \N{t'}{v} \setminus \N{t}{v}\right) \cap  \left\{ w: \exists t'' \in (t, t']:  |\N{t''}{w}| < 10^{2}d \right\}\
    \end{align*}
\end{definition}
In other words $P_{d, t, t'}^v$ contains all neighbors of $v$ that arrived between times $t$ and $t'$ whose degree at some point between those times is ``small''. 

In~\cref{lem: bound on Adu} we argue that $ |P_{d, t_v, t'}^v| < \epsilon/10^{2} d$ with high probability. The intuition behind the latter statement is that if for some $t$ $ |P_{d, t, t'}^v| \geq \epsilon/10^{2} d$ then, with high probability, $v$ participates in an ``interesting event'' at a time $t'' \in (t, t']$.

\begin{lem}\label{lem: bound on Adu}
 For every node $v \in V$, time $t'$ and a positive integer $d$:
    \begin{align*}
    \prob [ \left\{|P_{d, t_v, t'}^v| > \epsilon/10^{2} d \right\}] < 3 / n^{10^{3}}
    \end{align*}
\end{lem}

\begin{proof}
We define the following event:
\begin{align*}
    \mathcal{R} =  \bigcap_{u \in V, t_1, t_2: t_1 < t_2 } 
    \left\{ 
    \overline{A_u^{t_1}} \lor T_0^{u, t_1, t_2} \lor M_u^{t_1, t_2}
    \right\}
\end{align*}
for which, using~\cref{lem: When my degree drops I get notified} and a union bound, we get: $\prob [\overline{\mathcal{R}}] \leq 1/n^{10^{4} - 3}$.
Using the law of total probability, we have:
\begin{align*}
    \prob \left[ \left\{|P_{d, t_v, t'}^v| > \epsilon/10^{2} d \right\} \right]  &= \prob \left[ \left\{|P_{d, t_v, t'}^v| > \epsilon/10^{2} d \right\}  \wedge \mathcal{R} \right] + \prob \left[ \left\{|P_{d, t_v, t'}^v| > \epsilon/10^{2} d \right\} \wedge \overline{\mathcal{R}} \right]\\
    &\leq\prob \left[ \left\{|P_{d, t_v, t'}^v| > \epsilon/10^{2} d \right\} \wedge \mathcal{R} \right] + \prob \left[\overline{\mathcal{R}} \right] \\
    &\leq \prob \left[ \left\{|P_{d, t_v, t'}^v| > \epsilon/10^{2} d \right\} \wedge \mathcal{R} \right] + 1/n^{10^{4} - 3}
\end{align*}
Thus, in the rest of the proof we focus on upper bounding the term $\prob \left[ \left\{|P_{d, t_v, t'}^v| > \epsilon/10^{2} d \right\} \wedge \mathcal{R} \right]$. 

To that end, for each node $u \in V$ let $\ta{u}$ be its arrival time and note that $\forall u \in P_{d, t_v, t'}^v$ we have that $ \ta{u} \in (t, t']$. We define the following sets:
\begin{align*}
    S &= \left\{u \in  P_{d, t_v, t'}^v: |\N{\ta{v}}{v}| \leq 10^{3} d  \right\}\\
    L &= \left\{u \in  P_{d, t_v, t'}^v: |\N{\ta{v}}{v}| > 10^{3} d  \right\}
\end{align*}

The set $S$ contains nodes of $ P_{d, t_v, t'}^v$ whose degree when they arrived was relatively ``small'' and on the contrary $L$, which is equal to $ P_{d, t_v, t'}^v \setminus S$,  contains nodes whose degree on arrival was ``large''. Since $ |P_{d, t_v, t'}^v| = |L| + |S|, $ we have that event $\left\{|P_{d, t_v, t'}^v| > \epsilon/10^{2} d \right\} \wedge \left\{ |S| < \epsilon/10^{3} d \right\}$ implies $ \left\{ |L| > 9\epsilon/10^{3} d \right\}$. Using the total law of probability we have:
\begin{align*}
    &\prob \left[ \left\{|P_{d, t_v, t'}^v| > \epsilon/10^{2} d \right\} \wedge \mathcal{R} \right] \\
    &= \prob \left[ \left\{|P_{d, t_v, t'}^v| > \epsilon/10^{2} d \right\} \wedge \mathcal{R} \wedge \left\{ |S| \geq \epsilon/10^{3} d \right\} \right] + \prob \left[ \left\{|P_{d, t_v, t'}^v| > \epsilon/10^{2} d \right\} \wedge \mathcal{R} \wedge \left\{ |S| < \epsilon/10^{3} d \right\} \right] \\
    &\leq \prob \left[ \left\{|P_{d, t_v, t'}^v| > \epsilon/10^{2} d \right\} \wedge \mathcal{R} \wedge \left\{ |S| \geq \epsilon/10^{3} d \right\} \right] + \prob \left[ \left\{|P_{d, t_v, t'}^v| > \epsilon/10^{2} d \right\} \wedge \mathcal{R} \wedge \left\{ |L| > 9 \epsilon/10^{3} d \right\} \right] \\
    &\leq \prob \left[ \left\{|P_{d, t_v, t'}^v| > \epsilon/10^{2} d \right\} \wedge \left\{ |S| \geq \epsilon/10^{3} d \right\} \right] + \prob \left[ \left\{|P_{d, t_v, t'}^v| > \epsilon/10^{2} d \right\} \wedge \mathcal{R} \wedge \left\{ |L| > 9 \epsilon/10^{3} d \right\} \right]\\
\end{align*}
We bound each of the two terms separately. For the first term: $\forall u \in S$ let $I_u$ be the sample constructed by $u$ at arrival and note that $u$ sends a $\mathrm{Type_0}$ notification to all nodes in $I_u$. Since $\forall u \in S \quad \ta{u} > t_v$ we have that $\{ \forall u \in S: v \not \in I_u \}$. Thus:
\begin{align*}
    \prob \left[ \left\{|P_{d, t_v, t'}^v| > \epsilon/10^{2} d \right\} \wedge \left\{ |S| \geq \epsilon/10^{3} d \right\} \right] &\leq 
    \prob \left[\left\{ \forall u \in S: v \not \in I_u\right\} \wedge \left\{ |S| \geq \epsilon/10^{3} d \right\} \right]\\
    &\leq \left( 1 - \nicefrac{1}{10^{3}d} \right)^{(\nicefrac{10^{10} \log n}{\epsilon}) \cdot (\nicefrac{\epsilon d}{10^{3}})}\\
     &\leq 1 / n^{10^{4}}
\end{align*}
Where in the second inequality we use that $\forall u , u' \in S$ events $\{  v \not \in I_u \}$ and $\{  v \not \in I_{u'} \}$ are independent.
\newcommand{\htu}{\widehat{t}_u}

We now turn our attention to the second term. Note that $\forall u \in L$ event $A_u^{\ta{u}}$ is always true and for each node $u \in L$
denote by $\htu \in (t_u, t']$ the last time when $  |\N{\htu}{u}| < 10^{2}d$.
\begin{align*}
    |\N{\ta{u}}{u} \setminus \N{\htu}{u}|
        &\geq |\N{\ta{u}}{u}| - |\N{\htu}{u}| \\
        &> 10^{3} d - 10^{2} d\\
        &= 9\cdot 10^{2} d
\end{align*}
where in the second inequality we used the definition of set $L$.

Note that we just argued that event $A_u^{\ta{u}} \land \overline{M^{\ta{u}, \htu}_u}$ is always true. Consequently $\mathcal{R}$ implies that $u \in L$ will keep getting notifications of $\mathrm{Type_0}$ until its degree is close enough to its degree at time $\htu$, in other words, there exists a time $\widetilde{t}_u \in (\ta{u}, \htu]$ such that $T^{u, \widetilde{t}_u}_0 \land M^{\widetilde{t}_u, \htu}_u $ is true. At time $\widetilde{t}_u$, $u$ receives a $\mathrm{Type_0}$ notification and its degree can be upper bounded as follows:
\begin{align*}
|\N{\widetilde{t}_u}{u}|&\leq |\N{\widetilde{t}_u}{u} \setminus \N{\htu}{v}| + |\N{\htu}{v}|\\
                    &< \epsilon / 10^{5} |\N{\widetilde{t}_u}{u} |+ 10^{2} d
\end{align*}
where we used the fact that $M^{\widetilde{t}_u, \htu}_u$ is true.
Thus, for $\epsilon$ small enough it holds that 
\begin{align*}
    |\N{\widetilde{t}_u}{u}| &<  \frac{10^{2} d}{1 - \epsilon / 10^{5}} \Rightarrow{}\\
    |\N{\widetilde{t}_u}{u}| &<  10^{3} d \\
\end{align*}
Similarly to the arguments when we were bounding the first term of the sum, $\forall u \in L$ let $I_u$ be the sample constructed by $u$ at time $\widetilde{t}_u$, and note that $u$ sends a $\mathrm{Type_1}$ notification to all nodes in $I_u$. Since $\forall u \in S \quad \widetilde{t}_u > t_v$ we have that $\{ \forall u \in L: v \not \in I_u \}$. 
\begin{align*}
    \prob \left[ \left\{|P_{d, t_v, t'}^v| > \epsilon/10^{2} d \right\} \wedge \mathcal{R} \wedge \left\{ |L| > 9 \epsilon/10^{3} d \right\} \right] &\leq 
    \prob \left[\left\{ \{ \forall u \in L: v \not \in I_u \}\right\} \wedge \left\{ |L| > 9 \epsilon/10^{3} d \right\} \right]\\
    &\leq \left( 1 - \nicefrac{1}{10^{3}d} \right)^{(\nicefrac{10^{10} \log n}{\epsilon}) \cdot (\nicefrac{9\epsilon d}{10^{3}})}\\
     &\leq  1 / n^{10^{3}}
\end{align*}

Combining the previous bounds, we conclude the proof of the lemmas as follows:
\begin{align*}
    \prob \left[ \left\{|P_{d, t_v, t'}^v| > \epsilon/10^{2} d \right\} \right] <  1 / n^{10^{4}} +  1 / n^{10^{3}}  +  1 / n^{10^{4} - 3} <  3 / n^{10^{3}}
\end{align*}
\end{proof}

We now define the following event which will be a crucial for our arguments in the rest of the section

\newcounter{importantdefinitioncounter}
\setcounter{importantdefinitioncounter}{\value{definition}}
\begin{restatable}{definition}{defeventnotmanylosesnotmanynewnodes}\label{def:event such after tu, nodes neither lose too many neighbors nor they get new neighbors of low degree}
\begin{align*}
    \mathcal{R} = 
        &\bigcap_{u \in C, t>t_u } M_u^{t_u, t}
        \bigcap_{u \in C, d\in \left[1, n \right]}
            \left\{
                |P_{d, t_v, \tc}^v| \leq \epsilon/10^{2} d 
            \right\}
\end{align*}
\end{restatable}{definition}
From~\cref{lem: When my degree drops I get notified},~\cref{lem: bound on Adu}, and using the union bound we get the following:
\begin{obs}\label{obs: R happens with high probability}
\begin{align*}
\prob \left[ \mathcal{R} \right] > 1 - 1/n^{10^{2}}    
\end{align*}
\end{obs}
In addition, by the definition of event $\mathcal{R}$:
\begin{obs}\label{obs: Ntu can only increase with high degree nodes}
    $\mathcal{R}$ implies that for all $v \in C$ and $ t \in (t_v, \tc]$:
    \begin{enumerate}
    \item  $|\N{t_v}{v} \setminus \N{t}{v} | \leq \epsilon/{10^{5}} |\N{t_v}{v} |$; and
    \item  $|\left( \N{t}{v} \setminus \N{t_v}{v}\right) \cap  \left\{ u : |\N{t}{u}| < 10^{2} |\N{t}{v}| \right\}| \leq \epsilon / 10^{2} |\N{t}{v}|$
\end{enumerate}
\end{obs}

As in the previous~\cref{sec: Finding dense clusters}, we will focus on how the neighborhood of a node $v \in \Canchor$ can change after the last time it participated in an ``interesting event'', i.e., time $t_v$. From~\cref{obs: Ntu can only increase with high degree nodes} we know that with high probability $v$ does not lose more than a small fraction of its neighborhood until time $\tc$. At the same time, its neighborhood may increase drastically, with many nodes of high degree. Thus, in an approximate sense, we have that $\Nt{t}{v} \supseteq \Nt{t_v}{v}$.
The next~\cref{lem: anchor set nodes do not change their neighborhood too much} and~\cref{cor: anchor set nodes do not change their neighborhood too much} argue that for all $t \in (t_v, \tc]$ such that $v$ is in agreement with another node (a new edge adjacent to $v$ may be added to our sparse solution in that case) we have that $\N{t}{v} \simeq \N{t_v}{v}$.


\begin{lem}\label{lem: anchor set nodes do not change their neighborhood too much}
    Let $u \in \C$ and $t > t_u$ a time when $u$ is in $\epsilon$-agreement with $v$, and either $u$ or $v$ is $\epsilon$-heavy. Then $\mathcal{R}$ implies that:
    $$
    (1 + 10 \epsilon) |\N{t_u}{u}| \geq |\N{t}{u}| \geq (1 - \epsilon/10^{5}) |\N{t_u}{u}|
    $$
\end{lem}
\begin{proof}
The right hand side is implied immediately by (1) of~\cref{obs: Ntu can only increase with high degree nodes}. We prove the left hand side for the case when $v$ is $\epsilon$-heavy and omit the case where $u$ is $\epsilon$-heavy as it is proven similarly. Since $v$ is $\epsilon$-heavy and in $\epsilon$-agreement with $u$, from~\ref{property: N_G(u) setminus C  <  3 epsilon N_G(u)} and~\ref{property: N_G(u) cap N_G(v) >= (1 - 5 epsilon) max(N_G(u), N_G(v)}, we can deduce that $u$ is $5 \epsilon$-heavy, i.e., it is in $5 \epsilon$-agreement with at least a $(1 - 5 \epsilon)$ fraction of its neighborhood at time $t$. Thus, again using~\ref{property: N_G(u) cap N_G(v) >= (1 - 5 epsilon) max(N_G(u), N_G(v)}, we have that at least $(1 - 5 \epsilon) |\N{t}{u}|$ neighboring nodes of $u$ have degree at most $\frac{|\N{t}{u}|}{(1 - 5 \epsilon)} < 2 |\Nt{t}{u}|$ for $\epsilon$ small enough. Due to (2) of~\cref{obs: Ntu can only increase with high degree nodes}, from those nodes at most $\epsilon/10^{2} |\N{t}{u}| $ of them could have arrived after time $t_v$. Consequently, $|\N{t}{u} \setminus \N{t_v}{u}| \leq \epsilon/10^{2} |\N{t}{u}| + 5 \epsilon |\N{t}{u}| <  6 \epsilon |\N{t}{u}|$. Using that $|\N{t}{u}| - |\N{t_u}{u}| \leq |\N{t}{u} \setminus \N{t_u}{u}|$ we conclude that:
\begin{align*}
    |\N{t}{u}| \leq \frac{|\N{t_u}{u}|}{1 - 6 \epsilon} < (1 + 10 \epsilon) |\N{t_u}{u}|
\end{align*}
where the second inequality holds for $\epsilon$ small enough.
\end{proof}

\begin{cor}\label{cor: anchor set nodes do not change their neighborhood too much}
   Let $u \in \C$ and $t \geq t_u$ a time when $u$ is in $\epsilon$-agreement with $v$, and either $u$ or $v$ is $\epsilon$-heavy. Then $\mathcal{R}$ implies that:
    \begin{enumerate}[ref=(\arabic*) of Corollary \thecor]
        \item  $|\N{t_u}{u} \setminus \N{t}{u}| \leq \epsilon/10^{5} |\N{t_u}{u}| <  \epsilon / 10^{4}  |\N{t}{u}|$; and \label{cor: Nt remains the same p1}
        \item $|\N{t}{u} \setminus \N{t_u}{u}| \leq 6 \epsilon |\N{t}{u}| < 7 \epsilon  |\N{t_u}{u}|$\label{cor: Nt remains the same p2}
    \end{enumerate}
\end{cor}

\begin{proof}
If $t = t_u$ then the claim trivially holds. In the following we assume that $t > t_u$.

The left hand side of the first inequality is immediate from (1) of~\cref{obs: Ntu can only increase with high degree nodes}. For the right hand side of the first inequality note that from~\cref{lem: anchor set nodes do not change their neighborhood too much} we have $|\N{t_u}{u}| < 1 / (1 - \epsilon/10^{5}) |\N{t}{u}|$. Thus, for $\epsilon$ small enough we get $(\epsilon/10^{5})\cdot(1 / (1 - \epsilon/10^{5})) < \epsilon / 10^{4} $.

For the left hand side of the second inequality it suffices to repeat the arguments of~\cref{obs: Ntu can only increase with high degree nodes} and for the right hand side, again using~\cref{obs: Ntu can only increase with high degree nodes}, and the fact that $6 \epsilon (1 + 10 \epsilon) < 7 \epsilon$ for $\epsilon$ small enough.
\end{proof}

The next~\cref{lem: anchor set nodes do not change their neighborhood too much properties} concentrates on how the neighborhood of a node $u \in \Canchor$ in our sparse solution changes after time $t_u$. At a high level, we argue that $\Nt{t_u}{u} \simeq \Nt{t}{u}$.

\begin{lem}\label{lem: anchor set nodes do not change their neighborhood too much properties}
    Let $u \in \Canchor$ and $t > t_u$. Then $\mathcal{R}$ implies that:
    \begin{enumerate}[ref=(\arabic*) of Lemma \thelem]
        \item $\Nt{t_u}{u} \subseteq \N{t_u}{u}$\label{lem: Nttu is stable p1}
        \item $|\N{t_u}{u} \setminus \Nt{t_u}{u}| < \epsilon |\N{t_u}{u}|$\label{lem: Nttu is stable p2}
        \item $|\Nt{t_u}{u}| \geq (1 - \epsilon) |\N{t_u}{u}|$\label{lem: Nttu is stable p3}
        \item $|\Nt{t_u}{u} \setminus \Nt{t}{u}| < \epsilon |\Nt{t_u}{u}|$\label{lem: Nttu is stable p4}
        \item $|\Nt{t}{u} \setminus \Nt{t_u}{u}| < 8 \epsilon |\N{t_u}{u}|$\label{lem: Nttu is stable p5}
        \item $|\Nt{t}{u} \setminus \Nt{t_u}{u}| < 12 \epsilon |\Nt{t_u}{u}|$\label{lem: Nttu is stable p6}
        \item $|\Nt{t}{u} | < (1 + 12 \epsilon) |\Nt{t_u}{u}|$\label{lem: Nttu is stable p7}
        \item $( 1 - \epsilon) |\Nt{t_u}{u} | < |\Nt{t}{u}|$\label{lem: Nttu is stable p8}
    \end{enumerate}
\end{lem}

\begin{proof}
We prove each statement as follows:
\begin{enumerate}
    \item The neighborhood of a node in our sparse solution is always a subset of its true neighborhood.
    \item Since $u \in \Canchor$, $u$ is $\epsilon$-heavy at time $t_u$ and in $\epsilon$-agreement with at least an $\epsilon$ fraction of its neighborhood at that time. All edges $(u, v)$ where $v\in\N{t_u}{u}$ is in $\epsilon$-agreement with $u$ are added to our sparse solution at that time.
    \item Again, due to $u$ being $\epsilon$-heavy at time $t_u$.
    \item Due to the $Clean(\cdot, \epsilon)$ procedure $u$ cannot lose more than an $\epsilon$ fraction of its neighborhood in our sparse solution at time $t_u$ and remain in the anchor set.
    \item W.l.o.g. we assume that time $t$ was the last time $u$'s neighborhood in the sparse solution increased, and let $v$ be its last new neighbor in the sparse solution. Note that since $t > t_u$, $v$ must be in $\epsilon$-agreement with $u$ at time $t$ and either $u$ or $v$ must be $\epsilon$-heavy. We have:
    \begin{align*}
        |\Nt{t}{u} \setminus \Nt{t_u}{u}| &\leq |\Nt{t}{u} \setminus \N{t_u}{u}| + |\N{t_u}{u} \setminus \Nt{t_u}{u}|\\
        &\leq |\N{t}{u} \setminus \N{t_u}{u}| + |\N{t_u}{u} \setminus \Nt{t_u}{u}|\\
        &< 7 \epsilon |\N{t_u}{u}| + \epsilon|\N{t_u}{u}|\\
        &\leq 8 \epsilon |\N{t_u}{u}|\\
    \end{align*}
    Where in the second inequality we used the fact that $\Nt{t}{u} \subseteq \N{t}{u}$ is always true, and in the third we used both (2) of the current lemma and~\ref{cor: Nt remains the same p2}.
    \item It is immediate by combining (3) and (5), since for $\epsilon$ small enough it holds that $8 \epsilon / (1 - \epsilon) < 12 \epsilon$.
    \item It is immediate from (6).
    \item It is immediate from (4).
\end{enumerate}
\end{proof}

The following~\cref{lem: anchor set nodes have a high overlap if the interesection is nonempty}, ~\cref{lem: two anchor sets have a non empty intersection in our sparse solution} and~\cref{cor: main corolarry lower bounding the interesection of the neighborhoods between two anchor set nodes in the sparse solution} pave the road to the main~\cref{thm: C is dense} of the current section by arguing that for any two nodes $u, v \in \Canchor$ at time $t \geq \max \{t_u, t_v \} $ their neighborhood in the sparse solution is very similar, i.e., $\Nt{t}{u} \simeq \Nt{t}{v} $. To that end~\cref{lem: anchor set nodes have a high overlap if the interesection is nonempty} proves that if $\Nt{t}{u} \cap  \Nt{t}{v} \neq \emptyset$ then $\Nt{t}{u} \simeq \Nt{t}{v}$ ,~\cref{lem: two anchor sets have a non empty intersection in our sparse solution} continues arguing that $\Nt{t}{u} \cap  \Nt{t}{v}$ is always non-empty and~\cref{cor: main corolarry lower bounding the interesection of the neighborhoods between two anchor set nodes in the sparse solution} concludes that indeed $\Nt{t}{u} \simeq \Nt{t}{v} $.

Before proceeding to~\cref{lem: anchor set nodes have a high overlap if the interesection is nonempty} we state a useful set inequality.
\begin{ineq}\label{ineq: set inequality}
     Let $A, B, C$ and $D$ be four sets, then:
     \begin{align*}
             |A \cap B| \geq |C \cap D| - |C \setminus A| - |D \setminus B|
     \end{align*}
 \end{ineq}
\begin{proof}
     Let $s \in C \cap D$. It is enough to prove that the latter implies: $s \in (A \cap B) \cup (C \setminus A) \cup (D \setminus B) $. We consider the following cases:
\begin{enumerate}
    \item If $s$ is in $A \cap B$, the claim is satisfied.
    \item If $s$ is not in $A \cap B$, then either $s$ is not in set $A$ or $s$ is not in set $B$.
    \begin{enumerate}
        \item If $s$ is not in $A$, then $s$ must be in $C \setminus A$.
        \item If $s$ is not in $B$, then $s$ must be in $D \setminus B$.
    \end{enumerate}
\end{enumerate}
\end{proof}

\begin{lem}\label{lem: anchor set nodes have a high overlap if the interesection is nonempty}
    Let $u, v \in \Canchor$ and $t' \geq \max \{ t_u, t_v\}$  such that $\Nt{t'}{u} \cap \Nt{t'}{v}$ is non-empty. Then $\mathcal{R}$ implies that: $|\Nt{t'}{u} \cap \Nt{t'}{v}| \geq (1 -  80 \epsilon)  \max \{ |\Nt{t'}{u}|, |\Nt{t'}{v}| \}$. 
\end{lem}

\begin{proof}
W.l.o.g. we assume that $t_v \geq t_u $, i.e., $v$ entered the anchor set after $u$. Let $t \geq t_v$ be the minimum time such that $\Nt{t}{u} \cap \Nt{t}{v} \neq \emptyset$. We distinguish between the following two scenarios
\begin{enumerate}
    \item At time $t$ a node $w$ participates in an ``interesting'' event and gets connected to both nodes $u$ and $v$; or
    \item At time $t$, $v$ participates in an ``interesting event'' and gets connected to a node $w$ which was already connected to node $u$. Note that in this case we have $t = t_v$.
\end{enumerate}

We prove that in both scenarios $u$ and $v$ have a very large neighborhood overlap in our sparse solution. 

\paragraph{Case 1} We have that $u$ and $v$ are both in $\epsilon$-agreement with $w$ and that either $w$ is $\epsilon$-heavy or both $u$ and $v$ are $\epsilon$-heavy. From the latter observation and using~\ref{property: N_G(u) cap N_G(v) >= (1 - 5 epsilon) max(N_G(u), N_G(v)} we have that:
\begin{align}\label{eq: of the current proof}
    |\N{t}{u} \cap \N{t}{v}| \geq (1 -  5 \epsilon) \max \{ |\N{t}{u}|, |\N{t}{v}| \}
\end{align}
that is $\N{t}{u} \simeq \N{t}{v}$.

We proceed arguing that $\N{t_u}{u} \simeq \N{t_v}{v}$.

\begin{align*}
|\N{t_u}{u} \cap \N{t_v}{v}|
    &\geq |\N{t}{u} \cap \N{t}{v}| - |\N{t}{u} \setminus \N{t_u}{u}| - |\N{t}{v} \setminus \N{t_v}{v}|\\
    &\geq |\N{t}{u} \cap \N{t}{v}| - 7 \epsilon |\N{t_u}{u}| - 7 \epsilon |\N{t_v}{v}|\\
    &\geq (1 -  5 \epsilon) \max \{ |\N{t}{u}|, |\N{t}{v}| \} - 7 \epsilon |\N{t_u}{u}| - 7 \epsilon |\N{t_v}{v}|\\
    &\geq (1 -  5 \epsilon) (1 - \epsilon/10^{5}) \max \{ |\N{t_u}{u}|, |\N{t_v}{v}| \} - 7 \epsilon |\N{t_u}{u}| - 7 \epsilon |\N{t_v}{v}|\\
    &\geq (1 -  20 \epsilon) \max \{ |\N{t_u}{u}|, |\N{t_v}{v}| \}\\
\end{align*}
Where in the first inequality we use~\cref{ineq: set inequality}, in the second we use~\ref{cor: Nt remains the same p2}, in the third we use~\cref{eq: of the current proof} of the current proof, in the fourth~\cref{lem: anchor set nodes do not change their neighborhood too much} and the last one holds for $\epsilon$ small enough.

In the same manner we proceed arguing that $\Nt{t_u}{u} \simeq \Nt{t_v}{v}$ by lower bounding $|\Nt{t_u}{u} \cap \Nt{t_v}{v}|$ as follows:
\begin{align*}
|\Nt{t_u}{u} \cap \Nt{t_v}{v}|
    &\geq |\N{t_u}{u} \cap \N{t_v}{v}| - |\N{t_u}{u} \setminus \Nt{t_u}{u}| - |\N{t_v}{v} \setminus \Nt{t_v}{v}|\\
    &\geq|\N{t_u}{u} \cap \N{t_v}{v}|- \epsilon |\N{t_u}{u}| - \epsilon |\N{t_v}{v}|\\
    &\geq|\N{t_u}{u} \cap \N{t_v}{v}| - \epsilon/ (1 - \epsilon) |\Nt{t_u}{u}| - \epsilon/ (1 - \epsilon) |\Nt{t_v}{v}|\\
    &\geq (1 -  20 \epsilon) \max \{ |\N{t_u}{u}|, |\N{t_v}{v}| \} - \epsilon/ (1 - \epsilon) |\Nt{t_u}{u}| - \epsilon/ (1 - \epsilon) |\Nt{t_v}{v}|\\
    &\geq (1 -  20 \epsilon) \max \{ |\N{t_u}{u}|, |\N{t_v}{v}| \} - 3 \epsilon |\Nt{t_u}{u}| - 3 \epsilon |\Nt{t_v}{v}|\\
    &\geq (1 -  20 \epsilon) (1 - \epsilon) \max \{ |\Nt{t_u}{u}|, |\Nt{t_v}{v}| \} - 3 \epsilon |\Nt{t_u}{u}| - 3 \epsilon |\Nt{t_v}{v}|\\
    &\geq (1 -  30 \epsilon) \max \{ |\Nt{t_u}{u}|, |\Nt{t_v}{v}| \} \\
\end{align*}

Where in the first inequality we use~\cref{ineq: set inequality}, in the second inequality we use~\ref{lem: Nttu is stable p2}, in the third inequality we use~\ref{lem: Nttu is stable p3}, in the fourth one we used the lower bound on $|\N{t_u}{u} \cap \N{t_v}{v}|$ that we just proved in the current lemma, in the fifth and the seventh holds for $\epsilon$ small enough and, in the sixth inequality we use again~\ref{lem: Nttu is stable p3}, and the fifth and seventh inequality hold for $\epsilon$ small enough.

We can now finish the first case by lower bounding the $|\Nt{t'}{u} \cap \Nt{t'}{v}|$.
\begin{align*}
|\Nt{t'}{u} \cap \Nt{t'}{v}|
    &\geq |\Nt{t_u}{u} \cap \Nt{t_v}{v}| - |\Nt{t_u}{u} \setminus \Nt{t'}{u}| - |\Nt{t_v}{v} \setminus \Nt{t'}{v}|\\
    &\geq |\Nt{t_u}{u} \cap \Nt{t_v}{v}| - \epsilon |\Nt{t_u}{u}| -  \epsilon |\Nt{t_v}{v}|\\
    &\geq (1 -  30 \epsilon) \max \{ |\Nt{t_u}{u}|, |\Nt{t_v}{v}| \} - \epsilon |\Nt{t_u}{u}| -  \epsilon |\Nt{t_v}{v}||\\
    &\geq \frac{1 -  30 \epsilon}{1 + 12 \epsilon} \max \{ |\Nt{t'}{u}|, |\Nt{t'}{v}| \} - \frac{\epsilon}{1 - \epsilon} |\Nt{t'}{u}| -  \frac{\epsilon}{1 - \epsilon} |\Nt{t'}{v}|\\
    &\geq (1 -  60 \epsilon)  \max \{ |\Nt{t'}{u}|, |\Nt{t'}{v}| \}\\
\end{align*}
Where in the first inequality we use~\cref{ineq: set inequality}, in the second we use~\ref{lem: Nttu is stable p4}, in the third we use the lower bound on $\Nt{t_u}{u} \cap \Nt{t_v}{v}|$ that we proved in the current lemma, in the fourth inequality we use both~\ref{lem: Nttu is stable p7} and~\ref{lem: Nttu is stable p8}, and the last inequality holds for $\epsilon$ small enough.

\paragraph{Case 2} 
We now turn our attention to the second case.
In that case w.l.o.g. we assume that $u$ was connected to $w$ in our sparse solution before time $t$ and that at time $t$, $v$ was inserted in the anchor set and got connected with node $w$ with which they are in $\epsilon$-agreement at time $t$. Consequently, $|\N{t}{v} \cap \N{t}{w}| \geq (1 - \epsilon) \max \{ |\N{t}{v}|, |\N{t}{w}|\} $.

We can also assume that $w$ does not participate in an ``interesting event'' at time $t$ since this situation is already covered by the first case.

Let $t'' = \max \{t_u, t_w\}$ and note that at that time $u$ and $w$ were in $\epsilon$-agreement and one of them is $\epsilon$-heavy, thus we have that $|\N{t''}{u} \cap \N{t''}{w}| \geq (1 - \epsilon) \max \{ |\N{t''}{u}|, |\N{t''}{w}| \} $

The goal is to argue that $\Nt{t}{u} \approx \Nt{t}{v}$, and after that use~\cref{lem: anchor set nodes do not change their neighborhood too much properties} to deduce that $\Nt{t'}{u} \approx \Nt{t'}{v}$.

Towards this goal, we initially prove that $\N{t''}{w} \approx \N{t}{w}$. Note that in both times, $t$ and $t''$ node $w$ was in $\epsilon$-agreement with another node and got connected to it.
We distinguish between two cases:
\begin{enumerate}[label=(\alph*)]
    \item $t'' = t_u > t_w $; and
    \item $t '' = t_w$
\end{enumerate}
In (a) we prove that $\N{t''}{w} \simeq \N{t_w}{w}$ and $\N{t}{w} \simeq \N{t_w}{w}$ to conclude that $\N{t''}{w} \approx \N{t}{w}$.
\begin{align*}
    |\N{t''}{w} \cap \N{t_w}{w}| 
    &\geq \max \{ |\N{t''}{w}|, |\N{t_w}{w}|\} - |\N{t_w}{w} \setminus \N{t''}{w}| - |\N{t''}{w} \setminus \N{t_w}{w}|\\
    &\geq \max \{ |\N{t''}{w}|, |\N{t_w}{w}|\} - (7 \epsilon + \epsilon/10^{4}) \max \{ |\N{t''}{w}|, |\N{t_w}{w}|\}\\
    &\geq (1 - 8 \epsilon) \max \{ |\N{t''}{w}|, |\N{t_w}{w}|\}\\
\end{align*}
Where in the second inequality we use~\ref{cor: Nt remains the same p1} and~\ref{cor: Nt remains the same p2}, and the third inequality holds for $\epsilon$ small enough.

Similarly, we have 
\begin{align*}
    |\N{t}{w} \cap \N{t_w}{w}| 
    &\geq \max \{ |\N{t}{w}|, |\N{t_w}{w}|\} - |\N{t_w}{w} \setminus \N{t}{w}| - |\N{t}{w} \setminus \N{t_w}{w}|\\
    &\geq \max \{ |\N{t}{w}|, |\N{t_w}{w}|\} - (7 \epsilon + \epsilon/10^{4}) \max \{ |\N{t''}{w}|, |\N{t_w}{w}|\}\\
    &\geq (1 - 8 \epsilon) \max \{ |\N{t}{w}|, |\N{t_w}{w}|\}\\
\end{align*}

We now argue that $\N{t''}{w} \simeq \N{t}{w}$ as follows:
\begin{align*}
    |\N{t}{w} \cap \N{t''}{w}| 
    &\geq |\N{t_w}{w}| - |\N{t_w}{w} \setminus \N{t}{w}| - |\N{t_w}{w} \setminus \N{t''}{w}|\\
    &\geq (1 - 8 \epsilon)\max \{ |\N{t}{w}|, |\N{t''}{w}|\} - (8 \epsilon + 8 \epsilon) \max \{ |\N{t}{w}|, |\N{t''}{w}|\}\\
    &\geq (1 - 24 \epsilon) \max \{ |\N{t}{w}|, |\N{t''}{w}|\}\\
\end{align*}

Case (b) is simpler and we can immediately use~\ref{cor: Nt remains the same p1} and~\ref{cor: Nt remains the same p2} to argue that $\N{t''}{w} \simeq \N{t}{w}$, as follows:
\begin{align*}
    |\N{t}{w} \cap \N{t''}{w}| 
    &\geq \max \{ |\N{t}{w}|, |\N{t''}{w}|\} - |\N{t}{w} \setminus \N{t''}{w}| - |\N{t''}{w} \setminus \N{t}{w}|\\
     &\geq \max \{ |\N{t}{w}|, |\N{t''}{w}|\} -  (7 \epsilon + \epsilon/10^{4}) \max \{ |\N{t''}{w}|, |\N{t_w}{w}|\}\\
    &\geq (1 - 8 \epsilon) \max \{ |\N{t}{w}|, |\N{t''}{w}|\}
\end{align*}

Where the third inequality holds for $\epsilon$ small enough.

Thus, in both cases (a) and (b) we have that:

\begin{align*}
    |\N{t}{w} \cap \N{t''}{w}| \geq (1 - 24 \epsilon) \max \{ |\N{t}{w}|, |\N{t''}{w}|\}
\end{align*}

We continue arguing that $\N{t''}{u} \approx \N{t}{v}$.
Indeed we have 

\begin{align*}
|\N{t}{v} \cap \N{t''}{u}|
    &\geq |\N{t}{w} \cap \N{t''}{w} | - |\N{t}{w} \setminus \N{t}{v}| - |\N{t''}{w} \setminus \N{t''}{u}|\\
    &\geq (1 - 24 \epsilon) \max \{ |\N{t}{w}|, |\N{t''}{w}|\} - \epsilon |\N{t}{w}| - \epsilon |\N{t''}{w}|\\
    &\geq (1 - 26 \epsilon) \max \{ |\N{t}{w}|, |\N{t''}{w}|\}\\
    &\geq \frac{1 - 26 \epsilon}{1 - \epsilon} \max \{ |\N{t}{v}|, |\N{t''}{u}|\}\\
    &\geq (1 - 27 \epsilon) \max \{ |\N{t}{v}|, |\N{t''}{u}|\}\\
\end{align*}

Where in the second inequality we used the lower bound on $|\N{t}{w} \cap \N{t''}{w}| $ that we proved in the current lemma and the fact that at times $t$ and $t''$, $w$ is in $\epsilon$-agreement with nodes $v$ and $u$ respectively and got connected to them in our sparse solution. In the fourth inequality we again used the latter fact and the last inequality holds for $\epsilon$ small enough.

We now argue why $\N{t_u}{u} \approx \N{t_v}{v}$.  Note that $t = t_v$ and that $u$ does not participate in an ``interesting event'' after $t''$.
\begin{align*}
|\N{t_v}{v} \cap \N{t_u}{u}|
    &\geq |\N{t}{v} \cap \N{t''}{u} |  - |\N{t}{v} \setminus \N{t_v}{v}| - |\N{t''}{u} \setminus \N{t_u}{u}|\\
    &= |\N{t}{v} \cap \N{t''}{u} |  - |\N{t''}{u} \setminus \N{t_u}{u}|\\
    &\geq |\N{t}{v} \cap \N{t''}{u} |  - 6 \epsilon |\N{t''}{u}|\\
    &\geq (1 - 6 \epsilon)|\N{t}{v} \cap \N{t''}{u} |\\
    &\geq (1 - 6 \epsilon)(1 - 27 \epsilon) \max \{ |\N{t}{v}|, |\N{t''}{u}|\}\\
    &\geq (1 - \epsilon/10^{5}) (1 - 6 \epsilon)(1 - 27 \epsilon) \max \{ |\N{t}{v}|, |\N{t_u}{u}|\}\\
    &= (1 - \epsilon/10^{5}) (1 - 6 \epsilon)(1 - 27 \epsilon) \max \{ |\N{t_v}{v}|, |\N{t_u}{u}|\}\\
    &\geq (1 -  40 \epsilon) \max \{ |\N{t_v}{v}|, |\N{t_u}{u}|\}
\end{align*}
where in the first inequality we use~\cref{eq: of the current proof},
the second inequality holds since either $t'' > t_u$ and we use~\ref{cor: Nt remains the same p2} or $t'' = t_u$ and in that case $|\N{t''}{u} \setminus \N{t_u}{u}| = 0$, in the fourth inequality we use the lower bound on $|\N{t}{v} \cap \N{t''}{u} |$ that we proved in the current lemma, in the fifth inequality we use~\cref{lem: anchor set nodes do not change their neighborhood too much}, and the last inequality holds for $\epsilon$ small enough.

Using what we just proved, i.e., $|\N{t_v}{v} \cap \N{t_u}{u}| \geq (1 -  40 \epsilon) \max \{ |\N{t_v}{v}|, |\N{t_u}{u}|\} $, we now lower bound $|\Nt{t_u}{u} \cap \Nt{t_v}{v}|$ using the same arguments as in the first case.

\begin{align*}
|\Nt{t_u}{u} \cap \Nt{t_v}{v}|
    &\geq |\N{t_u}{u} \cap \N{t_v}{v}| - |\N{t_u}{u} \setminus \Nt{t_u}{u}| - |\N{t_v}{v} \setminus \Nt{t_v}{v}|\\
    &\geq|\N{t_u}{u} \cap \N{t_v}{v}|- \epsilon |\N{t_u}{u}| - \epsilon |\N{t_v}{v}|\\
    &\geq|\N{t_u}{u} \cap \N{t_v}{v}| - \epsilon/ (1 - \epsilon) |\Nt{t_u}{u}| - \epsilon/ (1 - \epsilon) |\Nt{t_v}{v}|\\
    &\geq (1 -  40 \epsilon) \max \{ |\N{t_u}{u}|, |\N{t_v}{v}| \} - \epsilon/ (1 - \epsilon) |\Nt{t_u}{u}| - \epsilon/ (1 - \epsilon) |\Nt{t_v}{v}|\\
    &\geq (1 -  40 \epsilon) \max \{ |\N{t_u}{u}|, |\N{t_v}{v}| \} - 3 \epsilon |\Nt{t_u}{u}| - 3 \epsilon |\Nt{t_v}{v}|\\
    &\geq (1 -  40 \epsilon) (1 - \epsilon) \max \{ |\Nt{t_u}{u}|, |\Nt{t_v}{v}| \} - 3 \epsilon |\Nt{t_u}{u}| - 3 \epsilon |\Nt{t_v}{v}|\\
    &\geq (1 -  50 \epsilon) \max \{ |\Nt{t_u}{u}|, |\Nt{t_v}{v}| \} \\
\end{align*}

And conclude the proof of the second case, in a similar manner to the first case, that is, proving $\Nt{t}{u} \approx \Nt{t}{v}$.

\begin{align*}
|\Nt{t'}{u} \cap \Nt{t'}{v}|
    &\geq |\Nt{t_u}{u} \cap \Nt{t_v}{v}| - |\Nt{t'}{u} \setminus \Nt{t_u}{u}| - |\Nt{t'}{v} \setminus \Nt{t_v}{v}|\\
    &\geq |\Nt{t_u}{u} \cap \Nt{t_v}{v}| - 12 \epsilon |\Nt{t_u}{u}| - 12 \epsilon |\Nt{t_v}{v}|\\
    &\geq (1 -  50 \epsilon) \max \{ |\Nt{t_u}{u}|, |\Nt{t_v}{v}| \} - 12 \epsilon |\Nt{t_u}{u}| - 12 \epsilon |\Nt{t_v}{v}|\\
    &\geq (1 -  74 \epsilon) / (1 + 12 \epsilon) \max \{ |\Nt{t'}{u}|, |\Nt{t'}{v}| \}\\
    &\geq (1 -  80 \epsilon)  \max \{ |\Nt{t'}{u}|, |\Nt{t'}{v}| \}\\
\end{align*}
 Where in the second inequality we use~\ref{lem: Nttu is stable p6}, in the third inequality we use the lower bound on $\Nt{t_u}{u} \cap \Nt{t_v}{v}|$ that we proved in the current lemma, in the fourth inequality we use~\ref{lem: Nttu is stable p7} and the last inequality holds for $\epsilon$ small enough.

\end{proof}

\begin{lem}\label{lem: two anchor sets have a non empty intersection in our sparse solution}
    $\mathcal{R}$ implies that $\forall u, v \in \Canchor$: $\Nt{\tc}{u} \cap \Nt{\tc}{v} \neq \emptyset$.
\end{lem}

\begin{proof}
The proof follows very similar arguments to the~\cite{DBLP:conf/icml/Cohen-AddadLMNP21} that we repeat here for completeness.
Let $d(x, y)$ be the distance of two nodes $x, y$ in our sparse solution at time $\tc$.
Suppose the lemma is not true, then we have that $d(u, v) > 2$.
Towards a contradiction also assume $u, v$ to be the minimum distance nodes in $\Canchor$ such that $d(u, v) > 2$.
Note that for any edge in our sparse solution $\widetilde{G}_{\tc}$, one of its endpoints is in $\Canchor$. If $d(u, v) \geq 5$, let $P = \langle u, u', u'', \ldots, v \rangle$ be the shortest $u$-$v$ path in $\widetilde{G}_{\tc}$, thus either $u'$ or $u''$ must be in $\Canchor$, contradicting the minimality assumption regarding the distance between $u$ and $v$.
If $d(u, v) \leq 4$, then either $d(u, v) = 4$ or $d(u, v) = 3$. We end up in a contradiction in both cases:
\begin{enumerate}
    \item If  $d(u, v) = 4$ then $\exists v', w, u'$ such that $v' \in \Nt{\tc}{v} \cap \Nt{\tc}{w}$, $w \in \Canchor$ and $u' \in \Nt{\tc}{w} \cap \Nt{\tc}{u}$. 
    \item If  $d(u, v) = 3$ then $\exists w, u'$ such that $w \in \Nt{\tc}{v} \cap \Nt{\tc}{u'}$, $u' \in \Nt{\tc}{w} \cap \Nt{\tc}{u}$ and either $w$ or $u$ are in $\Canchor$. Assume that $w \in \Canchor$ (and the case where $u' \in \Canchor$ is similar)
\end{enumerate}
In both cases, from~\cref{lem: anchor set nodes have a high overlap if the interesection is nonempty} we have:
\begin{align*}
|\Nt{\tc}{u} \cap \Nt{\tc}{v}|
    &\geq |\Nt{\tc}{w}| - |\Nt{\tc}{w} \setminus \Nt{\tc}{u}| - |\Nt{\tc}{w} \setminus \Nt{\tc}{v}|\\
    &\geq |\Nt{\tc}{w}| - 80 \epsilon  |\Nt{\tc}{u}| - 80 \epsilon  |\Nt{\tc}{v}| \\
    &\geq (1 - 80 \epsilon) \max \{|\Nt{\tc}{u}|, |\Nt{\tc}{v}| \} - 80 \epsilon  |\Nt{\tc}{u}| - 80 \epsilon  |\Nt{\tc}{v}| \\
    &\geq (1 - 240 \epsilon) \max \{|\Nt{\tc}{u}|, |\Nt{\tc}{v}| \}\\
    &> 0
\end{align*}
Where the last inequality holds for $\epsilon$ small enough.
Thus, we ended in a contradiction.
\end{proof}

\begin{cor}\label{cor: main corolarry lower bounding the interesection of the neighborhoods between two anchor set nodes in the sparse solution}
    Let $u, v \in \Canchor$, $\mathcal{R}$ implies that $|\Nt{\tc}{u} \cap \Nt{\tc}{v}| \geq (1 -  80 \epsilon)  \max \{ |\Nt{\tc}{u}|, |\Nt{\tc}{v}| \}$.
\end{cor}
\begin{proof}
    It is immediate from~\cref{lem: anchor set nodes have a high overlap if the interesection is nonempty} and~\cref{lem: two anchor sets have a non empty intersection in our sparse solution}.
\end{proof}

Before proceeding to the last theorems of this section we state some useful set inequalities.
\begin{obs}\label{lem: min implies max implies min set inequalities}
     For sets $A, B$, $C$ and positive reals $\alpha, \beta$: 
    \begin{enumerate}
    \item If $|A \triangle B| \leq \epsilon \max \{|A|, |B| \}$ then $|A \cap B| \geq (1 - \epsilon) \max \{|A|, |B| \}$.
    \item If $|A \cap B| \geq (1 - \epsilon) \max \{|A|, |B| \}$ then for $\epsilon$ small enough:\\
    $$|A \triangle B| \leq 2 \epsilon \max \{|A|, |B| \} \leq \frac{2 \epsilon}{1 - \epsilon} \min \{|A|, |B| \} \leq 3 \epsilon \min \{|A|, |B| \}$$
    \item If $|A \triangle B| \leq \epsilon \max \{|A|, |B| \}$ then for $\epsilon$ small enough $|A \triangle B| \leq 3 \epsilon \min \{|A|, |B| \}$
    \item $|A \triangle C| \leq |A \triangle B| + |B \triangle C|$.
    \item $\alpha \min \{|A|, |B| \} + \beta \min \{|B|, |C| \} \leq (\alpha + \beta ) \max \{ |A|, |C|\}$
\end{enumerate}
\end{obs}

\newcommand{\Nh}[2]{N_{\widehat{G}_{#1}}(#2) } 

We now proceed proving the main~\cref{thm: C is dense} of this section. Note that on~\cref{thm: C is dense} we also argue that $\forall v \in C$ their degree when they last participated in an ``interesting event'' is lower bounded by the $|C|/2$. The latter property will be crucial in~\cref{sec: Runtime analysis}.

\begin{thm}\label{thm: C is dense}
$\mathcal{R}$ implies that cluster $C$ is dense and $\forall v \in C$:
    $$
    |\N{\tc}{v} \cap C| \geq (1 - 541080 \epsilon) |C|
    $$
    and 
    $$
    |\N{t_u}{v}| \geq  | C| /2.
    $$
\end{thm}

\begin{proof}
\newcommand{\htv}{\widehat{t}_v}
\newcommand{\htu}{\widehat{t}_u}
\newcommand{\htw}{\widehat{t}_w}

Note that for node $v \in \Cnotanchor$ its neighborhood in $\widetilde{G}_{\tc}$ is a subset of $\Canchor$. To relate the neighborhood of $v$ inside $C$ in $G_{\tc}$ we construct an auxiliary graph which acts as a bridge between $\widetilde{G}_{\tc}$ and ${G}_{\tc}$. 
We construct a graph $\widehat{G}_{\tc}$ so that:
\begin{enumerate}
    \item $\forall u \in C$ and $u' \in \Canchor$ $   \Nh{\tc}{u} \simeq  \Nt{\tc}{u'}$;
    \item $\widehat{G}_{\tc} \subseteq  {G}_{\tc}$;
    \item $C$ is a connected component in $\widehat{G}_{\tc}$; and
    \item $\forall u \in C$ $|\Nh{\tc}{v} \cap C| \geq (1 - 541080 \epsilon) |C|$
\end{enumerate}
The existence of such a graph suffices to demonstrate the current theorem. For each $v \in \Cnotanchor$ let $\htv$ be the last time before $\tc$ that $v$ was in $\e$-agreement with a node $u \in \Canchor$, one of either $u$ or $v$ was heavy and edge $(u, v)$ is added to the sparse solution and remains until time $\tc$. Note that $\htv \geq t_v$, otherwise we would have that $v \not \in C$ and $\htv \geq t_u$ since edge $(u, v)$ remains in our sparse solution until time $\tc$.

We now prove that the neighborhood $\N{\htv}{v}$ is almost the same to the neighborhood of $\Nt{\tc}{u'}$ for any $u' \in \Canchor$. We first derive the following inequalities:
\begin{enumerate}[label=(\alph*)]
    \item $|\N{\htv}{v} \triangle \N{\htv}{u}| \leq \epsilon \max \{|\N{\htv}{v}|, |\N{\htv}{u}| \} \leq 3\epsilon \min \{|\N{\htv}{v}|, |\N{\htv}{u}| \}$ 
    \item $|\N{\htv}{u} \triangle \N{t_u}{u}| \leq 8 \epsilon \max \{|\N{\htv}{u}|, |\N{t_u}{u}| \} \leq 24\epsilon \min \{|\N{\htv}{u}|, |\N{t_u}{u}| \}$ 
    \item $|\N{t_u}{u} \triangle \Nt{t_u}{u}| \leq  \epsilon \max \{|\N{t_u}{u}|, |\Nt{t_u}{u}| \} \leq 3\epsilon \min \{|\N{t_u}{u}|, |\Nt{t_u}{u}| \}$ 
    \item $|\Nt{t_u}{u} \triangle \Nt{\tc}{u}| \leq  13 \epsilon \max \{|\Nt{t_u}{u}|, |\Nt{\tc}{u}| \} \leq 39\epsilon \min \{|\Nt{t_u}{u}|, |\Nt{\tc}{u}| \}$
    \item $|\Nt{\tc}{u} \triangle \Nt{\tc}{u'}| \leq  160 \epsilon \max \{|\Nt{\tc}{u}|, |\Nt{\tc}{u'}| \} \leq 240\epsilon \min \{|\Nt{\tc}{u}|, |\Nt{\tc}{u'}| \}$
\end{enumerate}

The right hand side of each inequality in (a), (b), (c), (d) and (e) uses~\cref{lem: min implies max implies min set inequalities} and the left hand side of (a) uses that $u$ and $v$ are in $\e$-agreement at time $\htv$, (b) uses~\cref{cor: anchor set nodes do not change their neighborhood too much}, (c) uses~\ref{lem: Nttu is stable p1} and~\ref{lem: Nttu is stable p2}, (d) uses~\ref{lem: Nttu is stable p4} and~\ref{lem: Nttu is stable p6} and (e) uses~\cref{cor: main corolarry lower bounding the interesection of the neighborhoods between two anchor set nodes in the sparse solution}.
By using the triangle inequality of~\cref{lem: min implies max implies min set inequalities} iteratively we can conclude that:
\begin{align*}
    |\N{\htv}{v} \cap \Nt{\tc}{u'}| \geq (1- 1113 \epsilon) \max \{|\N{\htv}{v}|, ||\Nt{\tc}{u'}| \}
\end{align*}

With the same line of reasoning for a node $u \in \Canchor$ we can define $\htu$ as the last time that an edge $(u, v)$ where $v \in C$ was added to our sparse solution and remained until at least time $\tc$. Inequalities (b), (c), (d) and (e) remain true and following the same arguments we conclude that:

$\forall u, u' \in \Canchor$:
\begin{align*}
    |\N{\htu}{u}  \cap \Nt{\tc}{u'}| \geq (1- 1113 \epsilon) \max \{|\N{\htu}{u} |, ||\Nt{\tc}{u'}| \}
\end{align*}

Consequently, $\forall u \in C$ and $u' \in \Canchor$:

\begin{equation*}
    |\N{\htu}{u}  \cap \Nt{\tc}{u'}| \geq (1- 1113 \epsilon) \max \{|\N{\htu}{u} |, |\Nt{\tc}{u'}| \} \tag{*}\label{eq: important}
\end{equation*}

Let $d = \max \{  |\Nt{\tc}{u'}|: u' \in \Canchor \}$. Using~\cref{eq: important} we can conclude that for $\e$ small enough:
\begin{equation*}
   \frac{d}{2}\leq (1- 1113 \epsilon) d \leq  |\N{\htu}{u}| \leq \frac{d}{(1- 1113 \epsilon)} \leq 2d \tag{**}\label{eq: second important}
\end{equation*}


Also note that since $\Nt{\tc}{u'} \subseteq C$ we also have that $\forall v \in \Cnotanchor$ and $u' \in \Canchor$:
\begin{align*}
    |(\N{\htv}{v} \cap C) \cap \Nt{\tc}{u'}| \geq (1- 1113 \epsilon) \max \{|\N{\htv}{v} \cap C|, ||\Nt{\tc}{u'}| \}\tag{***}\label{eq: third important}
\end{align*}
 and using the same arguments as in the case of $ |\N{\htu}{u}|$ we can also lower and upper bound $|\N{\htv}{v} \cap C|$ as follows: 
 \begin{equation*}
   \frac{d}{2}\leq (1- 1113 \epsilon) d \leq  |\N{\htv}{v} \cap C| \leq \frac{d}{(1- 1113 \epsilon)} \leq 2d \tag{****}\label{eq: fourth important}
\end{equation*}

    To facilitate the description of the intermediate graph $\widehat{G}_{\tc} $ for each $u \in C$ let $E_u = \{(u, v): v \in  \N{\htu}{u} \cap C \}$ be the set of edges between $u$ and nodes in $\N{\htu}{u} \cap C $.
    Note that all these edges also exist in $G_{\tc}$. We define $\widehat{G}_{\tc} = (C, \bigcup_{u \in C} E_u ) $ to be the graph with set of nodes to be $C$ and edges $\bigcup_{u \in C} E_u$. Let $N_u$ be the neighborhood of $u$ in that graph, then:
    \begin{enumerate}
        \item $N_u \supseteq \N{\htu}{u} \cap C$
        \item $\widehat{G}_{\tc} \subseteq  {G}_{\tc}$
        \item $N_u =  (\N{\htu}{u} \cap C) \cup \{ v\in C: \{\ta{v} > \htu\} \land \{ (v \in \N{\tc}{u} \} \}$
    \end{enumerate}
    While (1) and (2) are immediate for (3) we further elaborate. Note that for node $u \in C$ we can construct its neighborhood in $\widehat{G}_{\tc} $ as follows: first add all edges to nodes in $\N{\htu}{u} \cap C$ and then for every node $v \in  \N{\tc}{u} \cap  (C \setminus \N{\htu}{u})$ such that $\htv > \htu $ add edge $(u, v)$. Then (3) follows by noting that sets  $\{v \in  \N{\tc}{u} \cap  (C \setminus \N{\htu}{u}): \htv > \htu \}$ and $\{ v\in C: \{\ta{v} > \htu\} \land \{ (v \in \N{\tc}{u} \} \}$ are equal.
    We now argue that $N_u \simeq  (\N{\htu}{u} \cap C) $. For each node $ v\in C$ such that $\ta{v} > \htu $ and $ v \in \N{\tc}{u} $ it holds that:
    \begin{enumerate}
        \item $ v \in \N{\tc}{u} \setminus \N{\htu}{u}$; and
        \item $|\N{\htv}{v}| < 2d$ where $\htv > t_u$
    \end{enumerate}
    For the second point note that $\htv \geq \ta{v}$,  $\ta{v} > \htu$ and $\htu \geq t_u$. From these two points note that $\R$ implies that $ | \{ v\in C: \{\ta{v} > \htu\} \land \{ (v \in \N{\tc}{u} \} \ | < \e /10^{2}d$ and from~\cref{eq: fourth important} we have that $|\N{\htu}{u} \cap C| \geq \frac{d}{2}  $. Combining these facts:
    \begin{align*}
        |N_u \triangle  (\N{\htu}{u} \cap C) | &= |N_u \setminus  (\N{\htu}{u} \cap C) |\\
        &= |\{ v\in C: \{\ta{v} > \htu\} \land \{ (v \in \N{\tc}{u} \} \}| \\
        & \leq \e /10^{2}d  \\
        & \leq  2\e/10^{2} |\N{\htu}{u} \cap C|\\
         & \leq  2\e/10^{2} \min \{  |N_u| , |\N{\htu}{u} \cap C|\}
    \end{align*}

Also applying~\cref{lem: min implies max implies min set inequalities} to~\cref{eq: third important} we get that $\forall u \in \Cnotanchor$ and $u' \in \Canchor$:
\begin{align*}
    |(\N{\htu}{u} \cap C) \triangle \Nt{\tc}{u'}| \leq 3339 \epsilon \min \{|\N{\htu}{u} \cap C|, ||\Nt{\tc}{u'}| \}
\end{align*}
 Thus, again from~\cref{lem: min implies max implies min set inequalities} and combining:
 \begin{enumerate}
     \item $ |(\N{\htu}{u} \cap C) \triangle \Nt{\tc}{u'}| \leq 3339 \epsilon \min \{|\N{\htu}{u} \cap C|, ||\Nt{\tc}{u'}| \}$; and
     \item $ |N_u \triangle  (\N{\htu}{u} \cap C) |  \leq 2\e/10^{2} \min \{  |N_u| , |\N{\htu}{u} \cap C|\}$
 \end{enumerate}
 We conclude that $\forall u \in C$ and $u' \in \Canchor$:
    \begin{align*}
        |N_u \triangle \Nt{\tc}{u'} | & \leq 3 (3339 + 2\e/10^{2}) \min \{  |N_u| , |\Nt{\tc}{u'}|\}\\
        &\leq 10020 \e \min \{  |N_u| , |\Nt{\tc}{u'}|\}\tag{*****}\label{eq: fifth important}
    \end{align*}
Moreover, applying the latter inequality for any $u, w \in C$ and $u' \in \Canchor$ we conclude that:
\begin{align*}
        |N_u \triangle N_w | & \leq 3 (10020 + 10020) \min \{  |N_u| , |N_w|\}\\
        &\leq 60120 \e \min \{  |N_u| , |N_w|\}\\
        &\leq 60120 \e \max \{  |N_u| , |N_w|\}
\end{align*}
and 
\begin{align*}
        |N_u \cap  N_w | & \geq  60120 \e \max \{  |N_u| , |N_w|\}
\end{align*}

Thus, $\forall u, w \in C$:
\begin{enumerate}
    \item $u$ and $w$ are connected; and
    \item If $u$ and $w$ are connected in $\widehat{G}_{\tc} $ then they are in $ 60120 \e $-agreement
\end{enumerate}

Thus, $C$ is a cluster in the agreement decomposition of graph $\widehat{G}_{\tc} $ when agreement and heaviness parameters are set to  $ 60120 \e $. Using~\ref{property: N_G(u) cap C  >= (1 - 9 epsilon) C} we can deduce that $\forall u \in C$
\begin{align*}
        |N_u \cap C| & \geq ( 1 - 9 \cdot 60120 \e ) |C| = ( 1 - 541080 \e ) |C| 
\end{align*}

We now proceed lower bounding $|\N{t_u}{u}|$. For a node $u \in C$ and $u' \in \Canchor$ using (c), (d), (e) and~\cref{lem: min implies max implies min set inequalities} we can deduce that:
\begin{align*}
    |\N{t_u}{u} \triangle \Nt{\tc}{u'}| &\leq  3\cdot (3\cdot (3 + 39) + 240) \e \min \{ |\N{t_u}{u} |, | \Nt{\tc}{u'}| \}\\
    &\leq 1098 \e \min \{ |\N{t_u}{u} |, | \Nt{\tc}{u'}| \}
\end{align*}

Combining the latter with~\cref{eq: fifth important} we get:
\begin{align*}
    |\N{t_u}{u} \triangle N_u| &\leq  3\cdot (10020 + 1098) \e \min \{ |\N{t_u}{u} |, | N_u| \}\\
    &\leq 33354 \e \min \{ |\N{t_u}{u} |, | N_u| \}
\end{align*}
We are now ready to deduce that:
\begin{align*}
    |\N{t_u}{u} | &\geq  (1 -  33354\e)  |N_u| \\
    &\geq  (1 -  33354\e)  |N_u \cap C|\\
    &\geq  (1 -  33354\e) ( 1 - 541080 \e ) |C| \\
    &\geq  (1 - 574434 \e ) |C|\\
    &\geq |C|/2
\end{align*}
where in the third inequality we used that $C$ is dense and the last inequality holds for $\e$ small enough.
\end{proof}
\section{Runtime Analysis}\label{sec: Runtime analysis}

In this section we compute the complexity of our algorithm and prove that, on expectation, the time spent per each node insertion or deletion is $O(\polylog n)$. Note that the $\epsilon$ used in~\cref{alg:dynamicagreement} and in all proofs is a small enough constant which gets ``absorbed'' in the big $O(\cdot)$ notation. We introduce some auxiliary random variable definitions:
\begin{enumerate}
    \item $W_t$ denotes the complexity of operations at time $t$;
    \item  $L$ is the size of the largest connected component of $\widetilde{G}_{t-1}$; and
    \item $Q_t$ is the total number of agreement and heaviness calculations at time $t$.
\end{enumerate}
Note that edge additions are at least as many as edge deletions and each edge addition is preceded by a heaviness and agreement calculation between its endpoints. Thus, the total complexity of~\cref{alg:dynamicagreement} is upper bounded by a constant times the total complexity of agreement and heaviness calculations. 
Since from~\cref{sec: Agreement and heavyness calculation} each agreement and heaviness calculation requires $O (\log n)$ time we have that:
\begin{align*}
     \expect \left[ \sum_t W_t \right] &\leq O (\log n) \cdot \expect \left[ \sum_t Q_t \right] \\
      &\leq O (\log n) \cdot \sum_t \expect \left[ Q_t \right]
\end{align*}
The focus in the rest of the section will be to prove that $\expect \left[ Q_t \right] = O (\polylog n) $. We start by arguing that the number of nodes that participate in an ``interesting event'' at time $t$ is, on expectation, at most $\polylog n$.

\begin{lem}\label{lem: size of It is bounded}
    \begin{align*}
        \expect \left[|\I_t| \right] =  O (\polylog n)
    \end{align*}
\end{lem}
\begin{proof}
In the case of a node addition it is always true that $|\I_t| \leq  \left( \nicefrac{10^{10} \log n}{\e} \right)^{4}  =  O (\polylog n)$. To argue that the same bound holds on expectation in the case of a node deletion we introduce some auxiliary notation. For a node $v \in G_{t'}$ we use $I_v^{i, t'}$ and $B_v^{i, t'}$ to denote the sets $I_v^{i}$ and $B_v^{i}$ respectively at the end of iteration $t'$. At time $t$ we have that:
\begin{align*}
    |\I_t| &\leq \left| \bigcup_i  B_{v_{t}}^{i, t-1} \right| \cdot \left( \nicefrac{10^{10} \log n}{\e} \right)^{3} \\
        &\leq \sum_i | B_{v_{t}}^{i, t-1} | \cdot \Theta \left( \polylog n \right) 
\end{align*}

Also, our algorithm maintains the following invariant:
\begin{align*}
    \sum_{v \in G_{t'}, i} |I_v^{i, t'}| =  \sum_{v \in G_{t'}, i} |B_v^{i, t'}| = O (n \polylog n), \; \forall t'
\end{align*}
Thus, at time $t$ if we select uniformly at random a node and delete it, we have that: 
\begin{align*}
    \expect \left[  \sum_{i} |B_{v_t}^{i, t}| \right] = O (\polylog n)
\end{align*}
\end{proof}

The main~\cref{thm: time complexity bound} of this section requires a delicate coupling argument. We give a proof sketch before proceeding to the formal proof. Similarly to~\cref{sec: All found clusters are dense} let $t_u$ denote the last time strictly before $t$ that $u$ participated in an ``interesting event''.
We remind the reader an important event definition of~\cref{sec: All found clusters are dense} which will be used also in this section.

\newcounter{savedcounter}
\setcounter{savedcounter}{\value{lem}}

\setcounter{definition}{\value{importantdefinitioncounter}} 

\begin{definition}
\begin{align*}
    \mathcal{R} = 
        &\bigcap_{u \in C } M_u^{t_u, t-1}
        \bigcap_{u \in C, d\in \left[1, n \right]}
            \left\{
                |P_{d, t_u, t-1}^u| \leq \epsilon/10^{2} d 
            \right\}
\end{align*}
\end{definition}

\setcounter{definition}{\value{savedcounter}} 

Let $C_u$ be the connected component $u$ belongs to in $\widetilde{G}_{t-1}$ then from~\cref{thm: C is dense} we have that: $\R$ implies that $\forall v \in C_u$ it holds that $|\N{t_u}{v}| \geq |C_u|/2$.
 
Since event $\R$ happens with high probability and $|C_u| \leq n$, we can argue that  $\expect [\N{t_u}{u}|] \geq  |C_u|/3$. Consequently for a connected component $C$ of $\widetilde{G}_{t-1}$ we can upper bound, on expectation, the number of nodes in the anchor set as follows:
\begin{align*}
    \sum_{ u \in C} \nicefrac{10^{7} \log n}{ \e |\N{t_u}{u}|} &\simeq \sum_{ u \in C} \nicefrac{ 3 \cdot  10^{7} \log n}{ \e |C|} \\
    &=  \nicefrac{3 \cdot 10^{7} \log n}{ \e}\\
    &=O (\log n)
\end{align*}
Consequently $\expect [L] = O (\log n)$. In the last step of our proof we argue that $\expect [W_t] = O (\log n)  \expect [ |\I_t|] \expect [L] = O (\polylog n) $

\begin{thm}\label{thm: time complexity bound}
    \begin{align*}
     \expect \left[ \sum_t W_t \right] = O (\polylog n) 
\end{align*}
\end{thm}
\begin{proof}
As already mentioned in the beginning of the section it is enough to argue that  $\expect \left[ Q_t \right] = O (\polylog n) $. We use $\Phi$ to denote the set of anchor set nodes at time $t-1$ and by $L$ the maximum number of anchor set nodes in any connected component of our sparse solution at time $t-1$.

For a node $v \in \widetilde{G}_{t}$, on expectation, the number of agreement and heaviness calculations for each subroutine is upper bounded as follows:
\begin{enumerate}
     \item Connect($ v, \epsilon, t$): $10^{7} \log n / \e \cdot L = L \cdot O (\log n)$
    \item Anchor($ v, \epsilon, t$): $\min \left\{  \{ |\N{u}{t}| \cdot \frac{10^{7} \log n}{\e |\N{u}{t}|}, |\N{u}{t}|\} \right\} = O(\log n)$
    \item Clean($ v, \epsilon, t$): $L$
\end{enumerate}
Let $W_{t, v}$ be the total number of agreement and heaviness calculations required by those three procedures, then:
\begin{align*}
    \expect [W_{t, v}] =   O (\log n ) \cdot \expect [L]
\end{align*}
We continue upper bounding $\expect [L] \leq O (\log n)$.

From the law of total expectation we have:
\begin{align*}
    \expect [ L] &= \expect [ L \mid \R] \cdot \prob [R] + \expect [ L \mid \overline{\R}] \cdot \prob [\overline{R}]\\
    &\leq \expect [ L \mid \R]  + n \cdot  \prob [\overline{R}]\\
    &\leq \expect [ L \mid \R]  + n \cdot   1/n^{10^{2}}\\
    &<\expect [ L \mid \R]  + 1/n^{99}
\end{align*}
where in the first inequality we used that $ \prob [R] \leq 1$ and $L\leq n$ are always true and in the second inequality we used~\cref{obs: R happens with high probability}.
We now concentrate on bounding the first term of the summation.
We remind the reader that $u_i$ is the $i$-th node in our node stream.

From the principled of deferred decisions we can construct an instance of $\widetilde{G}_{t-1}$ as follows: we first run Notify($ v_1, \epsilon$),  Notify($ v_2, \epsilon$), \dots,  Notify($ v_{t-1}, \epsilon$) sequentially and calculate $\I_1,\I_2, \dots, \I_{t-1} $. Note that at this point for each node $u$ in $G_{t-1}$, the random variable $t_u$ is realized, indeed $t_u = \max \{ i \in [1,  t-1]: u \in \I_i\}$. Moreover $\Phi$ is not realized, since we did not sample any Bernoulli variable from our Anchor procedure yet. We now run the ``for loop'' of our algorithm which contains the Clean, Connect and Anchor procedures, sequentially for $\I_1,\I_2, \dots, \I_{t-1} $. Let $F$ denote a realization of the first step and let $X_u$ be the Bernoulli variables used in the Anchor($ u, \epsilon, t_u$) procedures for all $u$'s.

Towards a coupling argument we now describe a second stochastic procedure to construct a ``sparse graph'' $\widehat{G}_{t-1}$: we again run the Notify($ v_1, \epsilon$),  Notify($ v_2, \epsilon$), \dots,  Notify($ v_{t-1}, \epsilon$) sequentially and calculate $\I_1,\I_2, \dots, \I_{t-1} $. In the second step we run the ``for loop'' of our algorithm sequentially for $\I_1,\I_2, \dots, \I_{t-1} $ assuming that all the Bernoulli variables which were sampled by the Anchor procedures are $1$. As a final step we use the variables $X_u$ (also used by the first procedure) and delete from the anchor set all nodes for which $X_u = 0$ and also delete all edges between nodes $u, v$ such that both $X_u$ and $X_v$ are $0$. 


Let $\F$ be the set of possible realizations of the first step conditioned on event $\R$. From the law of total expectation it suffices to prove that $\forall F \in \F$ we have:
\begin{align*}
    \expect [ L  \mid F \land \R] = O (\log n)
\end{align*}

Let $\widehat{\C}$ denote the set of connected components of the sparse graph constructed by the second stochastic procedure at the second step. Note that given $F$, $\widehat{\C}$ is a deterministic set. Similarly let $\C$ be the random variable denoting the set of connected components of $\widetilde{G}_{t-1}$. Note that the following is always true: $\forall C \in \C$ there exists a set $C' \in \widehat{\C}$ such that $C \subseteq C'$.
Now, for all sets $C' \in \widehat{\C}$ let $L_{C'}$ be the random variable denoting the number of anchor set nodes in set $C'$ after the third step. We have that $\max_{C' \in \widehat{\C}} L_{C'} \geq L$.
Consequently, it is enough to bound the following quantity:
\begin{align*}
    \expect [ \max_{C' \in \widehat{\C}} L_{C'}  \mid F \land \R] = O (\log n)
\end{align*}

Note that~\cref{thm: C is dense} also applies to the clustering $\widehat{\C}$. Let $C' \in \widehat{\C} $  then $\forall u  \in C'$:
    $$ {\expect [ X_u \mid F \land \R] \leq   \nicefrac{2 \cdot 10^{7} \log n}{\epsilon|C'|}} $$
$$\expect [L_{C'} \mid F \land \R] = \expect [\sum_{u  \in C'} X_u \mid F \land \R] \leq \nicefrac{2 \cdot 10^{7} \log n}{\epsilon} $$
Using Chernoff we can get that: 

$$\prob [L_{C'} > 2 \cdot \frac{2 \cdot 10^{7} \log n}{\epsilon} \mid F \land \R] \leq 1/n^{10}$$

Let $T = \{ \forall C' \in \widehat{\C}:  L_{C'} \leq 2 \cdot  \frac{2 \cdot 10^{7} \log n}{\epsilon}\}$, we can conclude that:
\begin{align*}
\expect [ \max_{C' \in \widehat{\C}} L_{C'}  \mid F \land \R] &=\expect [ \max_{C' \in \widehat{\C}} L_{C'}  \mid T \land F \land \R] \cdot \prob [ T ] +\expect [ \max_{C' \in \widehat{\C}} L_{C'}  \mid \overline{T} \land F \land \R] \cdot \prob [\overline{T}]\\ 
&\leq 2 \cdot  \frac{2 \cdot 10^{7} \log n}{\epsilon} +n \cdot \prob [\overline{T}]\\ 
&\leq 2 \cdot  \frac{2 \cdot 10^{7} \log n}{\epsilon} +n \cdot n \cdot  1/n^{10}\\
&= O (\polylog  n)
\end{align*}
where in the first inequality we used the definition of event $T$ and that $ \max_{C' \in \widehat{\C}} L_{C'} \leq n$ and
in the second inequality used a union bound on the sets of $\widehat{\C}$.


We now conclude the proof. Note that while $W_{t, v}$ only depends on the random choices of our algorithm up until time $t-1$, $\I_t$ is independent of those. Consequently:
\begin{align*}
    \expect \left[ Q_t \right] &=   \expect \left[\sum_{v \in V} \mathbbm{1} \{v \in \I_t\} W_{t, v} \right ] \\
    &= \sum_{v \in V}  \expect \left[ \mathbbm{1} \{v \in \I_t\} W_{t, v} \right ] \\
    &= \sum_{v \in V}  \expect \left[ \mathbbm{1} \{v \in \I_t\} \right] \expect \left[ W_{t, v} \right] \\
    &\leq \sum_{v \in V}  \expect \left[ \mathbbm{1} \{v \in \I_t\} \right] O (\log n) \expect  \left[ L \right] \\
    &\leq O (\log n) \expect  \left[ L \right] \sum_{v \in V}  \expect \left[ \mathbbm{1} \{v \in \I_t\} \right]  \\
    &\leq O (\log n) \expect  \left[ L \right]  \expect \left[ \sum_{v \in V} \mathbbm{1} \{v \in \I_t\} \right]  \\
    &\leq O (\log n) \expect  \left[ L \right]  \expect \left[  |\I_t| \right]  \\
    &\leq O (\polylog n)  \\
\end{align*}
where in the last inequality we used~\cref{lem: size of It is bounded}.

\end{proof}

\section{Efficient Clustering Computation from Sparse Solutions}\label{sec: Efficient clustering computation from sparse solutions}
In previous sections we prove that our algorithm computes a series of sparse graphs $\widetilde{G}_1, \widetilde{G}_2, \dots, \widetilde{G}_n$ such that with high probability:
\begin{enumerate}
    \item a total update time complexity of $\Theta(n \polylog n)$ is required; and
    \item $\forall t$, the connected components of $\widetilde{G}_t$ define constant factor approximation correlation clustering for graph $G_t$.
\end{enumerate}

This section aims to demonstrate how we can compute and output the connected components of any sparse graph $\widetilde{G}_t$ in $\Theta(n)$ time. Note that $\widetilde{G}_t$ contains at most $\Theta(n \polylog n)$ edges. Consequently, a naive approach of computing connected components by traversing all edges of $\widetilde{G}_t$ would result in a time complexity of $\Theta(n \polylog n)$.

It is worth emphasizing that the complexity of the offline algorithm proposed in~\cite{AssasdiCC} is also  $\Theta(n \polylog n)$. Therefore, if $\Theta(n \polylog n)$ time is required just to compute the connected components of $\widetilde{G}_t$, simply rerunning the algorithm from~\cite{AssasdiCC} would be an equivalent solution our algorithm; rendering our effort in maintaining a sequence of sparse graphs meaningless.

Following the notation of the previous sections, let $\widetilde{G} = (V, \widetilde{E})$ be a sparse graph maintained by our algorithm, $\Phi$ the set of anchor nodes and $\Phi_{v} = \Nt{}{v} \cap \Phi$ the anchor set nodes which are connected to $v$ in our sparse solution $\widetilde{G}$. To ease notation, we write $u \sim v$ if and only if $u$ and $v$ belong to the same connected component of $\widetilde{G}$. 

We will describe a procedure that constructs a function $f:V \xrightarrow{}  \mathbb{N}$. 
$f$ will encode our clustering solution as follows: two nodes $u, v \in V$ are in the same cluster if and only if $f(u) = f(v)$. We then argue that: if $\forall u, v \in \Phi$ which are in the same connected of $\widetilde{G}$ it holds that $\Nt{}{u} \cap \Nt{}{v} \neq \emptyset$ then $f$'s clustering coincides with the connected components of $\widetilde{G}$. Note that from~\cref{sec: All found clusters are dense} we can use~\cref{lem: two anchor sets have a non empty intersection in our sparse solution} to prove that the latter property holds with high probability. To ease notation, We write $u \sim v$ if and only if $u$ and $v$ belong to

\renewcommand{\algorithmicwhile}{\textbf{while}}
\renewcommand{\algorithmicendwhile}{\algorithmicend\ \textbf{while}}
\newcommand{\algorithmicbreak}{\textbf{break from for loop}}
\newcommand{\BREAK}{\STATE \algorithmicbreak}
\renewcommand{\algorithmicrequire}{\textbf{Initialization:}}

\begin{algorithm}
\caption{ComputeConnectedComponents($\widetilde{G}$)}\label{alg:ComputeConnectedComponents}
\begin{algorithmic}
\REQUIRE $\forall v \in V$: $f(v) \xleftarrow{} -1$, $ID \xleftarrow{} 0$, $Q \xleftarrow{} \Phi$
\WHILE{$Q \neq \emptyset$}
    \STATE Let $u$ be any node in $Q$
    \STATE $f(u) \xleftarrow{} ID$
    \STATE $T \xleftarrow{} \{u \}$
    \FORALL{$ v \in \Nt{}{u}$}
        \IF{$f(v) = -1$}
            \STATE \textcolor{blue}{No conflict phase}
            \STATE $f(v) \xleftarrow{} ID$
            \STATE $T \xleftarrow{} T \cup \{v\}$
        \ELSE
            \STATE \textcolor{blue}{Resolving conflict phase}
            \FORALL{$ v' \in T$}
                \STATE $f(v') \xleftarrow{} f(v)$
            \ENDFOR
            \BREAK
        \ENDIF
    \ENDFOR
\STATE $ID \xleftarrow{} ID + 1$
\STATE $Q \xleftarrow{} Q \setminus T$
\ENDWHILE
\FORALL{$v \in V \setminus \Phi$ such that $f(v) = -1$}
    \STATE Let $u$ be a node in $\Phi_v$
    \STATE $f(v)\xleftarrow{} f(u) $
\ENDFOR
\end{algorithmic}
\end{algorithm}

\begin{lem}
    If for $\widetilde{G}$ and $\Phi$:
    \begin{enumerate}
        \item $V = \bigcup\limits_{v \in \Phi} \Nt{}{v}$; and
        \item $\forall u, v \in \Phi$ such that $u \sim v$, $\Nt{}{u} \cap \Nt{}{v} \neq \emptyset$
    \end{enumerate}
    then $f(u) = f(v)$ if and only if $u \sim v$.
\end{lem}

\begin{proof}
     Note that after the while loop is terminated all nodes in $\Phi$ are assigned a value different than $-1$. 
     In addition, $\forall v \in V \setminus \Phi$ there exists a node $u \in \Phi_{v}$ such that $f(v) = f(u)$.
     Thus, it suffices to argue that for any two nodes $u, v \in \Phi$: $f(u) = f(v)$ if and only if $u \sim v$. 
         Let $u_i$ be the node selected from $Q$ in the beginning of the while's loop $i$-th iteration, $f_i$ the $f$ function assignment after the termination of the $i$-th iteration and $\Phi_i$ all the nodes in $u \in \Phi$ such that $f_i(u) \neq -1$. Note that for any node $u$ if $f_i(u) \neq -1 $ then $\forall j > i$: $f_j(u) \neq -1 $ and $f_j(u) = f_i (u)$. Consequently, $\Phi_i \subset\Phi_{i+1}, \forall i $. We prove by induction the following statements in tandem:
         \begin{enumerate}
             \item if $f(u_i) = i$ then $\forall j \geq i$ and $v \in \Nt{}{u}$ $f_j(v) = k$; and
             \item $\forall u, v \in \Phi_i$: $f_i(u) = f_i(v)$ if and only if $u \sim v$.
         \end{enumerate}
         For $i = 1$ note that the algorithm updates the $f$ function only in the \textcolor{blue}{No conflict phase} for all nodes in $\Nt{}{u_1}$ with the same $ID = 1$. Thus, both statements hold.
         For $i>1$ we consider two cases:
         \begin{enumerate}
             \item If $f$ is updated only in the \textcolor{blue}{No conflict phase} then all nodes in $\Nt{}{u_i}$ receive the same $ID = i$. Since $ID = i$ is first used in iteration $i$ all nodes that received an $ID$ in previous iterations maintain their previous $ID$ (which is smaller than $i$). To conclude this case, we need to argue that $\forall u \in \Phi_{i-1}$ $u \not \sim u_i$. Towards a contradiction assume that such a node exists. Let $u \in \Phi_{i-1}$ be such that $f_{i-1} (u) = j < i$ and $u \sim u_i$. By definition, $u_j$ is the first node to receive $ID = j$ and since $\sim$ is a transitive relation $u_j \sim u_i$. By our inductive hypothesis $\forall v \in \Nt{}{u_j}$ $f_{i-1}(v) = j$ and by the conditions of the current lemma, $\exists w \in \Nt{}{u_j} \cap \Nt{}{u_i}$. However, this is a contradiction because $f$ would have been updated also in \textcolor{blue}{Resolving conflict phase}.
             \item If $f$ is updated also in the \textcolor{blue}{Resolving conflict phase} then we have that $f_i (u_i) = j < i$ and it suffices to argue that there exists a node $ u \in \Phi_{i-i}$ such that $f_{i-1} (u) = j$ it holds that $u \sim u_i$. Let $v \in \Nt{}{u_i}$ be a node such that $f_{i-1} (v) = j \neq -1$. Note that such a node exists since $f$ is also updated in the \textcolor{blue}{Resolving conflict phase}. W.l.o.g. assume that the value of $f$ at $v$ was updated at time $j'\in [j, i-1]$. Then, either $v \in \Phi_{i-i}$ or $v \in \Nt{}{u_{j'}}$ where for $u_{j'} \in \Phi_{i-1}$ it holds that $f_{i-1}(u_{j'}) = j$. In both cases since $\sim$ is a transitive relation we conclude that a node $ u \in \Phi_{i-i}$ such that $f_{i-1} (u) = j$ and $u \sim u_i$ always exists.
         \end{enumerate}
\end{proof}

\begin{lem}
    \cref{alg:ComputeConnectedComponents} has complexity $O (|V|)$.
\end{lem}
\begin{proof}
    The final for loop has complexity at most $\Theta(|V|)$. To bound the complexity of the while loop we simply note that for every node $v$ its $f$ function value changes at most twice. Consequently, operations of the form ``$f(\cdot) \xleftarrow{}$'' are at most $2 |V|$.
\end{proof}

\section{Agreement and Heaviness Calculation with \texorpdfstring{$ O( \log n )$}{} Sampled Nodes}\label{sec: Agreement and heavyness calculation}

In this section we design the ProbabilisticAgreement($ u, v, \epsilon$) and Heavy($ u, \epsilon$) procedures.

These two procedure are used to test if two nodes $u$ and $v$ are in $\e$-agreement and if a node $u$ is $\epsilon$-heavy.
The idea is that if we let some slack on how much in agreement and how much heavy a node is then $O (\polylog n)$ samples of each node's neighborhood are enough to design a ProbabilisticAgreement($ u, v, \epsilon$) procedure that, w.h.p. answers affirmatively if the two nodes are indeed in $0.1 \e$-agreement and at the same time answers negatively if these two nodes are not in $\e$-agreement. Consequently, if two nodes are in an $\e'$-agreement for $\e'$ between $0.1 \e$ and $\e$ then  ProbabilisticAgreement($ u, v, \epsilon$) may answer positively or negatively. We do the same for the Heavy($ u, \epsilon$) procedure.

We start by proving some useful inequalities regarding the neighborhood of two nodes in case they are indeed in $\e$-agreement and in case they are not. Subsequently we analyze the properties of ProbabilisticAgreement($ u, v, \epsilon$) and finally we do the same for the Heavy($ u, \epsilon$) procedure.

\begin{algorithm}
\caption{\probabilisticagreement{u}{v}{\epsilon}}\label{alg:probagreem}
\begin{algorithmic}
\STATE {\bfseries Initialization:} $k = 300 \log n / \epsilon$.
\FOR{$i = 1$ to $k$}
    \STATE Draw a random neighbor $r_i$ of $u$ and a random neighbor $s_i$ of $v$.
    \STATE Let $x_i = \mathbbm{1} \{ r_i \in N(u) \setminus N(v) \} $, $y_i = \mathbbm{1} \{ r_i \in N(v) \setminus N(u) \} $
\ENDFOR
\IF{$\sum_i x_i / k < 0.4 \epsilon$ and $\sum_i y_i / k < 0.4 \epsilon$}
    \STATE Output `` YES''
\ELSE
    \STATE Output ``NO''
\ENDIF
\end{algorithmic}
\end{algorithm}

\begin{algorithm}
\caption{\heavyprocedure{u}{\epsilon}}\label{alg:heavy}
\begin{algorithmic}
\STATE {\bfseries Initialization:} $k = 1200 \log n / \epsilon$.
\FOR{$i = 1$ to $k$}
    \STATE Draw a random neighbor $v_i$ of $u$.
    \STATE Let $x_i = \mathbbm{1} \{    \text{ProbabilisticAgreement}(u, v_i, \epsilon) = \text{``No''} \} $
\ENDFOR
\IF{$\sum_i x_i / k < 1.2 \epsilon$}
    \STATE Output `` Yes''
\ELSE
    \STATE Output ``No''
\ENDIF
\end{algorithmic}
\end{algorithm}

\begin{prop}
    Let $u, v$ be two nodes which are in $\epsilon$-agreement then $\frac{|N(v)|}{1 - \epsilon} \geq  |N(u)| \geq (1 - \epsilon) |N(v)|$.
\end{prop}
\begin{proof}
    Assume w.l.o.g. that $|N(u)| > |N(v)|$. Then $|N(u)| - |N(v)| \leq |N(u) \setminus N(v)| \leq |N(u) \triangle N(v)|  \leq \epsilon \max \{|N(u)|, |N(v)| \} $
\end{proof}

\begin{prop}
     Let $u, v$ be two nodes which are in $\epsilon$-agreement then $\frac{|N(u) \setminus N(v)|}{|N(u)|} \leq 1.1 \epsilon$ and $\frac{|N(v) \setminus N(u)|}{|N(v)|} \leq 1.1 \epsilon$.
\end{prop}
\begin{proof}
    $|N(u) \setminus N(v)| \leq |N(u) \triangle N(v)|  \leq \epsilon \max \{|N(u)|, |N(v)| \}  \leq  \epsilon \frac{|N(u)|}{1-\epsilon} \leq 1.1 |N(u)| $ for $\epsilon$ small enough.
\end{proof}

\begin{prop}
     Let $u, v$ be two nodes which are not in $\epsilon$-agreement then either $\frac{|N(u) \setminus N(v)|}{|N(u)|} \geq 0.5 \epsilon$ or $\frac{|N(v) \setminus N(u)|}{|N(v)|} \geq 0.5 \epsilon$.
\end{prop}
\begin{proof}
    Since $|N(v) \setminus N(u)| + |N(u) \setminus N(v)|  = |N(u) \triangle N(v)|  \geq \epsilon \max \{|N(u)|, |N(v)| \}$  then either  $|N(v) \setminus N(u)| > 0.5 \epsilon |N(v)| $ or $|N(u) \setminus N(v)| > 0.5 \epsilon |N(u)| $.
\end{proof}

\begin{obs} \label{obs:probsofindicatorsforagreement}
Let $u, v$ be two nodes which share an edge, $s, t$ two nodes chosen uniformly at random from $N(u)$ and $N(v)$ respectively and  $X_s = \mathbbm{1} \{ s \in N(u) \setminus N(v) \} $, $X_t = \mathbbm{1} \{ s \in N(v) \setminus N(u) \} $. We have the following:
\begin{enumerate}
    \item If $u, v$ are \underline{in agreement} then $\prob [ X_s = 1]  \leq 1.1 \epsilon$ and $\prob [ X_t = 1]  \leq 1.1 \epsilon$.
    \item If $u, v$ are \underline{not in agreement} then either $\prob [ X_s = 1]  \geq 0.5 \epsilon$ or $\prob [ X_t = 1]  \geq 0.5 \epsilon$.
\end{enumerate}
\end{obs}

\begin{thm}{(Chernoff bound)}\label{thm: chernoff}
    Let $X_1, X_2, \dots, X_k$ be $k$ i.i.d. random variables in $[0,1]$. Let $X = \sum_i X_i / k$. Then:
    \begin{enumerate}
        \item For any $\delta \in [0, 1]$ and $U \geq E [X] $ we have
        $$ \prob [X \geq (1 + \delta) U] \leq \exp (- \delta^2  U k /3)$$
        \item For any $\delta > 0$ and $U \leq E [X] $ we have
        $$ \prob [X \leq (1 - \delta) U] \leq \exp (- \delta^2  U k /2)$$
    \end{enumerate}
\end{thm}

\begin{lem}\label{lem: 1 agreement}
If Algorithm \ref{alg:probagreem} outputs ``YES'' then $u,v$ are in $\epsilon$-agreement with probability greater than $1 - 1/n^3$.
\end{lem}
\begin{proof}
    We proceed by upper bounding the probability that the algorithm outputs ``YES'' and $u,v$ are not in agreement. Since $u,v$ are not in agreement by observation \ref{obs:probsofindicatorsforagreement} we have that either $E[\sum_i X_i / k] > 0.5 \epsilon$ or $E[\sum_i Y_i / k] > 0.5 \epsilon$. W.l.o.g. assume that $E[\sum_i X_i / k] > 0.5 \epsilon$.  By using the second inequality of the Chernoff bound with $U = 0.5 \epsilon$ and $\delta = 0.2$. We have 
    $$
    \prob [\sum_i X_i / k < (1 -\delta) U ] \leq  \exp( - 0.2^2 \cdot 0.5 \epsilon \cdot 300 \log n / \epsilon \cdot 0.5) = \exp (- 3 \cdot \log n) = \frac{1}{n^3}
    $$
\end{proof}

\begin{lem}\label{lem: 2 agreement}
If $u,v$ are in $0.1 \epsilon$-agreement then Algorithm \ref{alg:probagreem} outputs ``YES'' with probability greater than $1 - 1/n^3$.
\end{lem}
\begin{proof}
We will upper bound the probability that $u, v$ are in $0.1 \epsilon$-agreement and the algorithm outputs ``NO''. Note that since $u,v$ are in $0.1 \epsilon$-agreement by observation \ref{obs:probsofindicatorsforagreement} we have that $E[\sum_i X_i / k] < 1.1 \cdot 0.1 \epsilon = 0.11 \epsilon$ and $E[\sum_i Y_i / k] \leq 1.1 \cdot 0.1 \epsilon = 0.11 \epsilon$. The algorithm outputs ``NO'' if either $ \sum_i X_i / k \geq 0.4 \epsilon $ or $ \sum_i Y_i / k \geq 0.4 \epsilon $. We bound the probability of the first event using the first inequality of the Chernoff bound with $U = 0.4 \epsilon$ and $\delta = 2 $ $( > 0.4/0.11 - 1)$. We have 
$$
  \prob [\sum_i X_i / k \geq (1 +\delta) U ] \leq  \exp( - 2^2 \cdot 0.11 \epsilon \cdot 300 \log n / \epsilon \cdot 0.334) < \exp (- 44 \cdot \log n) = 1 / n^{44}
$$
Overall the probability that the algorithm outputs ``NO'' by the union bound is upper bounded by $ 1 / n^{44} + 1 / n^{44} < 1 /n^3$
\end{proof}

\begin{lem} \label{lem:probsofindicatorsforheavy}
Let $u$ be a node and $v$ a random neighbor of $u$. Let $X_v = \mathbbm{1} \{  \text{ProbabilisticAgreement}( u, v, \epsilon) = \text{``No''} \} $. We have the following:
\begin{enumerate}
    \item If $u$ is in \underline{$0.1 \epsilon$-agreement} with a \underline{$(1 - \epsilon)$-fraction} of its neighborhood then $\prob [ X_u = 1]  \leq 1.1 \epsilon $.
    \item If $u$ is not in \underline{$\epsilon$-agreement} with a \underline{$\epsilon$-fraction} of its neighborhood then $\prob [ X_u = 1]  > 3 \epsilon $.
\end{enumerate}
\end{lem}

\begin{proof}
Let $A_{uv}$ be the event that $u$ and $v$ are in $0.1 \epsilon$-agreement and $B_{uv}$ be the event that $u$ and $v$ are not in $\epsilon$-agreement. For (1) we have:
    \begin{dmath*}
    {\prob [ X_u = 1]  = \prob [ X_u = 1 \mid A_{uv}] \cdot  \prob [A_{uv}]} + {\prob [ X_u = 1 \mid \widetilde{A_{uv}}] \cdot  \prob [\widetilde{A_{uv}}]}  \leq {1/n^3 \cdot 1 + 1 \cdot \epsilon} < 1.1 \epsilon
    \end{dmath*}
and for (2) we have that:
 \begin{dmath*}
    {\prob [ X_u = 1]  = \prob [ X_u = 1 \mid B_{uv}] \cdot  \prob [B_{uv}]} + {\prob [ X_u = 1 \mid \widetilde{B_{uv}}] \cdot  \prob [\widetilde{B_{uv}}]}  
    \geq { (1 - 1/n^3) \cdot 1 + 1 \cdot (1 - \epsilon)} > 3 \epsilon
    \end{dmath*}
\end{proof}

\begin{lem}\label{lem: 1 heavy}
If $u$ is in $0.1 \epsilon$-agreement with at least a $(1 - \epsilon)$-fraction of its neighborhood, then Algorithm \ref{alg:heavy} outputs ``Yes'', i.e., that the node is heavy, with probability greater than $1 - 1/n^3$.
\end{lem}

\begin{proof}
We will upper bound the probability that Algorithm \ref{alg:heavy} outputs ``No''. Note that from lemma \ref{lem:probsofindicatorsforheavy} $E[\sum_i X_i / k] > 1.1 \epsilon$. By using the first inequality of the Chernoff bound with $U = 1.1 \epsilon$ and $\delta = 1 / 11$. We have:
$$
  \prob [\sum_i X_i / k \geq (1 +\delta) U ] \leq  \exp( - 1/11^2 \cdot 1.1 \epsilon \cdot 1200 \log n / \epsilon \cdot 1/3) < \exp (- 3.63 \cdot \log n) < 1 / n^3
$$
\end{proof}

\begin{lem}\label{lem: 2 heavy}
If $u$ is not in $\epsilon$-agreement with at least a $\epsilon$-fraction of its neighborhood, then Algorithm \ref{alg:heavy} outputs ``No'', i.e., that the node is not heavy, with probability greater than $1 - 1/n^3$.
\end{lem}

\begin{proof}
We will upper bound the probability that Algorithm \ref{alg:heavy} outputs ``Yes''. Note that from lemma \ref{lem:probsofindicatorsforheavy} $E[\sum_i X_i / k] > 3 \epsilon$. By using the second inequality of the Chernoff bound with $U = 3 \epsilon$ and $\delta = 0.6$. We have:
$$
  \prob [\sum_i X_i / k \leq (1 -\delta) U ] \leq  \exp( - 0.6^2 \cdot 3 \epsilon \cdot 1200 \log n / \epsilon \cdot 1/2) < \exp (- 648 \cdot \log n) < 1 / n^3
$$
\end{proof}

\section{Final Theorem}

We are now ready to prove that the clustering produced by our algorithm for every $t$ is a constant factor approximation to the optimal correlation clustering solution for graph $G_t$. To this end note that by~\cref{lem: original paper and how to bound the approximation ratio} it is enough to argue that with high probability:
\begin{enumerate}
    \item all dense enough clusters found by the agreement algorithm on each graph $G_t$ are also identified by our algorithm and all their nodes are clustered together; and
    \item all clusters that are found by our dynamic agreement algorithm are dense enough
\end{enumerate} 
For (1) we use~\cref{thm: all nodes of C are clustered together} and for (2) we use~\cref{thm: C is dense}.

\maintheorem*

\begin{proof}
\cref{thm: all nodes of C are clustered together} assumes the following event:  for all pair of nodes $u, v$ that are in $\epsilon/10^{14}$-agreement and all nodes $u'$ that are $\epsilon/10^{14}$-heavy the ProbabilisticAgreement($ u, v, \epsilon$) and Heavy($ u', \epsilon$) procedures of~\cref{sec: Agreement and heavyness calculation} output ``Yes'' and~\cref{thm: C is dense} assumes that: for all or all pair of nodes $u, v$ that are not in $\epsilon$-agreement and all nodes $u'$ that are not $\epsilon$-heavy the ProbabilisticAgreement($ u, v, \epsilon$) and Heavy($ u', \epsilon$) procedures of~\cref{sec: Agreement and heavyness calculation} output ``No''. 
Using~\cref{lem: 1 agreement},~\cref{lem: 2 agreement},~\cref{lem: 1 heavy} and ~\cref{lem: 2 heavy} we bound the probability that this event does not happen by $n^2 /n^3 + n^2/n^3 +n /n^3 + n/n^3 < 4/n$. Upper bounding over all times and possible clusters and using~\cref{thm: all nodes of C are clustered together} and~\cref{thm: C is dense} we conclude that the probability of the dynamic algorithm to output at each time a $O(1)$-approximation is at least:
\begin{equation*}
    1 - n^3/n^{10^{2}} - n^3/n^{10^{2}} - 4/n > 1 - 5/n
\end{equation*}
\end{proof}
\section{Structural Properties of the \textit{Agreement} Decomposition}\label{sec: structural properties of the decomposition}

Let $G$ be the graph and let $\mathcal{C}$ be the clustering produced by \vanillaagreementalgo{G}. Let $u,v$ be two nodes which belong to the same non-trivial cluster $C$ of $\mathcal{C}$, then for $\epsilon$ small enough the following properties hold, which were shown in \cite{OCCC}.

\newlist{Properties}{enumerate}{2}
\setlist[Properties]{label=Property \arabic*, font=\textbf, itemindent=1.5cm}

\begin{Properties}\label{properties}
    \item $ \abs{ N_G(u) \cap C}  \geq (1 - 3 \epsilon) \abs{N_G(u)}$ \label{property: N_G(u) cap C  >= (1 - 3 epsilon) N_G(u)}
    \item $ \abs{ N_G(u) \setminus C}  <  3 \epsilon \abs{N_G(u)}$ \label{property: N_G(u) setminus C  <  3 epsilon N_G(u)}
    \item $ \abs{C}  \geq (1 - 3 \epsilon) \abs{N_G(u)}$ \label{property: C  >= (1 - 3 epsilon) N_G(u)}
    \item $ \abs{ N_G(u) \cap C}  \geq (1 - 9 \epsilon) \abs{C}$ \label{property: N_G(u) cap C  >= (1 - 9 epsilon) C}
    \item $ \abs{ C \setminus N_G(u) }  <  9 \epsilon \abs{C}$ \label{property: C setminus N_G(u)   <  9 epsilon C}
    \item $ \abs{N_G(u)}  \geq (1 - 9 \epsilon) \abs{C}$ \label{property: N_G(u)  >= (1 - 9 epsilon) C}
    \item $\abs{N_G(u) \cap N_G(v)} \geq (1 - 5 \epsilon) \max \{ \abs{N_G(u)}, \abs{N_G(v)}  \}   $ \label{property: N_G(u) cap N_G(v) >= (1 - 5 epsilon) max(N_G(u), N_G(v)}
    \item $ \abs{N_G(v)} (1 - 5 \epsilon) \leq \abs{N_G(u)} \leq \frac{\abs{N_G(v)}}{1 - 5 \epsilon}$ \label{property: N_G(v) (1 - 5 epsilon) <= N_G(u) <= N_G(v)/(1 - 5 epsilon)}
    \item $\abs{C \setminus N_G(u)} < 9 \epsilon \abs{C} < \frac{9 \epsilon}{1 - 9 \epsilon} \abs{N_G(u)} $ \label{property: C setminus N_G(u) < 9 epsilon C < (9 epsilon)/(1 - 9 epsilon) N_G(u) }
    \item $\abs{N_G(u) \setminus C} < 3 \epsilon \abs{N_G(u)} < \frac{3 \epsilon}{1 - 3 \epsilon} \abs{C} $ \label{property: N_G(u) setminus C < 3 epsilon N_G(u) < (3 epsilon)/(1 - 3 epsilon) C}
    \item $N_G (u) \cap N_G (v)  \neq \emptyset $ \label{property: N_G (u) cap N_G (v) neq emptyset}
\end{Properties}

\section{Additional Experiments}\label{sec: AdditionalExperiments}
Here we have the experiments for the rest of the datasets presented in~\cref{sec: experimental evaluation}. Moreover, we present the performance of each dataset/algorithm pair when we restrict the node stream to only additions and we calculate the objective after all nodes have arrived. The correlation clustering objective value of each algorithm is divided by the performance of \algosingletons. Since \textsc{PIVOT} is (on expectation) a 3-approximation we note that the solution achieved by all other algorithms is at most a multiplicative factor 3 away from the optimum offline solution.

\begin{table}[h!]\label{tab: comparison with offline pivot to get a sense of the approximation ratio}
\centering
\caption{Performance on the entire graph when node stream contains only additions}
\begin{tabular}{lcccc}
\toprule
\textbf{Dataset} & \algoagreement & \algopivot & \algosingletons & \textsc{PIVOT} \\
\midrule
\datafacebook & 0.97 & 1.13 & 1.00 & 1.19 \\
\dataemail    & 0.95 & 1.08 & 1.00 & 1.25 \\
\datahepth    & 1.00 & 1.20 & 1.00 & 1.22 \\
\dataastroph  & 0.99 & 1.10 & 1.00 & 0.98 \\
\bottomrule
\end{tabular}
\end{table}

\begin{figure}[ht]
  \centering
  \includegraphics[width=0.55\textwidth]{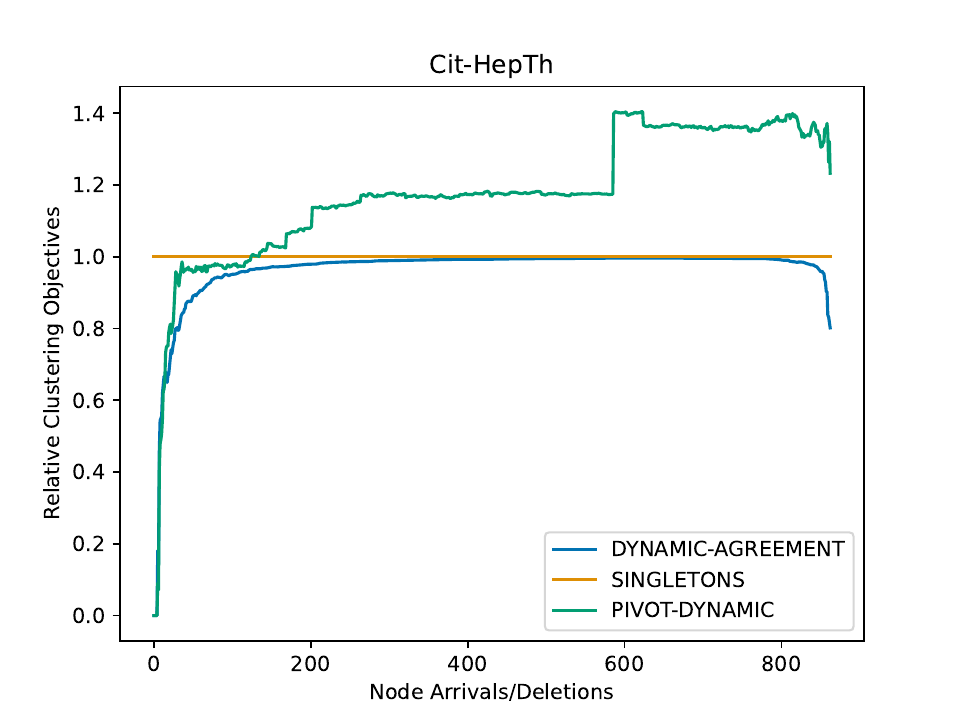} 
  \caption{Correlation clustering objective relative to singletons}
  \label{fig: Cit-HepTh}
\end{figure}

\begin{figure}[ht]
  \centering
  \includegraphics[width=0.55\textwidth]{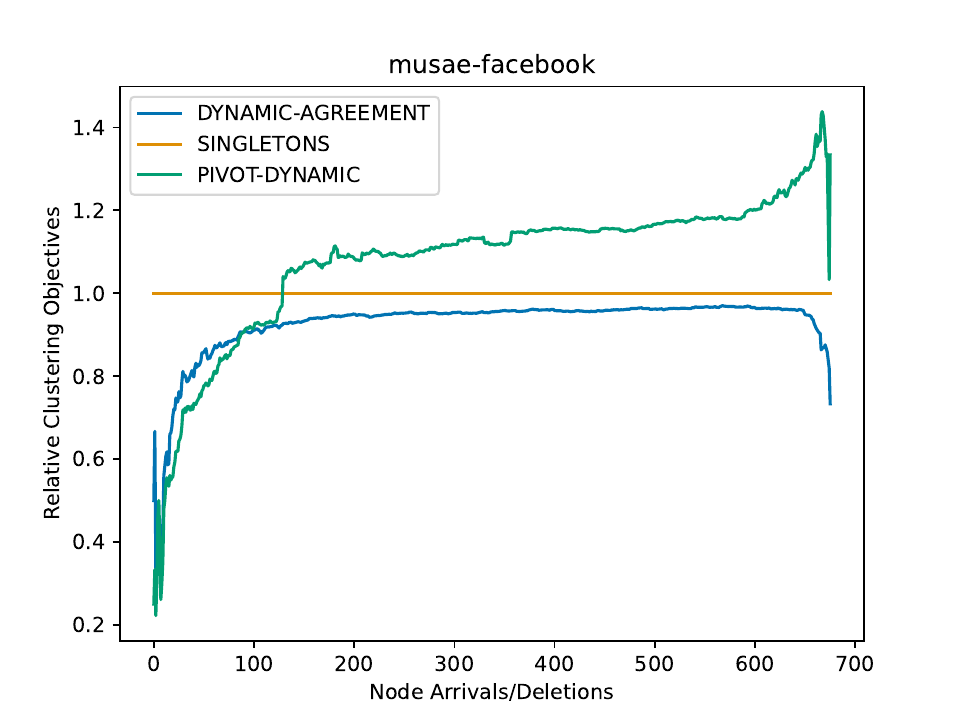} 
  \caption{Correlation clustering objective relative to singletons}
  \label{fig: musae-facebook}
\end{figure}

\begin{figure}[ht]
  \centering
  \includegraphics[width=0.55\textwidth]{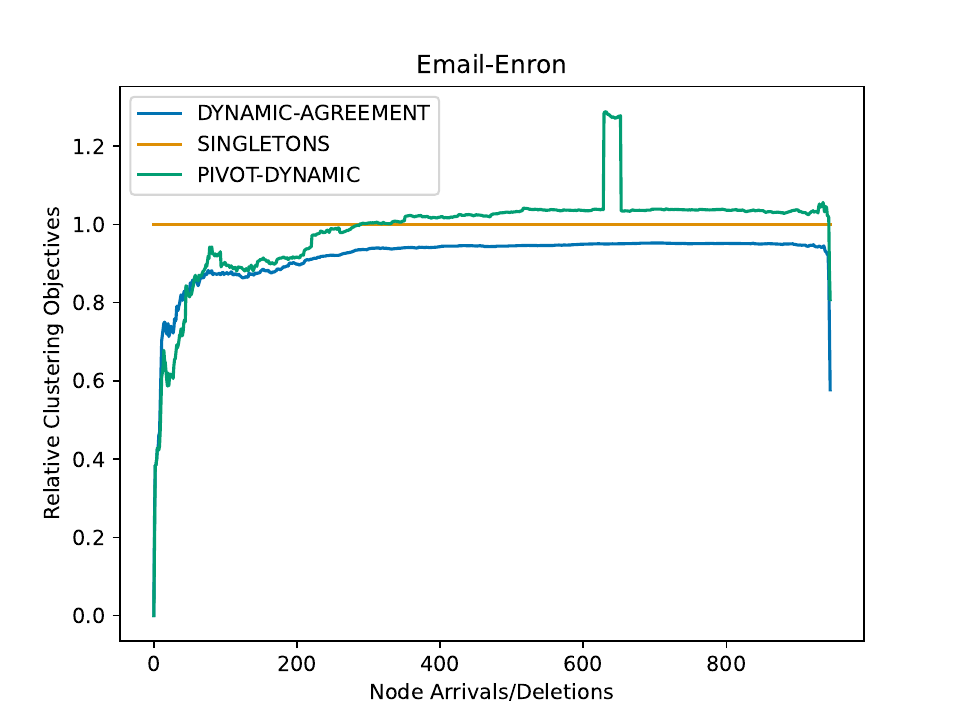} 
  \caption{Correlation clustering objective relative to singletons}
  \label{fig: Email-Enron}
\end{figure}

\end{document}